\documentclass{article} % For LaTeX2e
\usepackage{iclr2027_conference,times}

\usepackage{amsmath,amsfonts,bm}

\def\eqref#1{equation~\ref{#1}}
\def\1{\bm{1}}

\DeclareMathAlphabet{\mathsfit}{\encodingdefault}{\sfdefault}{m}{sl}
\SetMathAlphabet{\mathsfit}{bold}{\encodingdefault}{\sfdefault}{bx}{n}

\newcommand{\E}{\mathbb{E}}

\DeclareMathOperator*{\argmax}{arg\,max}
\DeclareMathOperator*{\argmin}{arg\,min}

\usepackage{url}
\usepackage{amsthm}
\usepackage{thmtools}
\usepackage{thm-restate}
\usepackage{mathtools}
\usepackage{hyperref}
\usepackage{cleveref} % You need to load this after theorem packages, otherwise the references break.
\usepackage{subcaption}

\declaretheorem[name=Assumption]{assumption}

\declaretheorem[name=Fact]{fact}
    
\title{Regularized policy gradient\\ with learned mixtures of Gaussians\\ for games with continuous actions}

\arxivpreprint

\author{%
  Ondřej Kubíček  \\
  Czech Technical University in Prague\\
  Carnegie Mellon University\\ 
  \texttt{kubicon3@fel.cvut.cz} \\
  \And
  Viliam Lisý \\
  Czech Technical University in Prague \\
  Artificial Intelligence Center \\
  \texttt{viliam.lisy@fel.cvut.cz} \\
  \AND
  Tuomas Sandholm \\
  Carnegie Mellon University \\
  Strategy Robot, Inc. \\
  Strategic Machine, Inc.\\
  Optimized Markets, Inc.\\
  \texttt{sandholm@cs.cmu.edu } \\ 
}

\crefname{appsec}{Appendix}{Appendices}

\begin{document}

\newcommand{\RealNumbers}[0]{\mathbb{R}}
\newcommand{\BorelSet}[0]{\mathcal{P}}
\newcommand{\Expectation}[0]{\E}

\newcommand{\Player}[0]{i}

\newcommand{\State}[0]{s}
\newcommand{\Actions}[1]{\mathcal{A}_{#1}}
\newcommand{\Action}[1]{a_{#1}}

\newcommand{\Strategy}[1]{\pi_{#1}}
\newcommand{\Strategies}[1]{\Pi_{#1}}

\newcommand{\Utility}[0]{u}
\newcommand{\ExpectedUtility}[1]{U_{#1}}
\newcommand{\Landscape}[1]{Q_{#1}}

\newcommand{\BR}[0]{BR}
\newcommand{\Nash}[0]{*}

\newcommand{\CategoricalStrategy}[0]{w}
\newcommand{\GaussianStrategy}[0]{a}

\newcommand{\Timestep}[0]{t}

\newcommand{\AtomStrategy}[0]{k}
\newcommand{\AtomStrategies}[0]{K}

\newcommand{\ActionValue}[1]{q_{#1}}

\newcommand{\RegularizationStrength}[0]{\eta}
\newcommand{\KLDiv}[0]{D_{\text{KL}}}

\newcommand{\LearningRate}[0]{\alpha}

\newcommand{\Mean}[0]{\mu}
\newcommand{\Variance}[0]{\sigma}
\newcommand{\Covariance}[0]{\Sigma}
\newcommand{\Gaussian}[0]{\mathcal{N}}

\newcommand{\LipConst}[0]{L}
\newcommand{\MonConst}[0]{\delta}

\newcommand{\GameField}[0]{V}
\newcommand{\RegularizerField}[0]{R}
\newcommand{\Jacobian}[0]{J}

\newcommand{\IdentityM}[0]{\textbf{I}}

\newcommand{\VI}[0]{\text{VI}}
\newcommand{\PGDA}[0]{T}
\newcommand{\ProjStrategies}[0]{P_{\Strategies{}}}
\newcommand{\Inner}[0]{\Psi}

\newcommand{\Exploitability}[0]{\mathcal{E}}

\newcommand{\Neural}[0]{\theta}

\newcommand{\Loss}[0]{\mathcal{L}}
\newcommand{\Advantage}[0]{A}
\newcommand{\EntropyCoefficient}[0]{\alpha}
\newcommand{\Entropy}[0]{\mathcal{H}}

\newcommand{\PPOSurrogate}[0]{\rho}

\newcommand{\AlgorithmName}[0]{Magnetic Mixture Policy Optimization}
\newcommand{\AlgorithmShort}[0]{MMPO}

\maketitle
\begin{abstract}
Most successes of superhuman game-playing algorithms are in games with discrete actions, yet in auctions, robotics, sports, or trading, actions are nearly continuous. Prior techniques either rely on expert-designed discretizations or are sample inefficient. We present a scalable policy-gradient algorithm for large sequential games with continuous or mixed discrete and continuous actions. It combines \textit{magnetic mirror descent} with a mixture of Gaussians reparametrization, trained via self-play. We show that it approximates equilibrium in games where gradient descent fails. In sequential games, it outperforms \textit{neural fictitious self-play} and matches or outperforms the final strategies of \textit{policy space response oracles} with 3.5--5.5$\times$ fewer samples. In heads-up no-limit Texas hold'em, it performs on par with Slumbot.
    % Prior techniques are either sample inefficient if they try to approximate equilibrium, or they only use the single-player algorithms and their strategies are prone to exploitation.
    % Prior scalable techniques for these continuous action spaces are sample inefficient as they either need to train a vast number of best response strategies, or they require the network to learn the whole distribution from just self-play.  % How to explain the inefficiency of randomized policy networks somehow elegantly?
    % is a combination and extension of two prior works, Magnetic Mirror Descent, which uses KL-divergence regularization to ensure strong monotonicity of the gradient updates to converge to Quantal Response Equilibrium in two-player zero-sum games and Soft Actor critic, which reparameterizes the continuous action policy by Gaussian mean and standard deviation, which are trained through policy-gradient updates. Our approach trains a mixture of Gaussians, where Projected Gradient Descent is used to train the means and standard deviations and Magnetic Mirror Descent is used to train the categorical distribution over these Gaussians. We show that including the magnet regularization from the Magnetic Mirror Descent into the Gaussian space also enables convergence in games, where the standard Gradient Descent fails. We then demonstrate our algorithm on several domains with varying sizes including Heads-up No-limit Texas Hold'em.
\end{abstract}
\section{Introduction}
Games can model many situations where multiple agents are trying to achieve possibly different goals over a long horizon. Many successes in this domain focus on games with discrete action spaces, like Chess, Go, or Stratego \citep{silver2018alphazero,schrittwieser2020muzero,perolat2022stratego,sokota2025stratego,zhang2026general}. However, some domains, like robotics or economics, use continuous or nearly-continuous actions \citep{todorov2012mujoco,haarnoja2018sac,vickrey1961auctions,zheng2022economics}. 

In single-player and cooperative multi-player reinforcement learning, many studies have demonstrated the effectiveness of policy-gradient algorithms, even for continuous action problems \citep{lillicrap2016ddpg,haarnoja2018sac,yu2022mappo}. However, these works often rely on the fact that, in those environments, an optimal pure (i.e., deterministic) strategy exists. It is well known that adversarial games with hidden information or simultaneous moves lack pure optimal strategies and instead require mixing among several actions to prevent the player from becoming predictable \citep{von1947theory,blackwell1979theory,roberson2006blotto}. Moreover, mixing can be beneficial even in fully observable physical systems. Consider a penalty kick in football or a tennis serve; players need to commit to their decision before observing what their opponent does. If players commit to the same movement all the time, they will become highly exploitable. 

Only recent developments in regularized policy-gradient algorithms have demonstrated how to converge to mixed strategies using scalable neural policy representations, without resorting to impractical and costly policy averaging \citep{perolat2021rnad,sokota2022mmd}. These algorithms then achieved superhuman performance in very large games, such as Stratego \citep{perolat2022stratego,sokota2025stratego} and Generals.io \citep{straka2026generals}, relying on their discrete action spaces.

The common approach to continuous action spaces is to discretize and then solve the resulting discretized game \citep{kroer2015discretization,tavakoli2018discrete,tang2020discretizing}. However, this disregards the relationships among discrete actions, and the algorithm must learn them from gameplay. Moreover, the solution quality depends on the chosen discretization, which often requires expert knowledge \citep{kroer2015discretization}. Alternatively, some approaches are directly applicable to continuous action spaces, but they either cannot sufficiently model mixed equilibria \citep{schulman2017ppo}, are not easily extended to large sequential games \citep{hsieh2019gan,enrich2020gan,martin2023rpn,martin2025jpspg}, or are sample inefficient because they train an ensemble of strategies \citep{heinrich2016nfsp,lanctot2017psro}.

% Policy-gradient algorithms have shown incredible performance in very large games, like a superhuman performance in Stratego \citep{perolat2022stratego,sokota2025stratego} or Generals.io \citep{straka2026generals}. However, in adversarial games, policy-gradient algorithms have so far been limited to discrete action settings. On the other hand, in single-player and cooperative multi-player reinforcement learning, policy-gradient algorithms are widely used, even in continuous action settings \citep{lillicrap2016ddpg,haarnoja2018sac}. This gap in usage is mainly because optimal strategies in adversarial games often require mixing among several actions. This is well known in games with imperfect information or simultaneous moves, but mixing can be beneficial even in fully observable physical systems. Consider a penalty kick in football or a tennis serve: for some actions, players need to commit to their decision before observing what their opponent does. If players commit to the same movement all the time, they will become highly exploitable. Similarly, in cybersecurity, predictable security patterns may be exploited by an adversary. Only recent developments in regularized policy-gradient algorithms have shown how to converge to these mixed strategies without costly averaging \citep{perolat2021rnad,sokota2022mmd}.

In this work, we introduce \AlgorithmName{}, an algorithm for approximating Nash equilibria in large sequential zero-sum games with continuous actions. It parametrizes the strategy for each decision by a mixture of Gaussians. The training jointly learns the means and standard deviations of several Gaussians, as well as the distribution over them. Recent advancements in regularized policy gradients for games were necessary to enable this method, because the mixture strategy over the fixed Gaussian distributions is itself a matrix game, where standard policy-gradient algorithms may not converge. We show that the regularization can be extended even to training the Gaussian parameters, and that it converges in games, where standard gradient descent does not.

We evaluate our method across several games of varying size and type. Our method is applicable to games that mix discrete and continuous actions, as we show in several Poker games, where we treat bets as continuous actions. We also demonstrate that this method is scalable to large games by applying it in full-size heads-up no-limit Texas hold'em. We show that, using our method in self-play on a single GPU, we trained a strategy that is on par with latest Slumbot, an improved version of the winner of the 2018 annual computer poker competition.
\section{Background}
In this work we consider two-player zero-sum games represented by a triplet $(\Actions{1}, \Actions{2}, \Utility)$, $\Actions{\Player} \subseteq \RealNumbers^{N_{\Player}}$ is a compact non-empty action space of the player $\Player \in \{1, 2\}$ and $\Utility: \Actions{1} \times \Actions{2} \to \RealNumbers$ is a continuous payoff function of player $\Player$, that is bounded and Borel measurable. We will use $\Utility_{1} = \Utility$ and $\Utility_{2} = -\Utility$ In the game both players choose simultaneously $\Action{\Player} \in \Actions{\Player}$. Player 1 is trying to select an action such that it maximizes the payoff, whereas Player 2 is minimizing it. 

$\BorelSet(\Actions{\Player})$ is a Borel probability measure in $\Action{\Player}$, which we call a mixed strategy. The expected payoff is then

$$\ExpectedUtility{\Player}(\Strategy{1}, \Strategy{2}) = \E_{\Strategy{1} \otimes \Strategy{2}} [\Utility_{\Player}] = \int_{\Actions{1}} \int_{\Actions{2}} \Utility_{\Player}(a, b) \,d\Strategy{2}(a) \, d\Strategy{1}(b), \;\;\;\; (\Strategy{1}, \Strategy{2}) \in \BorelSet(\Actions{1}) \times \BorelSet(\Actions{2})$$

% [Order of integration can be swapped because of Fubini-Tonelli theorem]

% Let us define effective landscape $\Landscape{\Strategy{\Player}}$ for player $\Player$ as 

% $$\Landscape{\Strategy{-\Player}} (\Action{\Player}) = \int_{\Actions{-\Player}} \Utility_{\Player}(\Action{\Player}, \Action{-\Player}) \,d\Strategy{-\Player}(\Action{-\Player}) $$

Best response to a strategy $\Strategy{-\Player}$ is defined as $\Strategy{\Player}^{BR} = \argmax_{\Strategy{\Player} \in \BorelSet(\Actions{\Player})} \ExpectedUtility{\Player}(\Strategy{\Player}, \Strategy{-\Player}) $. We use the traditional game-theory notation of $-\Player$ denotes the other player than $\Player$. Since $\ExpectedUtility{\Player}$ is affine, there exists a pure strategy that is a best response. Nash equilibrium is a strategy pair $\Strategy{}^\Nash = (\Strategy{1}^\Nash, \Strategy{2}^\Nash)$, such that it is a saddle point of $\ExpectedUtility{1}$, where neither can improve by unilaterally deviating from its strategy

% $$\sup_{\Strategy{\Player} \in \BorelSet(\Actions{\Player})} \ExpectedUtility{\Player}(\Strategy{\Player}, \Strategy{-\Player}) = \sup_{\Action{\Player} \in \BorelSet(\Actions{\Player})} \Landscape{\Strategy{-\Player}}(\Actions{\Player})$$

$$\ExpectedUtility{1}(\Strategy{1}, \Strategy{2}^\Nash) \leq \ExpectedUtility{1}(\Strategy{1}^\Nash, \Strategy{2}^\Nash) \leq \ExpectedUtility{1}(\Strategy{1}^\Nash, \Strategy{2})\;\;\;\;\; \forall \Strategy{1} \in \BorelSet(\Actions{1}), \Strategy{2} \in \BorelSet(\Actions{2})$$

We also use NashConv $\Exploitability$, sometimes called exploitability, to compare different baselines, and it serves as a distance from Nash equilibrium
$$\Exploitability(\Strategy{}) = \max_{\Strategy{1}'}\ExpectedUtility{1}(\Strategy{1}', \Strategy{2}) - \min_{\Strategy{2}'}\ExpectedUtility{1}(\Strategy{1}, \Strategy{2}') $$

\section{\AlgorithmName{}}
Gradient descent-style algorithms in continuous action games often serve to find a pure (deterministic) strategy. As such, most of those algorithms are limited to settings with pure Nash equilibria. A common approach to finding these mixed equilibria is to discretize the action spaces of both players and then compute a mixed strategy in the resulting matrix game. However, this approach will always have an unavoidable worst-case discretization error. Moreover, in an online setting with function approximators, the finer abstraction increases both the number of required samples and the overall model size.

In this work, we focus on a discretization that is optimized during training. Some prior work on \textit{generative adversarial networks} has already used discretization, with $\AtomStrategies$ pure actions in the action space \citep{hsieh2019gan,enrich2020gan}. Then they optimized the value of those pure actions using gradient descent on the exact payoff gradient and also the distribution over them. However, these techniques were not extended into sequential settings. We propose using a Gaussian distribution instead of pure actions, allowing us to easily estimate the gradient and apply entropy regularization. Moreover, in \cref{sec:gaussian}, we show that when we add secondary KL-divergence loss term, it enables convergence even in some games where standard gradient descent fails.
\subsection{Categorical head}
First, let us assume that oracle provides $\AtomStrategies$ actions $\GaussianStrategy^\AtomStrategy$, and the task is only to find a distribution over those. We do not impose constraints on the parametrization of the actions, except that each pair of actions has a uniquely defined expected value. This then results in a matrix game, where each player has $\AtomStrategies$ actions, and the utility given to the players is the expected value when the action is played.

The standard gradient descent methods are known to cycle around the equilibrium point and, as such, may not converge in the last iterate. Several techniques have been developed that converge to the Nash equilibrium in the last iterate \citep{korpelevich1976extragradient,popov1980optimistic,daskalakis2018optimistic}. In this work, we focus on KL-regularization-based methods, such as \textit{regularized Nash dynamics (RNaD)} and \textit{magnetic mirror descent (MMD)}, as they relate to our use of KL-regularization in Gaussian reparametrization \citep{perolat2021rnad,sokota2022mmd}.

Consider a strategy $\Strategy{\Player}$ of player $\Player$, which is a categorical distribution over the pure actions $\Action{\Player}^\AtomStrategy$. With the magnet regularization, the goal of each player is to find a solution to the following optimization problem
\begin{equation}
\label{eq:mmd_update}
\Strategy{\Player}' = \argmax_{\Strategy{\Player}} \sum_{\AtomStrategy = 1}^\AtomStrategies \Strategy{\Player}(\Action{\Player}^\AtomStrategy) \ActionValue{\Player}(\Action{\Player}^\AtomStrategy) - \RegularizationStrength \KLDiv(\Strategy{\Player}, \Strategy{\Player}^{M})% - \LearningRate \KLDiv(\Strategy{\Player}, \Strategy{\Player}^{\Timestep})
\end{equation}
$\ActionValue{\Player}(\Action{\Player}^\AtomStrategy)$ is a value player $\Player$ receives after playing action $\Action{\Player}^\AtomStrategy$, and $\Strategy{\Player}^M$ is a magnet strategy, which can be arbitrary in the game. It has been proven by prior work that this regularized objective has a unique fixed point for fixed magnet regularization $\RegularizationStrength$ and strategy $\Strategy{\Player}^M$, which is a quantal response equilibrium \citep{sokota2022mmd}. It has also been shown that periodically updating the $\Strategy{\Player}^M$ or annealing $\RegularizationStrength$ results in convergence to the Nash equilibrium \citep{perolat2021rnad}.

We are not making any improvements to these algorithms for discrete action games; we merely use them as an effective algorithm to approximate the Nash equilibrium in the matrix game. Other modifications of gradient descent that converge to the Nash equilibrium, such as optimism or extragradient \citep{korpelevich1976extragradient,popov1980optimistic}, are also viable. We used the magnet regularization for its compatibility with function approximators in large sequential games, which has been demonstrated in Stratego and Generals.io \citep{perolat2022stratego,sokota2025stratego,straka2026generals}. 

\subsection{Gaussian head}
\label{sec:gaussian}
In single-player settings, it is common to find strategies in continuous action spaces by parametrizing the strategy as a Gaussian distribution defined by its mean $
\Mean$ and a log of standard deviation $\log(\Variance)$ \citep{haarnoja2018sac}. During training, it samples from this parametrized Gaussian to estimate the gradient. In addition to gradient estimation, this reparametrization enables entropy regularization, which has been shown to improve exploration in discrete-action settings. Gradient descent is known not to converge in some simple continuous action games, such as $f(x, y) = xy$ \citep{martin2023rpn}. We show in experiments that this prevails even when the strategy is parametrized by a Gaussian. In this section, we focus only on strategies parametrized by a single Gaussian and show that adding the magnet KL regularization improves the convergence properties of gradient descent.

Consider a strategy $\Strategy{\Player}$ parametrized by the Gaussian mean $\Mean_{\Player}$ and standard deviation $\Variance_\Player$. Throughout the paper, when the mean $\Mean{}$, standard deviation $\Variance$, and strategy $\Strategy{}$ share the same subscripts and superscripts, they represent the same strategy. The expected utility of the strategy is then 
$$\ExpectedUtility{\Player}^G(\Strategy{}) = \E_{\Action{1} \sim \Gaussian(\Mean_{1}, \Variance_{1}^2), \Action{2} \sim \Gaussian(\Mean_{2}, \Variance_{2}^2)} \Utility_{\Player}(\Action{1}, \Action{2})$$
As with discrete games, gradient descent does not converge because the gradient field is not strongly monotonic. Adding a sufficiently strong regularizer to the field then makes the game monotonic, and therefore the gradient descent will converge to a fixed point. 

Consider a magnet strategy $\Strategy{}^M$. We will use $\Strategy{} := (\Strategy{1}, \Strategy{2})$. The gradient dynamics on this reparametrized strategy with the KL regularization are
$$\GameField^\RegularizerField(\Strategy{}) = \GameField(\Strategy{}) + \RegularizationStrength \RegularizerField(\Strategy{})  = \begin{pmatrix}
    -\nabla_{\Strategy{1}} \ExpectedUtility{1}^G(\Strategy{})\\
    \nabla_{\Strategy{2}} \ExpectedUtility{1}^G(\Strategy{}) \\
\end{pmatrix} + \RegularizationStrength \begin{pmatrix}
    \nabla_{\Strategy{1}} \KLDiv(\Strategy{1}, \Strategy{1}^M)\\
    \nabla_{\Strategy{2}} \KLDiv(\Strategy{2}, \Strategy{2}^M)\\
\end{pmatrix} 
$$

\begin{assumption}
\label{as:lip}
     $\GameField(\Strategy{})$ is $\LipConst_V$-Lipschitz continuous on the Gaussian strategy $\Strategy{}$,
    $$\|\GameField(\Strategy{}) - \GameField(\Strategy{}')\| \leq \LipConst_V \|\Strategy{} - \Strategy{}'\| \quad \quad \forall \Strategy{}, \Strategy{}' $$
\end{assumption}

\begin{restatable}[]{theorem}{MmdConvergence}
\label{thm:mmd_convergence}
     Let \cref{as:lip} hold and  $\RegularizationStrength > \frac{\LipConst_V}{\MonConst_\RegularizerField}$. If $\LearningRate < \frac{2\MonConst_M}{\LipConst_M^2}$ then the Projected Gradient Descent Ascent (PGDA) converges to unique solution $\Strategy{}^*$ of variational inequalities $\VI(\GameField^\RegularizerField, \Strategies{})$ and for each strategy $\Strategy{}^k$ in iteration $k$ following holds
    \begin{equation}
        \|\Strategy{}^k - \Strategy{}^*\|^2 \leq r(\RegularizationStrength)^k\|\Strategy{}^0     - \Strategy{}^*\|^2 \quad \quad r(\RegularizationStrength) = 1 - 2 \LearningRate \MonConst_M + \LearningRate^2 \LipConst_M^2 
    \end{equation}
\end{restatable}
Proof to \Cref{thm:mmd_convergence} follows a similar structure to that of MMD \citep{sokota2022mmd}. Similarly to that result, the solution of the regularized problem is unique, but this solution is not the Nash equilibrium of the original game. Contrary to the MMD, where the solution is a Quantal Response Equilibrium, the Gaussian reparametrization does not have this guarantee, because the space of feasible strategies is limited by hyperparameter $\Variance_{\min}$, which we use because standard deviation has to be $\Variance > 0$, such that the KL-divergence is defined.. 

Even if the fixed point is not a Nash equilibrium, several approaches can ensure convergence, such as annealing the regularization or updating the magnet. In the remainder of this work, we use the magnet replacement strategy, which was first applied in \textit{Regularized Nash Dynamics} \citep{perolat2021rnad, perolat2022stratego}, where the previous fixed point replaces the magnet. If we assume only a convex-concave games, which have a pure Nash equilibrium \citep{sion1958minimax}, then the sequence of the fixed point converges to an $\epsilon$-equilibrium, where $\epsilon$ is bounded by the $\Variance_{\min}$, which is a tunable hyperparameter. 
\begin{restatable}[]{theorem}{FinalConvergence}
\label{thm:final_convergence}
    In a convex-concave game with payoff $\Utility$ that has continuous second derivative. Assume only strategies $\Strategies{}^{\Variance_{\min}} \subseteq \Strategies{}$, which use $\Variance = \Variance_{\min}$ and let $\Variance_{\min} > 0$, $\RegularizationStrength > \LipConst_\GameField \Variance_{\min}^2$ and $M_{\Player} = \sup |\frac{\partial^2 \Utility}{\partial \Action{\Player}^2} |$. Let $\Strategy{}^{M, t+1} = \Inner(\Strategy{}^{M, t})$ be the exact solution for the fixed magnet at iteration $t$ from arbitrary $\Strategy{}^{M, 0} \in \Strategies{}^{\Variance_{\min}}$. Then the sequence of $\Strategy{}^{M, t}$ converges to a fixed solution $\Strategy{}^*$ with exploitability at most
    \begin{equation}
        \Exploitability(\Strategy{}^*) \leq \frac{M_1 + M_2}{2} \Variance_{\min}^2
    \end{equation}
\end{restatable}
The proofs of \Cref{thm:mmd_convergence,thm:final_convergence} are in \Cref{app:proof}. Proof of \Cref{thm:final_convergence} first shows that the fixed point is the magnet itself only when the magnet is a Nash equilibrium. Then, it shows that the fixed point for a specific magnet is closer to the Nash equilibrium than the magnet itself, in terms of geometric distance. However, this does not say that each strategy in the sequence is less exploitable than the previous one. But, it says that the sequence will eventually converge.
\subsection{Mixture of Gaussians}
The prior section focused on a single Gaussian, but in simultaneous-move or imperfect-information games, the equilibria are often not pure, requiring randomization among several actions. Consider a game where equilibrium mixes between two real values, -1 and 1, uniformly. Using a single Gaussian can never approximate such a strategy. In this section, we propose an algorithm, \AlgorithmName{} (\AlgorithmShort{}), that parametrizes the strategy as a Gaussian mixture. It uses the KL divergence between the current and magnet strategies as a secondary loss term. Since the KL divergence between a mixture of Gaussians does not have a closed-form solution, we use an upper bound, which is a sum of the KL divergence between the categorical distribution and the weighted sum of all the Gaussian KL divergences. The final loss is then combined by two PPO-like losses with magnet regularization, one for the categorical distribution and the second for the Gaussian distribution \citep{schulman2017ppo,sokota2022mmd}.

MMPO parametrizes the strategy space of player $\Player$ by $K$ Gaussian distributions represented by $\Mean_{\Player, k}, \Variance_{\Player, k}$ and a categorical distribution over them $\CategoricalStrategy_{\Player, k}$
% \begin{equation}
%     \Strategy{\Player} = \{\Mean_1, \dots \Mean_K\} \cup \{\Variance_1, \dots \Variance_K\} \cup \{\CategoricalStrategy_1, \dots \CategoricalStrategy_K\} 
% \end{equation}
\begin{equation}
    \Strategy{\Player} = \sum_{k = 1}^{K} \CategoricalStrategy_{\Player, k}\cdot \Gaussian(\Mean_{\Player, k}, \Variance_{\Player, k}^2)
\end{equation}
\begin{equation}
    \sum_{k = 1}^K \CategoricalStrategy_{\Player, k} = 1 \quad \quad \CategoricalStrategy_{\Player, k} \geq 0, \Variance_{\Player, k} \geq \Variance_{\min} \quad  \forall k \in \{1,\dots,K\} 
\end{equation}
The expected utility of this strategy profile is then
$$\ExpectedUtility{\Player}^G(\Strategy{}) = \sum_{k = 1}^K  \sum_{l = 1}^K \CategoricalStrategy_{1, k} \CategoricalStrategy_{2, l}  \E_{\Action{1} \sim \Gaussian(\Mean_{1, k}, \Variance_{1, k}^2), \Action{2} \sim \Gaussian(\Mean_{2, l}, \Variance_{2, l}^2)} \Utility_{\Player}(\Action{1}, \Action{2})$$
We store the strategy $\Strategy{\Neural}$ implicitly in the weights of the neural network $\Neural$. Consider that the algorithm samples action $\Action{}^t$, from a $(\CategoricalStrategy^t, \Mean^t, \Variance^t)$ in state $\State^t$. We will overload the notation $\Gaussian^t := \Gaussian(\Mean^t, \Variance^{2, t})$. The loss is then combined from the loss for the categorical head $\Loss_\CategoricalStrategy$ and for the Gaussian components $\Loss_\Gaussian$  
\begin{equation}
    \begin{split}
        \PPOSurrogate_w &= \min\Big(\frac{\Strategy{\Neural}(\State^t, \CategoricalStrategy^t)}{\Strategy{\text{old}}(\State^t, \CategoricalStrategy^t)}\Advantage (\State^t, \Action{}^t), \text{clip}\big(\frac{\Strategy{\Neural}(\State^t, \CategoricalStrategy^t)}{\Strategy{\text{old}}(\State^t, \CategoricalStrategy^t)}, 1 - \epsilon_w, 1+\epsilon_w \big)\Advantage (\State^t, \Action{}^t) \Big)\\
        \PPOSurrogate_\Gaussian &= \min\Big( \frac{\Strategy{\Neural}(\Action{}^t | \State^t, \Gaussian^t)}{\Strategy{\text{old}}(\Action{}^t | \State^t, \Gaussian^t)}\Advantage (\State^t, \Action{}^t), \text{clip} \big( \frac{\Strategy{\Neural}(\Action{}^t | \State^t, \Gaussian^t)}{\Strategy{\text{old}}(\Action{}^t | \State^t, \Gaussian^t)}, 1 - \epsilon_\Gaussian, 1+\epsilon_\Gaussian \big)\Advantage (\State^t, \Action{}^t) \Big)\\
        \Loss_w &= -\E_t \PPOSurrogate_w + \RegularizationStrength \KLDiv(\Strategy{\Neural}(\State^t, \cdot), \Strategy{}^M(\State^t, \cdot)) - \EntropyCoefficient_{w}\Entropy(\Strategy{\Neural}(\State^t, \cdot))\\
        \Loss_\Gaussian &= -\E_t\PPOSurrogate_\Gaussian+ \RegularizationStrength \KLDiv( \Gaussian^t, \Gaussian^M) - \EntropyCoefficient_{\Gaussian}\Entropy(\Gaussian^t) \\
        \Loss_{\Strategy{}} &= \Loss_\CategoricalStrategy + \Loss_\Gaussian
    \end{split}
\end{equation}
both $\Loss_w$ and $\Loss_\Gaussian$ include the same terms. PPO clipped loss, KL-divergence regularization to a magnet, and entropy loss. We also train a critic, with the generalized advantage estimate \citep{schulman2015gae}. The entire training is done in self-play, using only the game simulator.

Increasing the number of Gaussian components $\AtomStrategies$ broadens the set of strategies that can be modeled, but it also increases the number of samples required for training. Despite that, one advantage of this model is that it can learn to ignore certain components if they are not useful for the current decision. 

\AlgorithmShort{} is also applicable to games that combine discrete and continuous actions, such as some continuous Poker games that combine the fold and call actions with continuous betting. The only change to the \AlgorithmShort{} is that for the discrete actions, there will be no corresponding Gaussian $\Mean$ and~$\Variance$.  
\section{Experiments}
We have evaluated \AlgorithmShort{} across several game settings. We first evaluate the effect of the magnet regularization in isolation. Then we compare our algorithm to several baselines that approximate Nash equilibrium in large games with parametrized strategies, in both one-shot and sequential settings. Lastly, we demonstrate that the approach is applicable even in large games, as we use it in Heads-up no-limit Texas Hold 'em, 

\subsection{Convergence of the regularized Gaussian head}
\label{sec:exp_gaussian}
This experiment validates \Cref{thm:mmd_convergence} and explores MMPO's behavior in very simple games. We consider 3 games, with rules in \Cref{app:games}, in which there is a pure equilibrium, but gradient descent fails to converge to it. In \Cref{fig:single_ideal}, we show a comparison of runs across the three games, either with usual gradient descent or with regularized projected gradient descent. For the exploitability computation, we used $\sigma = 0$. With the additional regularization, the strategies indeed converge to the Nash equilibrium.

In \Cref{fig:single_ppo}, we show the same comparison, but a neural network now parameterizes the strategy, and the updates are estimated via on-policy sampling. Because of the additional noise from the training, the strategy does not converge to the exact Nash equilibrium, but the regularization improves performance by at least an order of magnitude. In the rotational game, the exploitability remains quite high, but that is due to the payoff function of the game, where the exploitability is high, even when strategies are geometrically very close to the Nash equilibrium. In \Cref{app:add_gradient}, we show the same graphs when gradient descent is applied to tabular representations of the strategy, with an update using on-policy sampling. The behavior is between the two presented here.

\begin{figure*}[!h]
    \centering    
    \begin{subfigure}[b]{0.32\textwidth}
        \centering
    \includegraphics[width=0.99\linewidth]{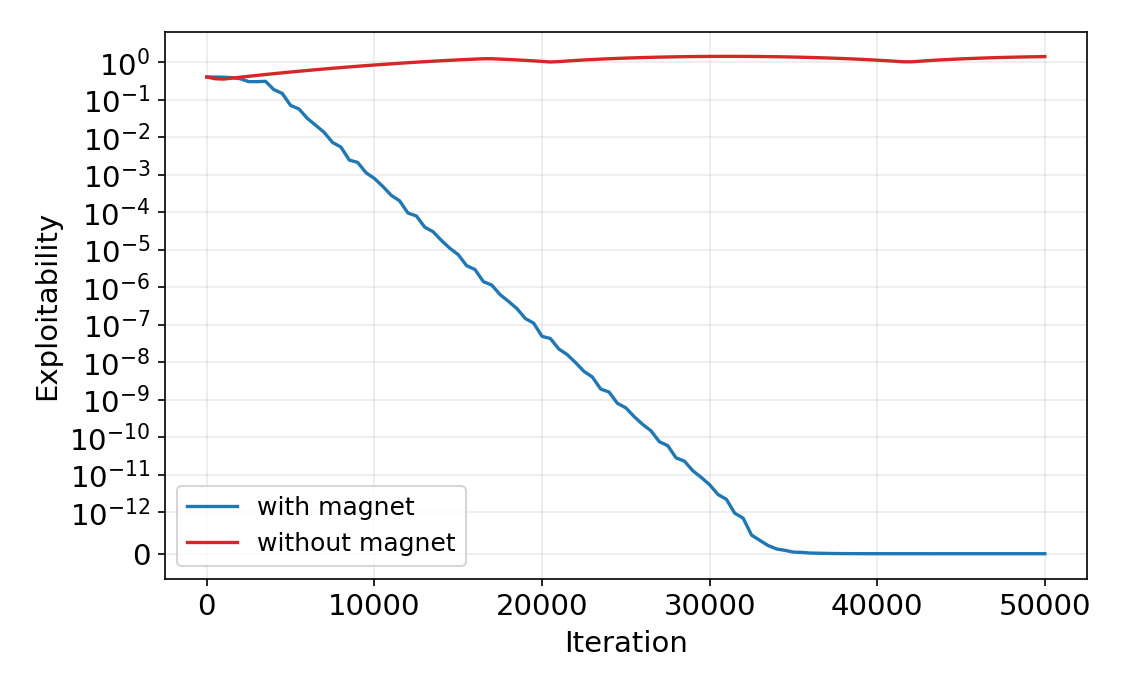}
        \caption{Continuous Matching Pennies}
        \label{fig:single_ideal_mp}
    \end{subfigure}
    \hfill % Adds horizontal space between subfigures
    \begin{subfigure}[b]{0.32\textwidth}
        \centering
    \includegraphics[width=0.99\linewidth]{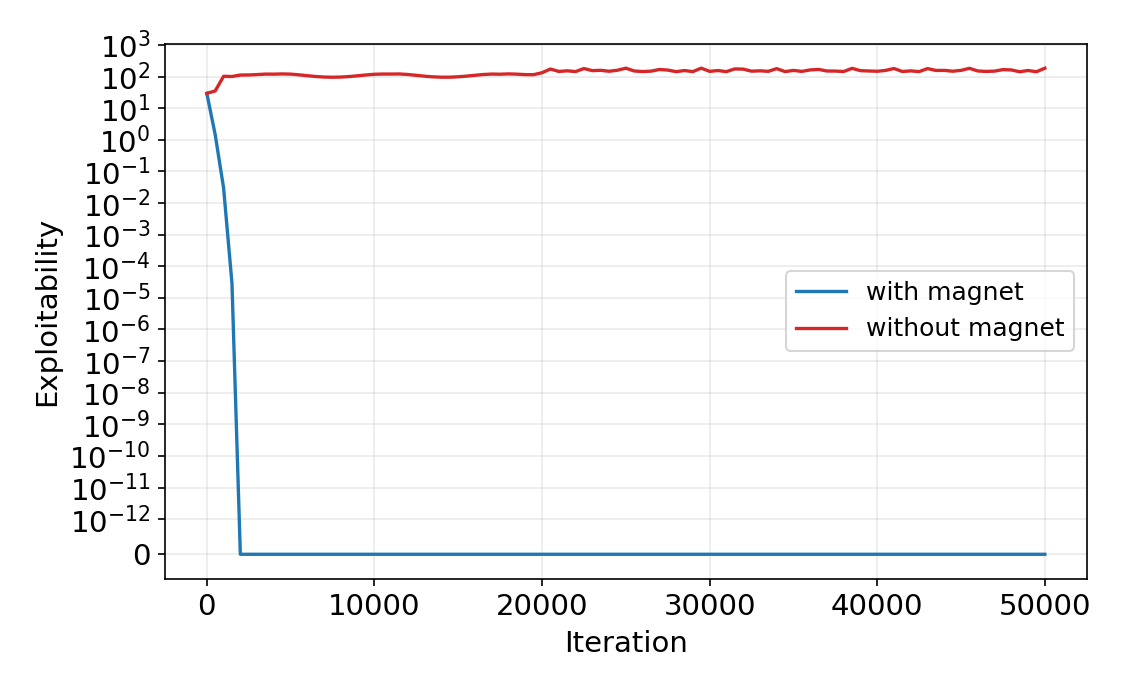}
        \caption{Rotational game, dim=2}
        \label{fig:single_ideal_rot2}
    \end{subfigure} 
    \hfill % Adds horizontal space between subfigures
    \begin{subfigure}[b]{0.32\textwidth}
        \centering
    \includegraphics[width=0.99\linewidth]{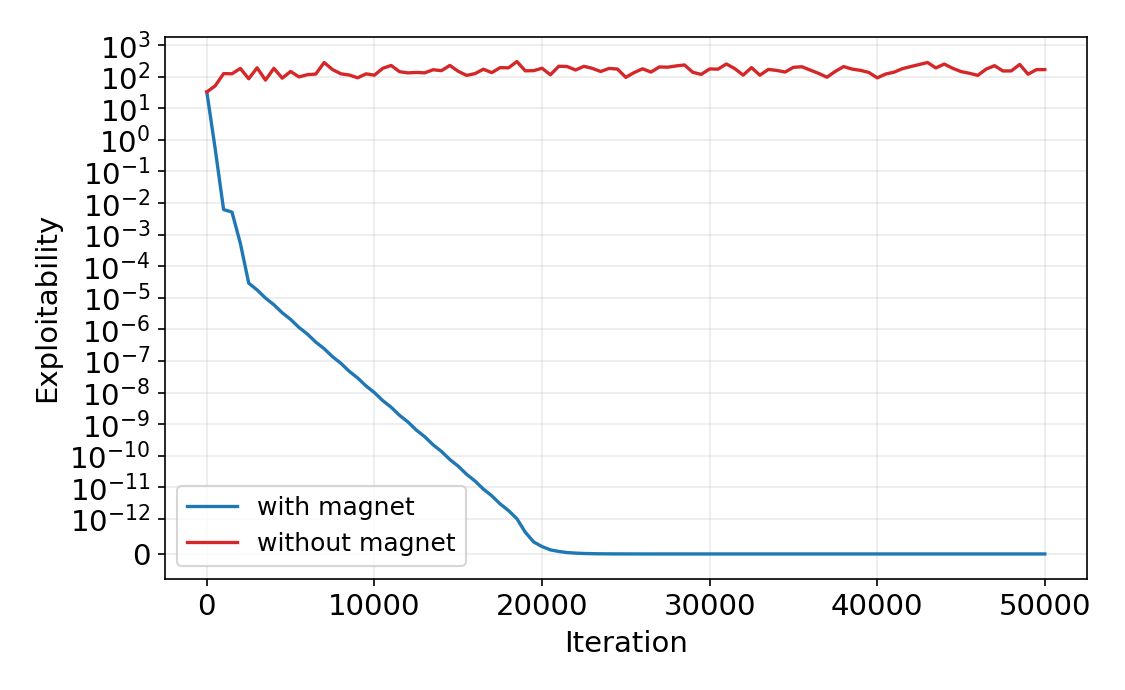}
        \caption{Rotational game, dim=3}
        \label{fig:single_ideal_rot3}
    \end{subfigure} 
    \caption{Exploitability in one-shot games with pure Nash equilibria using different gradient descent with or without magnet.}
    \label{fig:single_ideal} 
\end{figure*}  

\begin{figure*}[!h]
    \centering    
    \begin{subfigure}[b]{0.32\textwidth}
        \centering
    \includegraphics[width=0.99\linewidth]{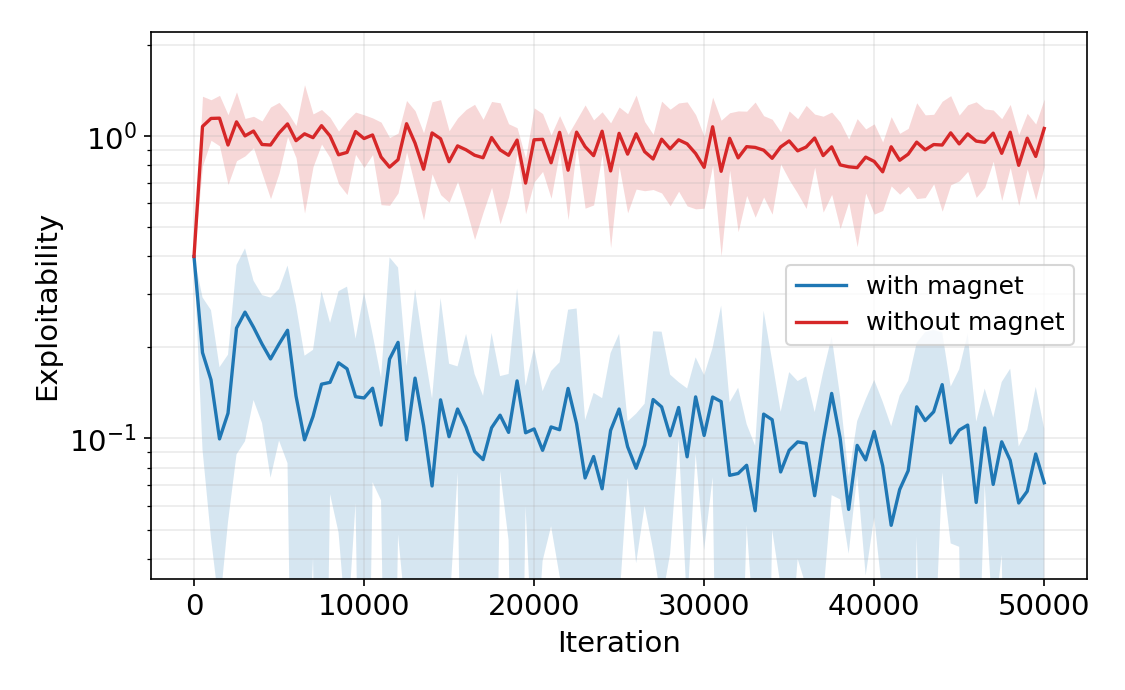}
        \caption{Continuous Matching Pennies}
        \label{fig:single_ppo_mp}
    \end{subfigure}
    \hfill % Adds horizontal space between subfigures
    \begin{subfigure}[b]{0.32\textwidth}
        \centering
    \includegraphics[width=0.99\linewidth]{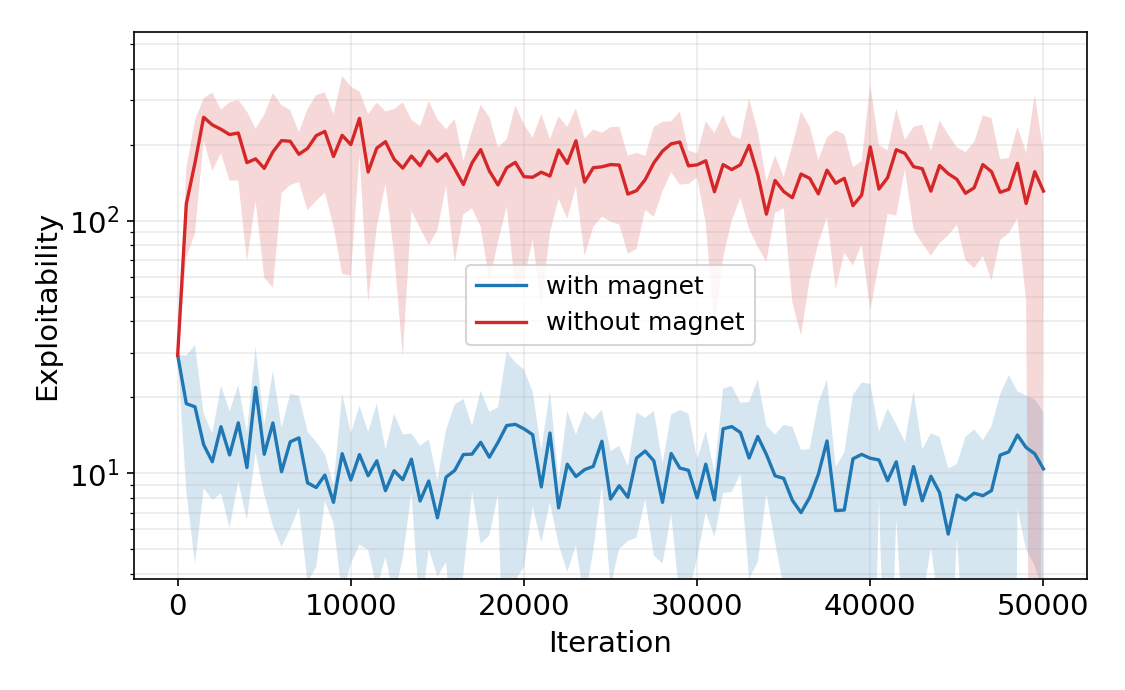}
        \caption{Rotational game, dim=2}
        \label{fig:single_ppo_rot2}
    \end{subfigure} 
    \hfill % Adds horizontal space between subfigures
    \begin{subfigure}[b]{0.32\textwidth}
        \centering
    \includegraphics[width=0.99\linewidth]{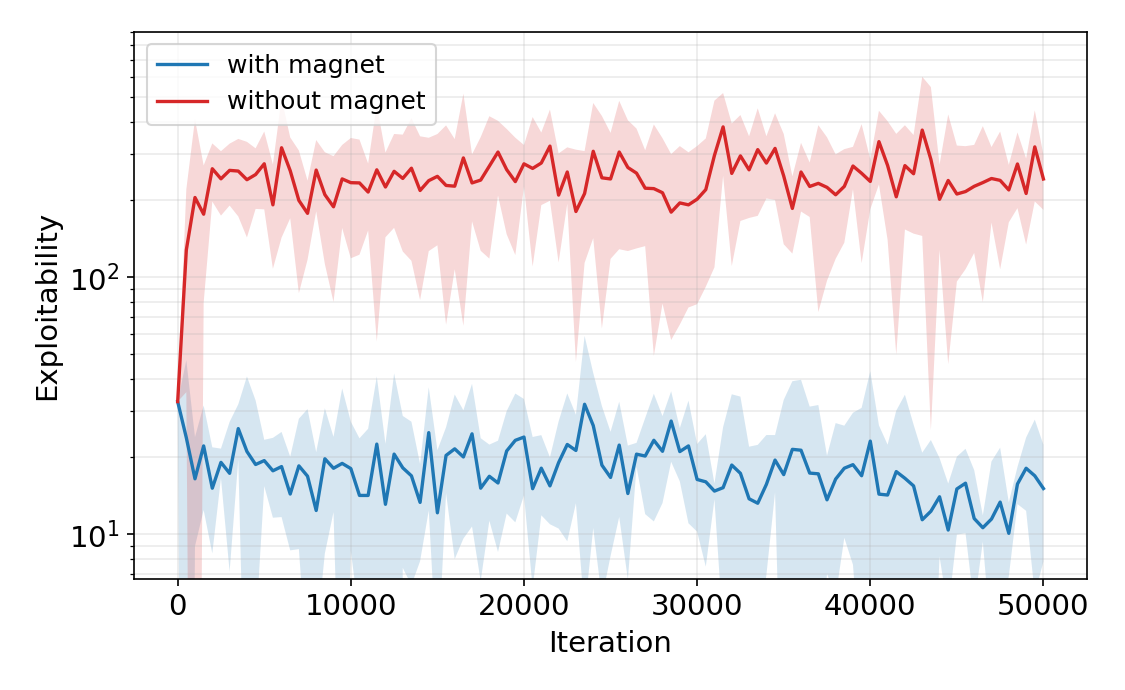}
        \caption{Rotational game, dim=3}
        \label{fig:single_ppo_rot3}
    \end{subfigure} 
    \caption{Exploitability in one-shot games with pure Nash equilibria, using PPO with and without magnet regularization}
    \label{fig:single_ppo} 
\end{figure*}  
\subsection{One-shot games}
\label{sec:exp_one_shot}
We compare our methods with several neural baselines on one-shot games. Some baselines we use in this work use vastly different training, so to keep the comparison fair, we compare them in number of environment interactions required. In \cref{app:walltime}, we show the same comparison using the wall-time required. We do not compare with tabular baselines because our work aims to enable efficient training in large games; as such, we do not claim that \AlgorithmShort{} should be used in games where tabular techniques are applicable. We believe that in these tractable games, the tabular techniques would outperform our method because they contain less noise and therefore can yield more precise gradient estimates.

The baselines we have used are Proximal Policy Optimization (PPO) \citep{schulman2017ppo}, Neural Fictitious Self-play (NFSP) \citep{heinrich2016nfsp}, Policy Space Response Oracles (PSRO) \citep{lanctot2017psro}, and Randomized Policy Networks (RPN) \citep{martin2023rpn}. We provide more details about the baselines and their hyperparameter setting in \Cref{app:exp_details}.

The results of this experiment are in \Cref{fig:one_shot}. We have used some games that contain non-finite support equilibria. The finite mixture model cannot approximate an arbitrary strategy. It is limited by the number of Gaussians. PSRO adds a single Gaussian strategy to the mix at each iteration, so it gradually mixes among more strategies than the mixture model. This can be seen in the Glicksberg-Gross game, where PSRO converges closer to the equilibrium. In the Appendix, we present an ablation study on the same games, varying the number of Gaussian components to show this effect in detail.

\begin{figure*}[!h]
    \centering    
    \hfill % Adds horizontal space between subfigures
    \begin{subfigure}[b]{0.245\textwidth}
        \centering
    \includegraphics[width=0.99\linewidth]{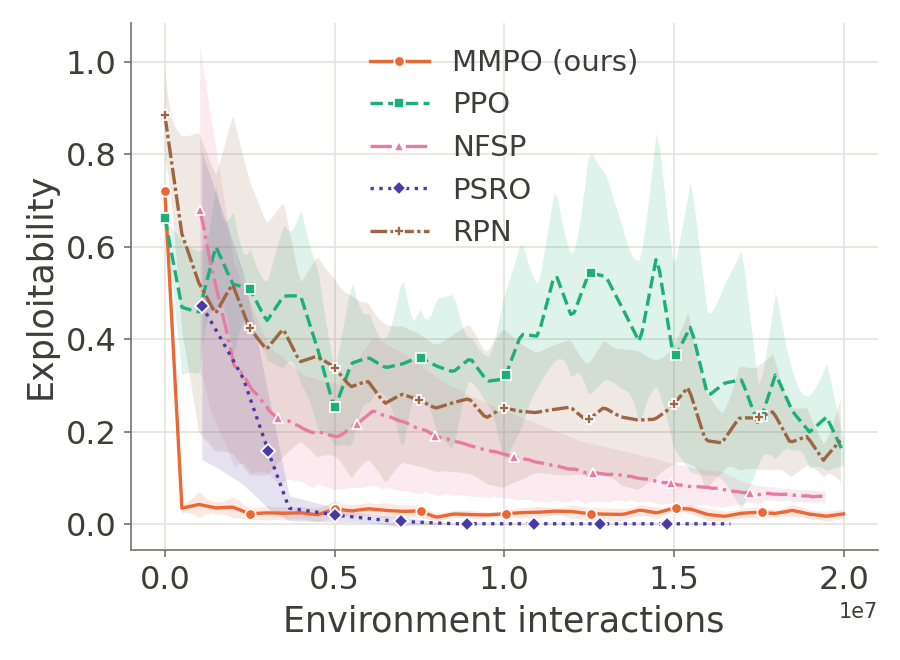}
        \caption{Glicksberg-Gross game}
        \label{fig:one_shot_glicks}
    \end{subfigure}  
    \begin{subfigure}[b]{0.245\textwidth}
        \centering
    \includegraphics[width=0.99\linewidth]{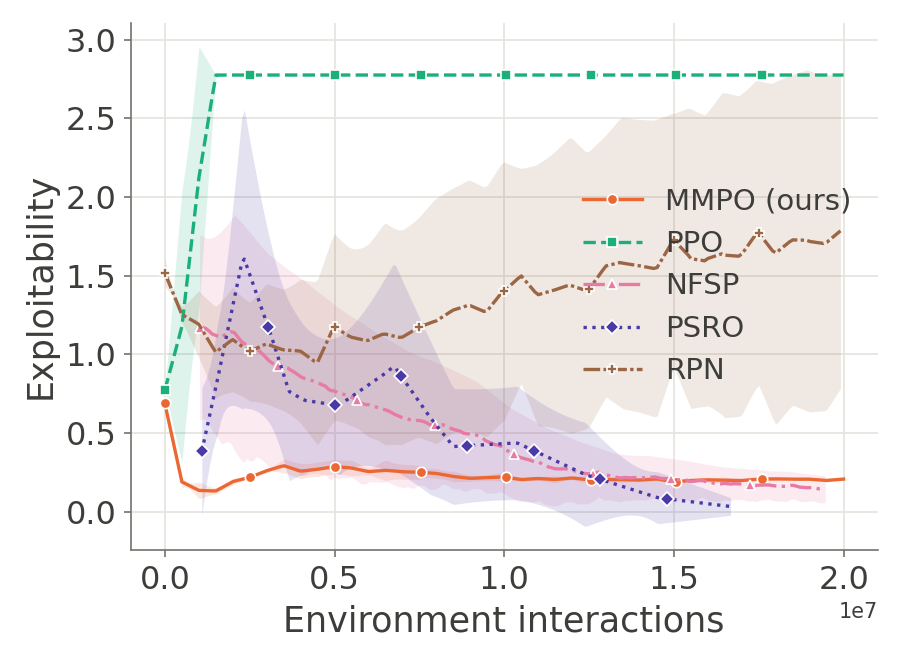}
        \caption{Circle game}
        \label{fig:one_shot_circle}
    \end{subfigure}
    \hfill % Adds horizontal space between subfigures
    \begin{subfigure}[b]{0.245\textwidth}
        \centering
    \includegraphics[width=0.99\linewidth]{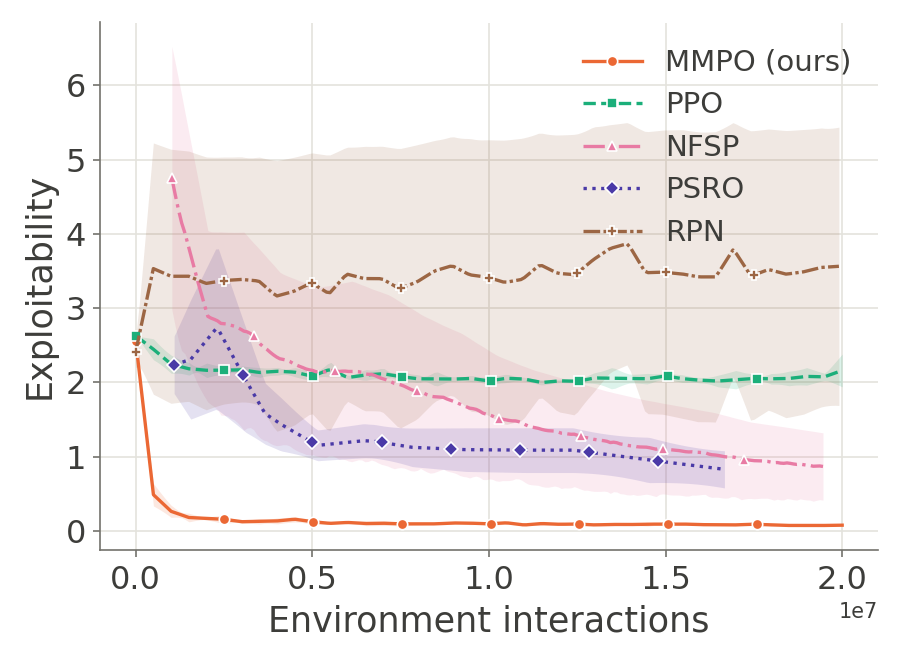}
        \caption{Two-point game}
        \label{fig:one_shot_two}
    \end{subfigure}  
    \begin{subfigure}[b]{0.245\textwidth}
        \centering
    \includegraphics[width=0.99\linewidth]{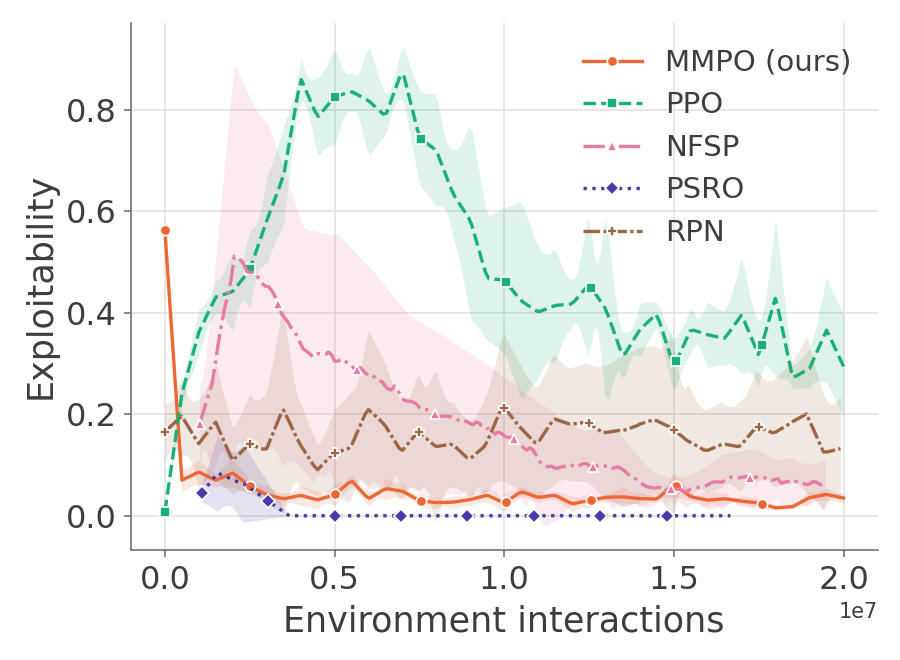}
        \caption{Matching Pennies}
        \label{fig:one_shot_mmd}
    \end{subfigure}
    \caption{Exploitability with 95\% confidence intervals of different baselines in one-shot games}
    \label{fig:one_shot} 
\end{figure*}  

\subsection{Sequential games}
\label{sec:exp_sequential}
We also evaluated \AlgorithmShort{} in sequential games, which is its intended setting. We have used four different games: Kuhn Poker, Leduc Poker, and Sequential Blotto with either 3 or 5 turns. We provide more details on each game in the Appendix. In \Cref{fig:seq}, we show exact exploitability in Kuhn Poker and approximate exploitability computed via PPO for other games, as a function of environment interactions and wall time.

In all games except Leduc, PPO cannot sufficiently model the Nash equilibrium because it uses only a single Gaussian. \AlgorithmShort{} converges slightly slower than PPO, but it always converges close to the equilibrium. With the increased training time, PSRO converges closer to the Nash equilibrium than \AlgorithmShort{}, which is because PSRO adds a single Gaussian strategy to the mix at every iteration, so eventually it contains more Gaussians than \AlgorithmShort{}. However, \AlgorithmShort{} converges much quicker to strong strategies both in terms of wall-time and environment interactions, which is crucial, especially in large games.

The approximate exploitabilities of some strategies in Blotto are negative because we compute a best response using PPO, so the resulting value is only a lower bound on the true exploitability. 

\begin{figure*}[!h]
    \centering    
    \begin{subfigure}[b]{0.245\textwidth}
        \centering
    \includegraphics[width=0.99\linewidth]{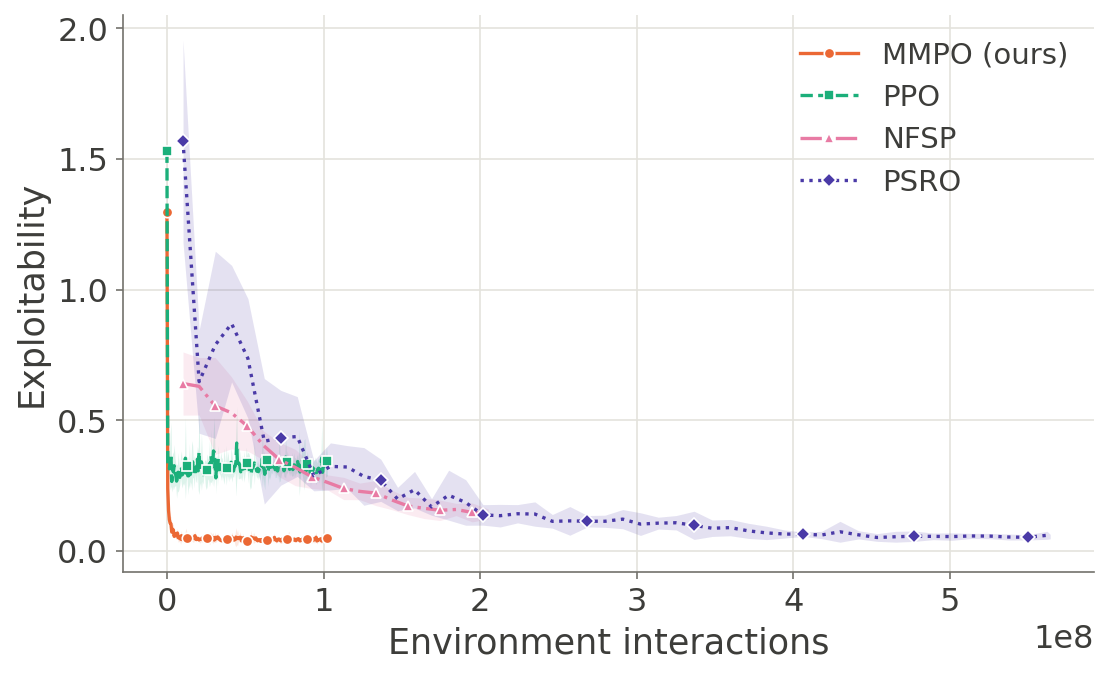}
        \caption{Kuhn Poker}
        \label{fig:seq_kuhn}
    \end{subfigure}
    \hfill % Adds horizontal space between subfigures
    \begin{subfigure}[b]{0.245\textwidth}
        \centering
    \includegraphics[width=0.99\linewidth]{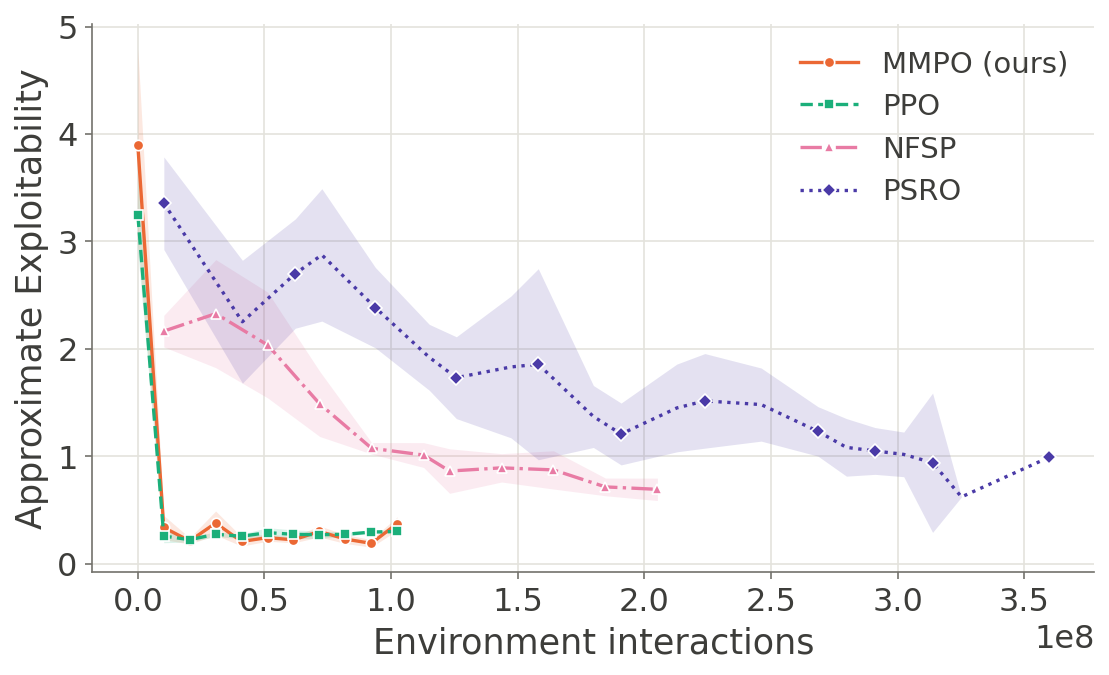}
        \caption{Leduc Poker}
        \label{fig:seq_leduc}
    \end{subfigure}  
    \begin{subfigure}[b]{0.245\textwidth}
        \centering
    \includegraphics[width=0.99\linewidth]{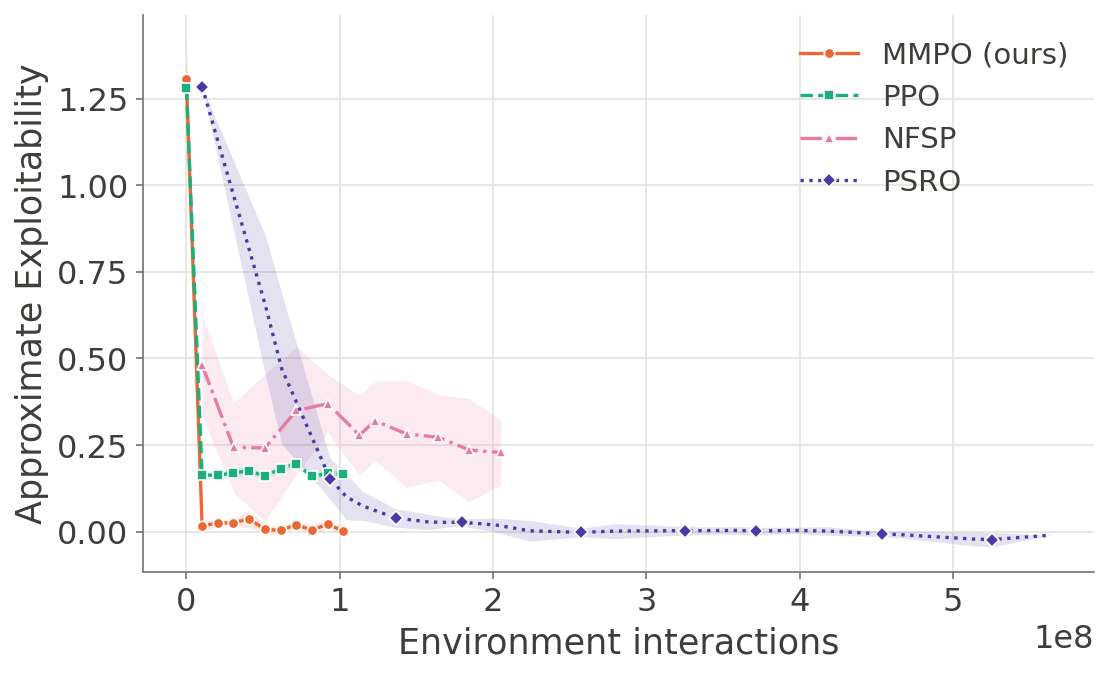}
        \caption{Blotto 3}
        \label{fig:seq_blotto3}
    \end{subfigure}
    \hfill % Adds horizontal space between subfigures
    \begin{subfigure}[b]{0.245\textwidth}
        \centering
    \includegraphics[width=0.99\linewidth]{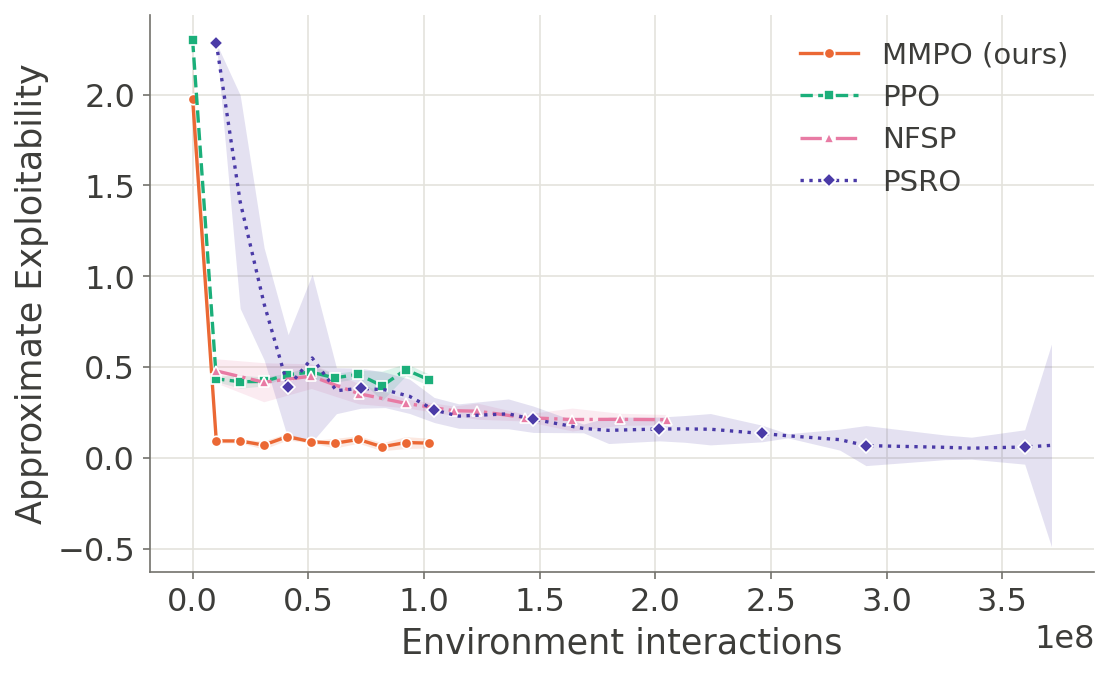}
        \caption{Blotto 5}
        \label{fig:seq_blotto5}
    \end{subfigure}
    \caption{Approximate exploitability with 95\% confidence intervals of different baselines in sequential games}
    \label{fig:seq} 
\end{figure*}  

\subsection{Heads-up no limit Texas hold'em}
\label{sec:exp_hunl}
In the last experiment, we evaluate the scalability of our approach by applying it to Heads-up no-limit Texas Hold'em. An advantage of mixture training is that the algorithm explores parts of the game that would not be visited with static discretization. If a discretized bot is matched against a bot trained with a different discretization, it may encounter out-of-distribution game states. This problem has been observed in Poker bots in the past, and at first, it was solved by interpolating between the closest legal states and later by subgame solving. \AlgorithmShort{} avoids both of these as it naturally explores different bets during the training.

Recently, subgame solving has been combined with policy-gradient algorithms \citep{kubicek2026ttrl}. It is possible to combine our training with this subgame-solving approach. We decided not to apply it here because our goal is not to create the strongest Poker bot, but to evaluate the policy gradient approach with our reparametrization. 

We trained the strategy network for 4 days and during training we used two evaluation metrics, the first is local best response \cite{lisy2017lbr}, which is a Poker-specific evaluation tool that is commonly used as a lower bound for exploitability in Texas Hold'em, the second is head-to-head performance against Slumbot, winner of the 2018 annual computer poker competition \citep{Jackson2013Slumbot}, which is the strongest available bot which does not use any subgame solving. Both are shown in \Cref{fig:hunl_main}. We have also used the last checkpoint for more matches against Slumbot, and after playing $900000$ hands, the average win-rate was $13.5 \pm 36.7$ mbb/hand. Based on the information provided by the Slumbot, \AlgorithmShort{} scored $28.2$ mbb/hand better than Slumbot would score against itself with the same cards. Local best response value of the final checkpoint is $-1622.7 \pm 158.9$.  

Based on these results, our approach, trained via self-play for few days with minimal domain knowledge, can stay on par with a bot developed over five years specifically for the Annual Computer Poker Competition. We use no subgame solving, card abstraction, or fixed bet abstraction; the only poker-specific choices are expressing bets as a fraction of the pot, a common practice in Poker AI, and initializing the Gaussian means to cover the interval $[\frac{1}{3}, 3]$ of the pot, from which they are free to move during training (\Cref{app:hunl}).

% We have also matched the latest checkpoint against the strongest available Poker bot from GTO wizard which uses subgame solving \citep{provost2026gto}. In \Cref{tab:h2h} we show the head-to-head win-rate against both Slumbot and GTO Wizard with 95\% Confidence intervals, along with the games played.

\begin{figure*}[!h]
    \centering    
    \begin{subfigure}[b]{0.49\textwidth}
        \centering
    \includegraphics[width=0.99\linewidth]{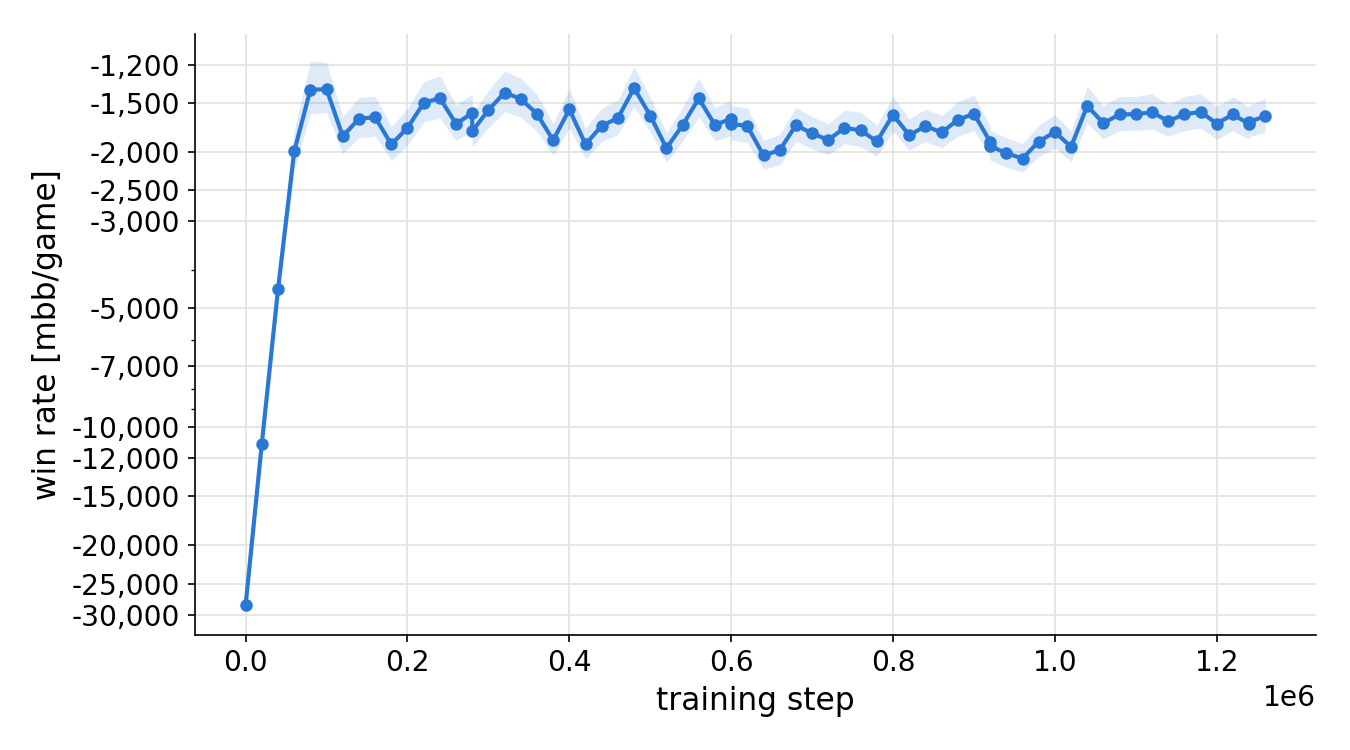}
        \caption{Local best response}
        \label{fig:hunl_lbr}
    \end{subfigure}
    \hfill % Adds horizontal space between subfigures
    \begin{subfigure}[b]{0.49\textwidth}
        \centering
    \includegraphics[width=0.99\linewidth]{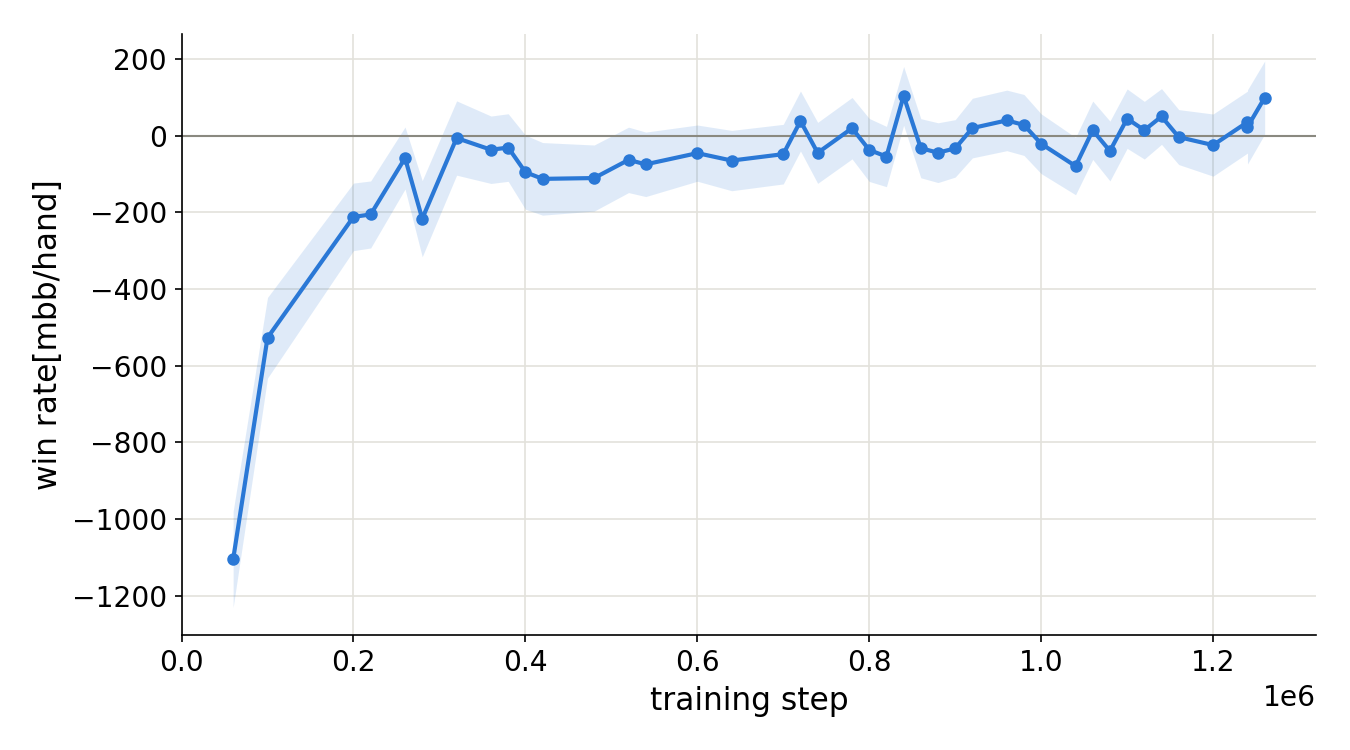}
        \caption{Slumbot}
        \label{fig:hunl_slumbot}
    \end{subfigure}  
    \caption{Head-to-head performance with 95\% confidence intervals against different methods during the training.}
    \label{fig:hunl_main} 
\end{figure*}  

% \begin{table}[]
%     \centering
%     \begin{tabular}{l|c|c|c}
%          & Slumbot & GTO Wizard & GTO Wizard AIVAT  \\
%          Hands & $20000$  & $50$ & $50$ \\
%          Win rate [mbb/hand] & $31 \pm$ & $-100 \pm$  & $-100 \pm$
%     \end{tabular}
%     \caption{Caption}
%     \label{tab:h2h}
% \end{table}
\section{Related Work} 
\textit{Counterfactual regret minimization (CFR)} is a tabular algorithm, which minimizes counterfactual regret at each decision point for both players \citep{zinkevich2007cfr,tammelin2015cfrp}. It has had tremendous success in Texas Hold'em, especially in its application with subgame solving \citep{burch2014cfrd,brown2018libratus,moravcik2017deepstack}. The subgame-solving version of CFR was extended to function-approximation settings with Deep CFR, and its modern extensions, including ReBeL or Dream \citep{brown2019deepcfr,brown2020rebel,steinberger2020dream}. However, these extensions rely on discrete action space, and to the best of our knowledge, scalable counterfactual regret minimization for continuous action spaces has not yet been developed.

% However, CFR only converges in average strategies, which in function approximation settings requires keeping all the checkpoint strategies and then expensively sampling from them. Even if some modern variants of CFR converge in the last iterate, these techniques have not been successfully scaled to function-approximation regimes that maintain the last-iterate guarantee \citep{farina2019pcfr,zhang2026sipcfr}. However, these algorithms rely on discrete action spaces, and, to the best of our knowledge, a scalable counterfactual regret minimization algorithm for continuous action spaces has not yet been developed.

\textit{Fictitious Self-play} and \textit{Double oracle} are tabular algorithms in which each player maintains an ensemble of strategies and then mixes among them \citep{brown1951fp,heinrich2015fsp,mcmahan2003do,bosansky2013do}. The new strategies are computed as a best response to the current ensemble. These algorithms are easily combined with deep learning, because the best response can be trained via reinforcement learning \citep{heinrich2016nfsp,lanctot2017psro}. Also, they can be easily used for games with continuous action spaces by modifying the algorithm to find the best response. AlphaStar has used a League system, in which different types of agents were used in the ensemble, with some trying to exploit specific weaknesses of other strategies \citep{vinyals2019alphastar}. However, these algorithms are often sample-inefficient because training the best response is the bottleneck, and they typically require hundreds of best responses.

\textit{Policy-gradient algorithms} have been central to the successes in single-player reinforcement learning. Stochastic value gradients (SVG) parametrize the strategy as a Gaussian and then train its mean and standard deviation directly from observed traces \citep{heess2015svg}. This same reparametrization trick was later used in proximal policy optimization (PPO) or soft actor critic (SAC) \citep{schulman2017ppo,haarnoja2018sac}.  Prior work has already introduced a mixture of Gaussians as a reparametrization of the strategy, but those focused on increasing the variability of trained strategies and did not take adversarial opponents into account \citep{baram2021mixture,ren2021mixture,islam2025mixture}. Multi-agent PPO (MAPPO) has extended PPO into a multi-agent setting, but it primarily focused on cooperative settings \citep{yu2022mappo}. Some extensions have been applied in adversarial settings, but they use only a single Gaussian to parametrize the strategy \citep{li2019robust,rahman2022robust}. PPO and MAPPO have achieved surprisingly strong performance in some zero-sum games, exemplified by OpenAI Five \citep{openai2019five,rudolph2026reevaluating} or AlphaHoldem \citep{zhao2022alphaholdem}. Still, they do not provide any convergence guarantees and may converge to highly exploitable strategies when used in self-play \citep{sokota2022mmd}. Magnetic mirror descent has shown that in games with discrete actions, using the KL-regularization with respect to a magnet strategy guarantees convergence to a Quantal Response Equilibrium \citep{sokota2022mmd}. Regularized Nash dynamics have shown that, by decreasing the regularization strength or by moving the magnet, these algorithms converge to a Nash equilibrium \citep{perolat2021rnad,sokota2022mmd,kalogiannis2026convergence}. Our work falls into this category of algorithms, as \AlgorithmShort{} parametrizes the strategy as a mixture of Gaussian and categorical distributions, with both components trained via a KL-regularized policy-gradient algorithm. %However based on our results, we can

\textit{Action abstraction} shrinks down the action space of the game, by either removing some actions or by discretizing the action space. A naive approach to discretization is to use a uniform grid for the actions. This approach is often quite effective, but it scales poorly in extensive-form games because errors can propagate. However, it has been shown that if the game utilities are Lipschitz, then the error of this uniform abstraction can be bounded \citep{kroer2015discretization}. Abstractions were a major part of Poker solvers, reducing the size of the game tree and making standard approaches applicable \citep{Jackson2013Slumbot,brown2018libratus,moravcik2017deepstack}. Several techniques have been developed to dynamically adjust the abstraction in Poker bets, but they require the game itself to be tractable for running CFR \citep{hawkin2012discretization,brown2014regret}. A different perspective on double oracle algorithms is that the pure best responses can be seen as a discretization \citep{adam2021do}. Our approach can be seen as finding an action abstraction while also solving the game using this abstraction. And since \AlgorithmShort{} does not explicitly construct the game, it is not limited by the size of the (abstracted) game tree.
\section{Conclusion}
We have extended the regularized policy-gradient techniques from discrete-action games to continuous-action games with algorithm \AlgorithmName{} (\AlgorithmShort{}). Many real-world adversarial situations in economics, robotics, or sports are continuous, and as such, it is important to develop scalable and efficient algorithms for these settings. \AlgorithmShort{} parametrizes the strategy as a mixture of Gaussians, enabling it to be used even in games that combine discrete and continuous actions. We show that even with a single Gaussian, KL regularization enables convergence in games where plain gradient descent diverges. We compare our method across several domains and baselines, demonstrating that it is both efficient and robust. We also scaled \AlgorithmShort{} to Heads-up No Limit Texas Hold'em, where we show that training in self-play without any subgame-solving and with minimal domain-specific knowledge manages to stay on par with a strong Poker bot.
\newpage
\section*{AI use statement} 
In this work, we used these AI tools for these tasks: 
\begin{itemize}
    \item Experiments: LLM-assisted programming, scheduling experiments, managing the results of the experiments, such as plotting and summarization. Models: Claude Opus 4.8, Claude Sonnet 5, Claude Opus 5, Claude Fable 5 through Claude Code. GPT5.6-Sol, GPT5.6-Terra, GPT6-Astra through OpenAI Codex. Composer 2.5 and Cursor Grok 4.6 through Cursor
    \item Help with proofs: Fleshing out details and formalizing the proofs in the paper. Models: Claude Opus 4.8 and Claude Opus 5 through Claude Code.
    \item Exploration of the related literature. Models: Gemini 3.6 Flash and Gemini 3.6 Thinking. Claude Opus 4.8 and Claude Opus 5 through Claude Code.
    \item Writing: Improving formulation and overall flow of the paper. Models: Gemini 3.6 Flash, Claude Opus 5 through Claude Code.
\end{itemize}
We have not used AI tools to generate the main ideas in the paper, draft the sections of the paper, produce data for experiments, or generate citations. We have reviewed all AI-assisted work by testing the generated implementation and verifying the generated parts of the proofs. We take responsibility for the final content of this work, including text, claims, or artifacts produced with the aid of generative AI.  

\bibliography{iclr2026_conference}
\bibliographystyle{iclr2027_conference}
\newpage
\appendix
\crefalias{section}{appsec} 
\section{Proofs}
\label{app:proof}

\begin{fact} 
    \begin{enumerate}
        \item KL divergence between two univariate Gaussians $p, q$ with corresponding means $\Mean_p, \Mean_q$ and standard deviations $\Variance_p, \Variance_q$ is
        $$\KLDiv(p, q) = \log(\frac{\Variance_q}{\Variance_p}) + \frac{\Variance_p^2 + (\Mean_p - \Mean_q)^2}{2\Variance_q^2} - \frac{1}{2}$$
        \item First order derivatives of KL divergence between two univariate Gaussians $p, q$ wrt. $\Mean_p$ and $\Variance_p$ are
        $$\frac{\partial }{\partial \Mean_p} \KLDiv(p, q) = \frac{\Mean_p - \Mean_q}{\Variance_q^2}$$
        $$\frac{\partial }{\partial \Variance_p} \KLDiv(p, q) = \frac{\Variance_p}{\Variance_q^2} - \frac{1}{\Variance_p} $$
        \item Second order derivatives of KL divergence between two univariate Gaussians $p, q$ wrt. $\Mean_p$ and $\Variance_p$ are
        $$\frac{\partial^2 }{\partial \Mean_p^2}  \KLDiv(p, q) =  \frac{1}{\Variance_q^2}$$
        $$\frac{\partial^2 }{\partial \Variance_p^2}  \KLDiv(p, q) =  \frac{1}{\Variance_q^2} + \frac{1}{\Variance_p^2}$$
        $$\frac{\partial^2 }{\partial \Mean_p \partial \Variance_p}  \KLDiv(p, q) = 0$$
         
        % \item KL divergence between two multivariate $d$-dimensional Gaussians $p, q$ with corresponding means $\Mean_p, \Mean_q$ and covariance matrices $\Covariance_p, \Covariance_q$ is
        % $$\KLDiv(p, q) = \frac{1}{2} \big(\log(\frac{|\Covariance_q|}{|\Covariance_p|}) + \text{tr}(\Covariance_q^{-1} \Covariance_p) + (\Mean_p - \Mean_q)^\top \Covariance_q^{-1} (\Mean_p - \Mean_q) - d\big)$$
        % \item First order derivatives of KL divergence between two multivariate Gaussians $p, q$ wrt. $\Mean_p$ and $\Covariance_q$ are
        % $$\nabla_{\Mean_p} \KLDiv(p, q) =  \Covariance_q^{-1} (\Mean_p - \Mean_q) $$
        % $$\nabla_{\Covariance_p} \KLDiv(p, q) = \frac{1}{2}  (\Covariance_q^{-1} - \Covariance_p^{-1}) $$
        % \item Second order derivatives of KL divergence between two multivariate Gaussians $p, q$ wrt. $\Mean_p$ and $\Covariance_p$ are
        % $$\nabla^2_{\Mean_p} \KLDiv(p, q) =  \Covariance_q^{-1}$$
        % $$\nabla^2_{\Covariance_p} \KLDiv(p, q) =  \Covariance_p^{-1} \otimes \Covariance_p^{-1}$$
        % $$\nabla_{\Mean_p} \nabla_{\Covariance_p} = \nabla_{\Covariance_p} \nabla_{\Mean_p} = 0$$
        
    \end{enumerate}
\end{fact}

\begin{restatable}[]{lemma}{BoundedNonMonotonicity}
\label{lem:bound_non_mono}
    Under \cref{as:lip}, for all $\Strategy{}, \Strategy{}'$ the non-monotonicity of $V$ is bounded
     $$\langle \GameField(\Strategy{}) - \GameField(\Strategy{}'), \Strategy{} - \Strategy{}' \rangle \geq - \LipConst_V \|\Strategy{} - \Strategy{}'\|^2 $$
\end{restatable}
\begin{proof}
    The \cref{as:lip} states that the function is $\LipConst_V$-Lipschitz continuous. By using Cauchy-Schwarz inequality $\langle u, v \rangle \leq \|u\| \|v\|$
    \begin{align}
        \langle \GameField(\Strategy{}) - \GameField(\Strategy{}'), \Strategy{} - \Strategy{}' \rangle \geq - \|\GameField(\Strategy{}) - \GameField(\Strategy{}')\| \cdot \|\Strategy{} - \Strategy{}'\| \geq  -\LipConst_V \|\Strategy{} - \Strategy{}'\|^2
    \end{align}
\end{proof}

\begin{restatable}[]{lemma}{RegularizerMonotonicity}
\label{lem:regularizer_mono}
    Assume a strategy space $\Strategies{}$, where $\Variance_1, \Variance_2 \geq \Variance_{\text{min}} > 0$.  The regularizer $\RegularizerField(\Strategy{})$ is a continuously differentiable, $\MonConst_R$-strongly monotone and $\LipConst_R$-Lipschitz on $\Strategy{}$ for a magnet strategy $\Strategy{}^M \in \Strategies{}$ with
    
    \begin{equation}
        \MonConst_R = \min(\frac{1}{\Variance_1^{2,M}}, \frac{1}{\Variance_2^{2,M}}) > 0
    \end{equation}
    \begin{equation}
        \LipConst_R = \frac{1}{\Variance_{\text{min}}^2} + \max(\frac{1}{\Variance_1^{2,M}}, \frac{1}{\Variance_2^{2,M}}) < \infty
    \end{equation}
    
\end{restatable}
\begin{proof}
    Using the fact that 
    
    \begin{equation}
        \RegularizerField(\Strategy{}) = \begin{pmatrix}
            \frac{\partial }{\partial \Mean_1} \KLDiv(\Strategy{}, \Strategy{}^M) \\
            \frac{\partial }{\partial \Variance_1} \KLDiv(\Strategy{}, \Strategy{}^M) \\
            \frac{\partial }{\partial \Mean_2} \KLDiv(\Strategy{}, \Strategy{}^M) \\
            \frac{\partial }{\partial \Variance_2} \KLDiv(\Strategy{}, \Strategy{}^M) 
        \end{pmatrix} 
        = \begin{pmatrix}
            \frac{\Mean_1 - \Mean_1^M}{\Variance_1^{2,M}} \\
            \frac{\Variance_1}{\Variance_1^{2,M}} - \frac{1}{\Variance_1}\\
            \frac{\Mean_2 - \Mean_2^M}{\Variance_2^{2,M}} \\
            \frac{\Variance_2}{\Variance_2^{2,M}} - \frac{1}{\Variance_2}
        \end{pmatrix}
    \end{equation}

The Jacobian of $\RegularizerField(\Strategy{})$ is
    \begin{equation}
        \Jacobian_\RegularizerField(\Strategy{}) = \text{diag}\Big(\frac{1}{\Variance_1^{2,M}}, \frac{1}{\Variance_1^{2}} + \frac{1}{\Variance_1^{2,M}}, \frac{1}{\Variance_2^{2,M}}, \frac{1}{\Variance_2^{2}} + \frac{1}{\Variance_2^{2,M}}  \Big)
    \end{equation}

The Jacobian $\Jacobian_\RegularizerField(\Strategy{})$ is a diagonal matrix, and each diagonal entry is positive, because $\Variance_1, \Variance_1^M, \Variance_2, \Variance_2^M > 0$.

It holds that $\MonConst_\RegularizerField \leq \frac{1}{\Variance_{\Player}^{2, M}} \leq \frac{1}{\Variance_{\Player}^{2, M}} + \frac{1}{\Variance_{\Player}^{2}} \leq \LipConst_\RegularizerField $  for $\Player \in \{1, 2\}$ and all strategies $\Strategy{}$. Out of this follows $\MonConst_\RegularizerField \IdentityM \preceq \Jacobian_\RegularizerField(\Strategy{}) \preceq \LipConst_\RegularizerField \IdentityM $, where $\preceq$ is an Loewner ordering $\textbf{A} \preceq \textbf{B} \Leftrightarrow x^\top (\textbf{B} - \textbf{A}) x  \geq 0$, for all $x$. In other words that $(\textbf{B} - \textbf{A})$ is positive semi-definite.

Strategy space $\Strategies{}$ is convex, therefore if $\Strategy{}, \Strategy{}' \in \Strategies{} $ then $\Strategy{}' + t(\Strategy{} - \Strategy{}'); t \in [0, 1]$  is also in $\Strategies{}$. From the fundamental theorem of calculus follows
\begin{equation}
    \langle \RegularizerField(\Strategy{}) - \RegularizerField(\Strategy{}'), \Strategy{} - \Strategy{}' \rangle = \int_0^1 ( \Strategy{} - \Strategy{}')^\top \Jacobian_\RegularizerField(\Strategy{}' + t( \Strategy{} - \Strategy{}')) ( \Strategy{} - \Strategy{}') dt \geq \MonConst_\RegularizerField \| \Strategy{} - \Strategy{}'\|^2
\end{equation}
\begin{equation}
\begin{split}
    \|\RegularizerField(\Strategy{}) - \RegularizerField(\Strategy{}')\| &= \|\int_0^1 \Jacobian_\RegularizerField(\Strategy{}' + t( \Strategy{} - \Strategy{}'))(\Strategy{} - \Strategy{}') dt\| \\
    &\leq \int_0^1 \|\Jacobian_\RegularizerField(\Strategy{}' + t( \Strategy{} - \Strategy{}'))(\Strategy{} - \Strategy{}')\| dt \\
    &\leq \int_0^1 \|\Jacobian_\RegularizerField(\Strategy{}' + t( \Strategy{} - \Strategy{}'))\|_{\text{op}} \|(\Strategy{} - \Strategy{}')\| dt\\
    &\leq \int_0^1 \LipConst_\RegularizerField \|(\Strategy{} - \Strategy{}')\|dt = \LipConst_\RegularizerField \|(\Strategy{} - \Strategy{}')\|
    \end{split}
\end{equation} 
\end{proof}

\begin{restatable}[]{lemma}{MmdMonotonicity}
\label{lem:mmd_mono}
    Under \cref{as:lip} and $\RegularizationStrength > \frac{\LipConst_V}{\MonConst_\RegularizerField}$ the $\GameField^\RegularizerField$ is $\MonConst_M$-strongly monotone and $\LipConst_M$-Lipschitz on strategy space $\Strategies{}$, where $\Variance_1, \Variance_2 \geq \Variance_{\text{min}} > 0$ with

    \begin{equation}
        \MonConst_M = \RegularizationStrength \MonConst_\RegularizerField - \LipConst_\GameField > 0
    \end{equation}
    \begin{equation}
        \LipConst_M = \RegularizationStrength \LipConst_\RegularizerField + \LipConst_\GameField < \infty
    \end{equation}
    \begin{equation}
        \LipConst_M \geq \MonConst_M
    \end{equation}
\end{restatable}
\begin{proof}
    \begin{equation}
    \begin{split}
        \langle \GameField^\RegularizerField(\Strategy{}) - \GameField^\RegularizerField(\Strategy{}'), \Strategy{} - \Strategy{}' \rangle &= \langle \GameField(\Strategy{}) + \RegularizationStrength \RegularizerField(\Strategy{}) - \GameField(\Strategy{}') - \RegularizationStrength \RegularizerField(\Strategy{}'), \Strategy{} - \Strategy{}' \rangle \\
        &= \langle \GameField(\Strategy{}) - \GameField(\Strategy{}'), \Strategy{} - \Strategy{}' \rangle + \RegularizationStrength \langle \RegularizerField(\Strategy{}) - \RegularizerField(\Strategy{}'), \Strategy{} - \Strategy{}' \rangle\\
        &\geq  - \LipConst_V \|\Strategy{} - \Strategy{}'\|^2  + \RegularizationStrength \MonConst_\RegularizerField  \|\Strategy{} - \Strategy{}'\|^2  \\
        &= (\RegularizationStrength \MonConst_\RegularizerField - \LipConst_V )\|\Strategy{} - \Strategy{}'\|^2 = \MonConst_M \|\Strategy{} - \Strategy{}'\|^2
    \end{split}
    \end{equation}
    $\MonConst_M > 0 $ based on the assumption $\RegularizationStrength > \frac{\LipConst_\GameField}{\MonConst_\RegularizerField}$.
    \begin{equation}
        \begin{split}
            \|\GameField^\RegularizerField(\Strategy{}) - \GameField^\RegularizerField(\Strategy{}')\| &= \|\GameField(\Strategy{}) + \RegularizationStrength \RegularizerField(\Strategy{}) - \GameField(\Strategy{}') - \RegularizationStrength \RegularizerField(\Strategy{}')\| \\
            &\leq \|\GameField(\Strategy{}) - \GameField(\Strategy{}')\| + \RegularizationStrength \|\RegularizerField(\Strategy{}) - \RegularizerField(\Strategy{}')\| \\
            &\leq \LipConst_V \|\Strategy{} - \Strategy{}'\|  + \RegularizationStrength \LipConst_\RegularizerField \|(\Strategy{} - \Strategy{}')\| \\
            &= (\LipConst_V + \RegularizationStrength \LipConst_\RegularizerField ) \|\Strategy{} - \Strategy{}'\|  = \LipConst_{M}  \|\Strategy{} - \Strategy{}'\| 
        \end{split}
    \end{equation} 
    The fact that $\LipConst_M > \MonConst_M$ follows from \cref{lem:regularizer_mono}, specifically $\LipConst_R > \MonConst_R$ and from the fact that $\LipConst_\GameField > 0$.
\end{proof}

\begin{fact}
\label{fact:strategy}
    Let $\Strategies{}$ be nonempty, closed and convex. The Euclidean projection of any $\Strategy{}^+$ is defined as
    \begin{equation}
        \ProjStrategies(\Strategy{}^+) = \argmin_{\Strategy{} \in \Strategies{}}||\Strategy{} - \Strategy{}^+\|
    \end{equation}

    Then for any $\Strategy{}^+$
    \begin{enumerate}
        \item The $\ProjStrategies(\Strategy{}^+)$ is unique
        \item Feasible $\Strategy{}^\circ \in \Strategies{}$ satisfies $\Strategy{}^\circ = \ProjStrategies(\Strategy{}^+)$ if and only if
        \begin{equation}
        \langle \Strategy{}^+ - \Strategy{}^\circ, \Strategy{} - \Strategy{}^\circ  \rangle \leq 0 \quad\quad \forall \Strategy{} \in \Strategies{}
        \end{equation}
        \item $\ProjStrategies$ is non-expansive for any two such $\Strategy{}^+, \hat{\Strategy{}}^+$
        \begin{equation}
            \| \ProjStrategies(\Strategy{}^+) - \ProjStrategies(\hat{\Strategy{}}^+) \| \leq \| \Strategy{}^+ - \hat{\Strategy{}}^+ \|
        \end{equation}
    \end{enumerate}
\end{fact}

\begin{restatable}[]{lemma}{PgdaContraction}
\label{lem:pgda_contraction}
Let  \cref{as:lip} hold and $\RegularizationStrength > \frac{\LipConst_\GameField}{\MonConst_\RegularizerField}$. Let $\LipConst_M, \MonConst_M$ be the constants from \cref{lem:mmd_mono} and $0 < \LearningRate < \frac{2\MonConst_M}{\LipConst_M^2}$. Define the Projected Gradient Ascent Descent operator as
\begin{equation}
    \PGDA(\Strategy{}) = \ProjStrategies(\Strategy{} - \LearningRate \GameField^\RegularizerField(\Strategy{}))
\end{equation}
Then  $\PGDA$ is a contraction
\begin{equation}
    \|\PGDA(\Strategy{}) - \PGDA(\Strategy{}')\|^2 \leq r(\LearningRate) \|\Strategy{} - \Strategy{}'\|^2
\end{equation}
\end{restatable}      
\begin{proof}

    \begin{equation}
    \begin{split}
        \| \PGDA(\Strategy{}) - \PGDA(\Strategy{}')\|^2& = \|\ProjStrategies(\Strategy{} - \LearningRate \GameField^\RegularizerField(\Strategy{})) - \ProjStrategies(\Strategy{}' - \LearningRate \GameField^\RegularizerField(\Strategy{}')) \|^2 \\
        &\leq \| \Strategy{} - \Strategy{}' - \LearningRate (\GameField^\RegularizerField(\Strategy{}) - \GameField^\RegularizerField(\Strategy{}')) \|^2 \\
        &= \| \Strategy{} - \Strategy{}' \|^2 - 2\LearningRate \langle \Strategy{} - \Strategy{}', \GameField^\RegularizerField(\Strategy{}) - \GameField^\RegularizerField(\Strategy{}') \rangle + \LearningRate^2 \| \GameField^\RegularizerField(\Strategy{}) - \GameField^\RegularizerField(\Strategy{}') \|^2
    \end{split}
    \end{equation}
    The first inequality follows from the fact that $\Strategies{}$ is closed and convex, which makes $\ProjStrategies$ non-expansive.
    We will use the $\MonConst_M$-strong monotonicity and $\LipConst_M$-Lipschitz continuity to get
    \begin{equation}
          - 2\LearningRate \langle \Strategy{} - \Strategy{}', \GameField^\RegularizerField(\Strategy{}) - \GameField^\RegularizerField(\Strategy{}') \rangle \leq -  2\LearningRate \MonConst_M \| \Strategy{} - \Strategy{}' \|^2
    \end{equation}
    \begin{equation}
        \LearningRate^2 \| (\GameField^\RegularizerField(\Strategy{}) - \GameField^\RegularizerField(\Strategy{}')) \|^2 \leq \LearningRate^2 \LipConst_M^2  \| \Strategy{} - \Strategy{}' \|^2
    \end{equation}
    This gets the following inequality
    \begin{equation}
         \| \PGDA(\Strategy{}) - \PGDA(\Strategy{}')\|^2 \leq (1 - 2\LearningRate \MonConst_M + \LearningRate^2 \LipConst_M^2) \|\Strategy{} - \Strategy{}' \|^2
    \end{equation}
    Then the PGDA is a contraction when $0 < 1 - 2\LearningRate \MonConst_M + \LearningRate^2 \LipConst_M^2 < 1$, which holds for learning rate $0 < \LearningRate <  \frac{2\MonConst_M}{\LipConst_M^2} $.    
\end{proof}

\begin{restatable}[]{proposition}{MmdUniqueFixedPoint}
\label{prop:mmd_unique}
     Under \cref{as:lip} the problem $\VI(\GameField^\RegularizerField, \Strategies{})$ has a unique solution $\Strategy{}^*$ for every $\RegularizationStrength > \frac{\LipConst_V}{\MonConst_\RegularizerField}$ and $\Strategy{}^M \in \Strategies{}$. For every $\LearningRate > 0$ the $\Strategy{}^*$ is a fixed point
     \begin{equation}
         \Strategy{}^* = \PGDA(\Strategy{}^*)
     \end{equation}
     If in addition $\Strategy{}^*$ is in interior of $\Strategies{}$, then $\GameField^\RegularizerField(\Strategy{}^*) = 0$
\end{restatable}      
\begin{proof}
We start by showing that the fixed points $\Strategy{}^*$ of $\PGDA$ are solutions of the $\VI(\GameField^\RegularizerField, \Strategies{})$

\begin{equation}
\label{eq:fixed_point}
    \Strategy{}^* = \PGDA(\Strategy{}^*) \Leftrightarrow \langle \GameField^\RegularizerField(\Strategy{}^*), \Strategy{} - \Strategy{}^* \rangle \geq 0
\end{equation}

We will use \cref{fact:strategy}. The unprojected step is $\Strategy{}^+ = \Strategy{}^* - \LearningRate \GameField^\RegularizerField(\Strategy{}^*)$ and the feasible point is $\Strategy{}^* = \Strategy{}^\circ = \ProjStrategies(\Strategy{}^+)$. From fact 2 we know that it also needs to hold 

\begin{equation}
    \langle\Strategy{}^+ - \Strategy{}^\circ, \Strategy{} - \Strategy{}^\circ \rangle = \langle \Strategy{}^* - \LearningRate \GameField^\RegularizerField(\Strategy{}^*) - \Strategy{}^*, \Strategy{} - \Strategy{}^\circ \rangle = \langle  - \LearningRate \GameField^\RegularizerField(\Strategy{}^*), \Strategy{} - \Strategy{}^\circ \rangle \leq 0
\end{equation}

This means that the $ \langle \LearningRate \GameField^\RegularizerField(\Strategy{}^*), \Strategy{} - \Strategy{}^\circ \rangle \geq 0$.

Now we will show that fixed point exists. First the $0 < \LearningRate < \frac{2\MonConst_M}{\LipConst_M^2}$ is a non-empty because of \cref{lem:mmd_mono}. Let's fix any $\LearningRate$ from that feasible set. Based on \cref{lem:pgda_contraction}, we know that $\PGDA$ is a contraction and the $\Strategies{}$ is closed and nonempty. Through Banach fixed-point theorem, there is a fixed point for the fixed $\LearningRate$. This fixed point is $\Strategy{}^*$, when using \cref{eq:fixed_point} from left to right at that same $\LearningRate$ it is a solution of the $\VI(\GameField^\RegularizerField, \Strategies{})$. Then using the $\cref{eq:fixed_point}$ from right to left, this same fixed point is a solution for any feasible $\LearningRate$, not only the one used to find this point.

The uniqueness of the fixed point comes from the strong monotonicity.
\end{proof}

\MmdConvergence*
\begin{proof}
The strategy computed at iteration $k+1$ is
    \begin{equation}
        \Strategy{}^{k+1} = \PGDA(\Strategy{}^{k})
    \end{equation}
    By \cref{prop:mmd_unique}, there is a unique solution $\Strategy{}^* = \PGDA(\Strategy{}^*)$ for the same $\LearningRate$. Now let us use $\Strategy{}' = \Strategy{}^*$ in \cref{lem:pgda_contraction}
    \begin{equation}
        \| \Strategy{}^{k+1} - \Strategy{}^*\|^2 = \| \PGDA(\Strategy{}^{k}) - \PGDA(\Strategy{}^*)\|^2 \leq (1 - 2\LearningRate \MonConst_M + \LearningRate^2 \LipConst_M^2) \|\Strategy{}^{k} - \Strategy{}^* \|^2
    \end{equation}
    Iterating this procedure over $k$ steps results in
    \begin{equation}
         \|\Strategy{}^{k} - \Strategy{}^*\|^2 \leq r(\RegularizationStrength)^k\|\Strategy{}^{0} - \Strategy{}^*\|^2
    \end{equation}
\end{proof}

% This proves that the MMD converges to a fixed point. We will use this to also show that if the game is convex-concave, then updating magnet to the fixed point will gradually converge to a $\epsilon-$Nash equilibrium, limited only by the $\Variance_{\min}$. We will assume $\Strategies{}^{\Variance_{\min}} = \{ \Strategy{} \in \Strategies{}: \Variance = \Variance_{\min} \}$, which is that $\Strategies{}^{\Variance_{\min}}$ is a subset of strategies, that have fixed $\Variance$. The fact that the convex-concave games have a pure Nash follows from Sion's theorem. 

% Consider an oracle function $\Inner(\Strategy{}^M)$ which takes a magnet strategy $\Strategy{}^M$ as an input and returns a solution of $\VI(\GameField^\RegularizerField, \Strategies{}^{\Variance_{\min}})$ for this magnet. The goal of the algorithm is to find a solution which satisfies

% \begin{equation}
%     \Strategy{}^* = \langle \GameField(\Strategy{}^*), \Strategy{} - \Strategy{}^* \rangle \geq 0 \quad \forall \Strategy{} \in \Strategies{}^{\Variance_{\min}}
% \end{equation}

% For these following proofs let us assume $\Strategie?s{}^{\Variance_{\min}} = \{ \Strategy{} \in \Strategies{}: \Variance = \Variance_{\min} \}$, which is that $\Strategies{}^{\Variance_{\min}}$ is a subset of strategies, that have fixed $\Variance$.

\begin{restatable}[]{proposition}{OuterFixed}
\label{prop:outer_fixed}
For a convex-concave game, for strategy $\Strategy{}^M \in \Strategies{}^{\Variance_{\min}}$ it holds that
\begin{equation}
    \Strategy{}^M = \Inner(\Strategy{}^M) \Leftrightarrow  \langle \GameField(\Strategy{}^M), \Strategy{} - \Strategy{}^M \rangle \geq 0 \quad \forall \Strategy{} \in  \Strategies{}^{\Variance_{\min}}
\end{equation}

     % Under \cref{as:lip} the problem $\VI(\GameField^\RegularizerField, \Strategies{})$ has a unique solution $\Strategy{}^*$ for every $\RegularizationStrength > \frac{\LipConst_V}{\MonConst_\RegularizerField}$ and $\Strategy{}^M \in \Strategies{}$. For every $\LearningRate > 0$ the $\Strategy{}^*$ is a fixed point
     % \begin{equation}
     %     \Strategy{}^* = \PGDA(\Strategy{}^*)
     % \end{equation}
     % If in addition $\Strategy{}^*$ is in interior of $\Strategies{}$, then $\GameField^\RegularizerField(\Strategy{}^*) = 0$
\end{restatable}      
\begin{proof}
    Since $\Variance_\Player = \Variance_{\Player}^M$, we know that the $\frac{\partial }{\partial \Variance_\Player} \KLDiv(\Strategy{}, \Strategy{}^M) = 0$, as such it is enough to work with the $\Mean$. Since the magnet is a fixed point of $\Inner$, we know that
    \begin{equation}
    \begin{split}
        \langle \GameField^\RegularizerField(\Strategy{}^M), \Strategy{} - \Strategy{}^M \rangle \geq 0  \quad \forall \Strategy{} \in  \Strategies{}^{\Variance_{\min}}\\
        \langle \GameField(\Strategy{}^M) + \RegularizationStrength \RegularizerField(\Strategy{}^M), \Strategy{} - \Strategy{}^M \rangle \geq 0  \quad \forall \Strategy{}  \in  \Strategies{}^{\Variance_{\min}}\\
        \langle \GameField(\Strategy{}^M), \Strategy{} - \Strategy{}^M \rangle \geq 0  \quad \forall \Strategy{}  \in  \Strategies{}^{\Variance_{\min}}\\
    \end{split}
    \end{equation}
    which comes from the fact that $\RegularizerField(\Strategy{}^M) = 0$ for fixed point. Conversely, the same approach shows that any solution of $\langle \GameField(\Strategy{}^M), \Strategy{} - \Strategy{}^M \rangle \geq 0$ solves also the $\VI(\GameField(\Strategy{}^M) + \RegularizationStrength \RegularizerField, \Strategies{})$ 
\end{proof}

\begin{restatable}[]{lemma}{OuterUpdateProxStep}
\label{lem:outer_update_prox_step}
    Let $\Strategy{}^M \in  \Strategies{}^{\Variance_{\min}}$ and $\Strategy{}^+ = \Inner(\Strategy{}^M)$ it holds that
    \begin{equation}
        \langle \GameField(\Strategy{}^+), \Strategy{}^+ - \Strategy{} \rangle \leq \frac{\RegularizationStrength}{2\Variance_{\min}^2}\big(\|\Mean - \Mean^M\|^2 - \|\Mean - \Mean^+\|^2 - \|\Mean^+ - \Mean^M \|^2 \big)
    \end{equation}
    Moreover, for $\Strategy{}^*$ holds
    \begin{equation}
        \| \Mean^* - \Mean^+ \|^2 \leq \|\Mean^* - \Mean^M\|^2 - \|\Mean^+ - \Mean^M \|^2
    \end{equation}
\end{restatable}
\begin{proof}
    We know that $\Strategy{}^+$ is a fixed point of $\VI(\GameField^\RegularizerField, \Strategies{}^{\Variance_{\min}})$, so it holds
    \begin{equation} 
        \langle \GameField^\RegularizerField(\Strategy{}^+), \Strategy{} - \Strategy{}^+ \rangle \geq 0 \quad \forall \Strategy{} \in \Strategies{}^{\Variance_{\min}}
    \end{equation}

    Let us focus only on the $\Mean$ part of the inequality, because the $\Variance$ is fixed. We use $\GameField(\Mean)$ for only the mean part of the field

    \begin{equation}
        \begin{split}
        \langle \GameField^\RegularizerField(\Mean^+), \Mean - \Mean^+ \rangle &= \langle \GameField(\Mean^+) + \RegularizationStrength \RegularizerField(\Mean^+), \Mean - \Mean^+ \rangle\\
        &=  \langle \GameField(\Mean^+) + \frac{\RegularizationStrength}{\Variance_{\min}^2} (\Mean^+ - \Mean^M), \Mean - \Mean^+ \rangle\\
        &\geq 0
        \end{split}
    \end{equation}
We can further reorder this into
\begin{equation}
    \langle \GameField(\Strategy{}^+), \Strategy{}^+ - \Strategy{} \rangle = \langle \GameField(\Mean{}^+), \Mean{}^+ - \Mean{} \rangle \leq \frac{\RegularizationStrength}{\Variance_{\min}^2} \langle \Mean^+ - \Mean^M, \Mean - \Mean^+ \rangle
\end{equation}
Now we can now use following identity $2\langle q - s, p - q\rangle = \|p-s\|^2 - \|p - q\|^2 - \|q - s\|^2$ to get
\begin{equation}
     \frac{\RegularizationStrength}{\Variance_{\min}^2} \langle \Mean^+ - \Mean^M, \Mean - \Mean^+ \rangle =  \frac{\RegularizationStrength}{2\Variance_{\min}^2} (\| \Mean - \Mean^M\|^2 - \| \Mean - \Mean^+ \|^2 - \|\Mean^+  - \Mean^M\|^2)
\end{equation}
We know that for $\Strategy{}^*$ it holds $\langle \GameField(\Strategy{}^*), \Strategy{}^+ - \Strategy{}^* \rangle \geq 0$. Due to the monotonicity of the convex-concave game it also holds  $\langle \GameField(\Strategy{}) - \GameField(\Strategy{}'), \Strategy{} - \Strategy{}' \rangle \geq 0$. Therefore
\begin{equation}
    \begin{split}
        \langle \GameField(\Strategy{}^+) - \GameField(\Strategy{}^*), \Strategy{}^+ - \Strategy{}^* \rangle &= \langle \GameField(\Strategy{}^+), \Strategy{}^+ - \Strategy{}^* \rangle  - \langle  \GameField(\Strategy{}^*), \Strategy{}^+ - \Strategy{}^* \rangle \geq 0\\
        \langle \GameField(\Strategy{}^+), \Strategy{}^+ - \Strategy{}^* \rangle  &\geq \langle  \GameField(\Strategy{}^*), \Strategy{}^+ - \Strategy{}^* \rangle \geq 0
    \end{split}
\end{equation}
Putting it all together, we get
\begin{equation}
    0 \leq \langle \GameField(\Strategy{}^+), \Strategy{}^+ - \Strategy{}^* \rangle \leq \frac{\RegularizationStrength}{2\Variance_{\min}^2}\big(\|\Mean^* - \Mean^M\|^2 - \|\Mean^* - \Mean^+\|^2 - \|\Mean^+ - \Mean^M \|^2 \big)
\end{equation}
Since it holds that $\RegularizationStrength \geq 0$ and $\Variance_{\min} \geq 0$, it has to hold that
\begin{equation}
\begin{split}
    \|\Mean^* - \Mean^M\|^2 - \|\Mean^* - \Mean^+\|^2 - \|\Mean^+ - \Mean^M \|^2 &\geq 0\\
    \|\Mean^* - \Mean^M\|^2 - \|\Mean^+ - \Mean^M \|^2 &\geq  \|\Mean^* - \Mean^+\|^2
\end{split}
\end{equation}
\end{proof}

\begin{restatable}[]{theorem}{OuterConvergence}
\label{thm:outer_convergence_}
For a convex-concave game and $\RegularizationStrength > \LipConst_\GameField \Variance_{\min}^2$, let $\Strategy{}^{M, t+1} = \Inner(\Strategy{}^{M, t})$ from arbitrary $\Strategy{}^{M, 0} \in \Strategies{}^{\Variance_{\min}}$, then it holds that
\begin{equation}
\sum_{t \geq 0} \|\Strategy{}^{M, t+1} - \Strategy{}^{M,t} \|^2 \leq D^2
\end{equation} 
where $D$ is the diameter of the action space. Moreover, $\Mean^{M, t} \to \overline{\Mean}$, where $\overline{\Strategy{}} = (\overline{\Mean}, \Variance_{\text{min}})$ solves variational inequalities $\VI(\GameField, \Strategies{}^{\Variance_{\text{min}}})$
\end{restatable}      
\begin{proof}
    If we use the result from \Cref{lem:outer_update_prox_step}, we know that following holds for $\Strategy{}^*$
    \begin{equation}
        \|\Mean{}^* - \Mean{}^{M, t+1}\|^2 \leq \|\Mean{}^* - \Mean{}^{M, t}\|^2 - \| \Mean{}^{M, t+1} - \Mean{}^{M, t}\|^2
    \end{equation}
This means that $\|\Mean{}^* - \Mean{}^{M, t}\| $ is non-increasing, which gets
\begin{equation}
    \begin{split}
        \sum_{t \geq 0}\| \Mean{}^{M, t+1} - \Mean{}^{M, t}\|^2 &\leq \sum_{t \geq 0} \|\Mean{}^* - \Mean{}^{M, t}\|^2 - \|\Mean{}^* - \Mean{}^{M, t+1}\|^2 \\
        &= \|\Mean{}^* - \Mean{}^{M, 0}\|^2\\
        &\leq D^2 < \infty
    \end{split}
\end{equation} 
The solution set of $\VI(\GameField, \Strategies{}^{\Variance_{\text{min}}})$ is nonempty, because $\Strategies{}^{\Variance_{\text{min}}}$ is compact and convex and $\GameField$ is continuous. \Cref{lem:outer_update_prox_step} holds for every solution of the variation inequality solution $\Strategy{}^*$, so the sequence $\|\Strategy{}^* - \Strategy{}^{M, t}\|$ is non-increasing for all solutions $\Strategy{}^*$. The above bound shows  $\| \Mean{}^{M, t+1} - \Mean{}^{M, t}\|^2 \to 0$.

Since $\Strategy{}^{M, t+1}$ is a solution of $\VI(\GameField^\RegularizerField, \Strategies{}^{\Variance_{\min}})$, it holds
\begin{equation}
    \langle \GameField(\Mean{}^{M, t+1}) + \frac{\RegularizationStrength}{\Variance_{\min}^2} (\Mean{}^{M, t+1} - \Mean{}^{M, t}), \Mean - \Mean{}^{M, t+1} \rangle \geq 0
\end{equation}
Splitting the inner product and applying Cauchy-Schwarz inequality, we get
\begin{equation}
    \begin{split}
        \langle \GameField(\Mean{}^{M, t+1}), \Mean - \Mean{}^{M, t+1} \rangle &\geq -\frac{\RegularizationStrength}{\Variance_{\min}^2}\langle \Mean{}^{M, t+1} - \Mean{}^{M, t}, \Mean - \Mean{}^{M, t+1} \rangle\\
        &\geq -\frac{\RegularizationStrength}{\Variance_{\min}^2} \|\Mean{}^{M, t+1} - \Mean{}^{M, t}\| \|\Mean - \Mean{}^{M, t+1}\| \\
        &\geq -\frac{\RegularizationStrength}{\Variance_{\min}^2} \|\Mean{}^{M, t+1} - \Mean{}^{M, t}\|D
    \end{split}
\end{equation}

The sequence $\Mean^{M, t}$ lies in a compact set, so there is subsequence that starts at time $t_k$ and converges to fixed point $\overline{\Mean}$. Since we know $\|\Mean^{M, t_k+1} - \Mean^{M, t_k}\| \to 0$, we know that even the subsequence starting at $\Mean^{M, t_k+1}$ will converge to the same point. So, if we use $t = t_k$, the right hand side vanishes when $k \to \infty$, so

\begin{equation}
     \langle \GameField(\overline{\Mean{}}), \Mean - \overline{\Mean{}} \rangle \geq 0
\end{equation}
which is the solution of $\VI(\GameField, \Strategies{}^{\Variance_{\min}})$. Applying \cref{lem:outer_update_prox_step} on $\overline{\Strategy{}}$, the sequence $\| \overline{\Mean} - \Mean^{M, t}\|$ is non-increasing and it converges to 0 along wtih subsequence starting at $t_k$. Hence $\| \overline{\Mean} - \Mean^{M, t}\| \to 0$ and the whole sequence converges to $\overline{\Mean}$.
\end{proof}

\begin{restatable}[]{lemma}{GaussianBias}
\label{lem:gaussian_bias}
    Consider a function $f$, which has continuous second derivative with $\sup|f''| \leq M_f$, then $|\E_{x \sim \Gaussian(\Mean, \Variance^2)}[f(x)] - f(\Mean) | \leq \frac{M_f \Variance^2}{2}$ .
\end{restatable}
\begin{proof}
Let $Z \sim \Gaussian(0, 1)$ and $z$ is a single realization of $Z$, then it holds 
    \begin{equation}
        f(\Mean + \Variance z) = f(\Mean) + \Variance z f'(\Mean) + \frac{1}{2} \Variance^2 z^2 f''(\xi_z)
    \end{equation}
where $\xi_z$ is a point between $\Mean$ and $\Mean + \Variance z$. Expectation of this is
\begin{equation}
   \E_{x \sim \Gaussian(\Mean, \Variance^2)}(f(x)) = \E(f(\Mean + \Variance Z) ) =  f(\Mean) + \Variance f'(\Mean) \E(Z) + \frac{1}{2} \Variance^2 \E(Z^2 f''(\xi_Z))
\end{equation}
Because Z is Gaussian with mean 0, then $\E(Z) = 0$. For the second term it holds that
\begin{equation}
    |\frac{1}{2} \Variance^2 \E(Z^2 f''(\xi_Z))| \leq \frac{1}{2} \Variance^2 \E(Z^2 |f''(\xi_Z)|) \leq \frac{1}{2} \Variance^2 M_f \E(Z^2) = \frac{1}{2} \Variance^2 M_f.
\end{equation}

So Putting this together we get

\begin{equation}
    |\E_{x \sim \Gaussian(\Mean, \Variance^2)}[f(x)] - f(\Mean) | \leq | f(\Mean) + \frac{1}{2} \Variance^2 M_f - f(\Mean)| = \frac{1}{2} \Variance^2 M_f
\end{equation}
     
\end{proof}

\begin{restatable}[]{proposition}{GaussianExploitability}
\label{prop:gaussian_expl}
    In a convex-concave game, the exploitability for each strategy $\Strategy{} \in \Strategies{}^{\Variance_{\min}}$ is at most
\begin{equation}
    \Exploitability(\Strategy{}) \leq \max_{\Strategy{1}' \in \Strategies{1}^{\Variance_{\min}}}\ExpectedUtility{}(\Strategy{1}', \Strategy{2}) - \min_{\Strategy{2}' \in \Strategies{2}^{\Variance_{\min}}}\ExpectedUtility{}(\Strategy{1}, \Strategy{2}') + \frac{1}{2}(M_1 + M_2) \Variance^2_{\min}
\end{equation}
\end{restatable}
\begin{proof} Let $\Action{1}^{\text{BR}} \in \argmax_{\Action{1} \in \Actions{1}} \E_{\Action{2} \sim \Strategy{2}} \Utility(\Action{1}, \Action{2})$. We can apply \cref{lem:gaussian_bias} with $M_g = M_1$ to get
    \begin{equation}
        \E_{\Action{2} \sim \Strategy{2}} \Utility(\Action{1}^{\text{BR}}, \Action{2}) \leq \ExpectedUtility{}(\Action{1}^{\text{BR}}, \Strategy{2}) + \frac{M_1 \Variance_{\min}^2}{2} \leq \max_{\Strategy{1}'} \ExpectedUtility{}(\Strategy{1}', \Strategy{2}) + \frac{M_1 \Variance_{\min}^2}{2}
    \end{equation}
Symmetrically we could do the same process for the Player 2, to get the final bound.
\end{proof}

\FinalConvergence*
\begin{proof}
    The proof follows directly from the \cref{thm:outer_convergence_}, which shows that the sequence converges to a solution of $\VI(\GameField, \Strategies{}^{\Variance_{\min}})$ and from \cref{prop:gaussian_expl}, which bounds the exploitability of the strategy profile based on $\Variance_{\min}$.
\end{proof}
\section{Games}
\label{app:games}
In this section, we provide rules for all the games used in the experiments.
\subsection{One-shot games}
\label{app:os_games}
\paragraph{Continuous Matching Pennies} Both players are choosing a real value from the range [-1, 1]. This game contains a pure Nash equilibrium at $\Action{1} = \Action{2} = 0$, but the gradient descent dynamics fail to converge to it.
\begin{equation}
    \Utility_{1}(\Action{1}, \Action{2}) = \Action{1} \Action{2}
\end{equation}
\paragraph{Rotational game} Both players are choosing a real vector from $[-1, 1]^2$ or $[-1, 1]^3$. Similarly to continuous matching pennies, this game contains a pure Nash equilibrium for $\textbf{a}_{1} = \textbf{a}_{2} = \textbf{0}$ to which the gradient descent does not converge. This game shows that the result of \Cref{thm:final_convergence} can hold even beyond convex-concave games. Also it shows that even if the theorem is stated for single-dimension Gaussian, it holds even in multi-dimensional case.
\begin{equation}
    \Utility_{1}(\textbf{a}_{1}, \textbf{a}_{2}) = 20 (\textbf{a}_{1} + 0.3 \textbf{a}_{1}^3)^\top \textbf{A} (\textbf{a}_{2}+ 0.3 \textbf{a}_{2}^3) - \frac{1}{16} (\|\textbf{a}_{1}\|^4 - \|\textbf{a}_{2}\|^4)
\end{equation}
where $\textbf{A} = \begin{pmatrix}
    0 & 1 \\
    -1 & 0
\end{pmatrix}$ for 2D case and $\textbf{A} = \begin{pmatrix}
    0 & 1 & 0 \\
    -1 & 0 & -1 \\
    0 & 1 & 0
\end{pmatrix}$
% All pay Auction
% \begin{equation}
%    \Utility_{\Player}(x, y) = 0.5 \cdot \text{sign}(x - y) - x + y
% \end{equation}
% This game has a Nash equilibrium as a uniform distribution in the interval [0, 1]
\paragraph{Circle game} Both players are choosing a real value from the range [0, 1]. The game is parametrized by the number of harmonics; we have used $K=3$. For any $K$, this game has a Nash equilibrium of size $K+1$, with uniformly spread points on the interval.
\begin{equation}
    \Utility_{1}(\Action{1}, \Action{2}) = \sum_{k = 1}^K 2^{-(k - 1)} \sin(2\pi k (\Action{1} - \Action{2}))
\end{equation}

\paragraph{Glicksberg-Gross game} Both players are choosing a real value from the range [0, 1]. It contains a Nash equilbrium that is $\frac{4}{\pi} \cdot \arctan(\sqrt{t})$, for $t \in [0, 1]$. Again, the mixture of Gaussians cannot converge to this equilibrium with a finite number of Gaussians. Along with the circle game, these demonstrate how the mixture of Gaussians performs in games outside the scope of the algorithm.
\begin{equation}
    \Utility_{\Player}(\Action{1}, \Action{2}) = \frac{(1 + \Action{1})(1 + \Action{2})(1 - \Action{1}\Action{2})}{(1+\Action{1}\Action{2})^2}
\end{equation}
\paragraph{Two-point game} Both players are choosing a real value from the range [-2, 2]. This game has a Nash equilibrium in which both players mix between 30\% playing -1 and 70\% playing 1. This is the best-case scenario for the mixture of Gaussians, as it can approximate these 2 pure strategies very well.
\begin{equation}
\begin{split}
    \Utility_{\Player}(\Action{1}, \Action{2}) &= \exp\Big(-\frac{(\Action{1} + 1)^2}{0.02}\Big) + \exp\Big(-\frac{(\Action{1} - 1)^2}{0.02}\Big)\\
    &- \exp\Big(-\frac{(\Action{2} + 1)^2}{0.02}\Big) - \exp\Big(-\frac{(\Action{2} - 1)^2}{0.02}\Big)\\
    &+ (\Action{1} - 0.4)(\Action{2} - 0.4)
\end{split}
\end{equation}

\subsection{Sequential Games}
\paragraph{Continuous Kuhn Poker} The game is played with the deck of 3 cards, Jack (J), Queen (Q), and King (K). The game starts with each player receiving 1 card and paying an ante of 1. Then player 1 either checks or bets between 0.25 and 2. If player 1 checks, the other player can either check or bet again. After the bet, the other player can either fold and lose their ante or call the bet. If neither player folds, both players reveal their cards, and the player with the higher value wins. The game uses the following order of cards: K $>$ Q $>$ J.

\paragraph{Continuous Leduc Hold'em} The game is played with the deck of 6 cards, which contains 2 Jacks (J), 2 Queens (Q), and 2 Kings (K). The game starts with each player receiving 1 card and placing an ante of 1. The first player then either checks or bets between 0.25 and 2. After that, the players alternate by either folding, which results in a loss, calling the opponent's bet, or raising between 0.25 and 2. The game consists of 2 rounds, and in each, there can be at most 2 raises. The round ends when one of the players calls the other one, and between the first and second rounds, the public card is revealed. If neither player folds, then at the end of round 2, both show their hands, and either the player who matches the public card wins or the player with the higher-value card, with K $>$ Q $>$ J.

\paragraph{Sequential Blotto N} A sequential game over $N$ turns, where players are assigning resources simultaneously to the same base. The player who uses more resources "wins" that base. At the end of the turn, only the players who won the base are revealed, not how many resources they used. The player who wins more bases in the $N$ turns wins the game.

\paragraph{Heads-up No-limit Texas Hold'em} The game is played with the standard 52-card deck. Each player starts with 400 chips and receives 2 cards at the beginning of the game. The first player (small blind) makes an ante of 1 chip, while the second player (big blind) makes an ante of 2 chips. The game is played in 4 rounds. After the first round, 3 public cards are revealed; after the second round, 1 public card is revealed; and after the third round, the last public card is revealed. In the round, players alternate in betting until one of them calls. Either player can bet all their chips at any time, forcing the opponent to call or fold. At the end of the game, the players reveal their private cards, and the player with the highest combined value of their private and public cards wins. Even if Texas Hold'em has discrete bets, we treat them as continuous. So the algorithm samples a real-valued bet, and we round it to the nearest integer, which is then used in the game.

\section{Experimental details}
\label{app:exp_details}

In this section, we provide the hyperparameters for all algorithms across all experiments, along with the computational resources used for each experiment.% We also provide details about the specific implementations of the algorithms we used.
\subsection{Proximal policy optimization}
We use Proximal Policy Optimization (PPO) as one of the baselines, as an algorithm to compute the best response for evaluation, and as an oracle in PSRO and NFSP. We use Gaussian reparametrization for the PPO, where we train just a single Gaussian head through the following loss function
\begin{equation}
        \begin{split}
        \PPOSurrogate_{\text{PPO}} &= \min\Big( \frac{\Strategy{\Neural}(\Action{k}^t | \State^t, \Gaussian_{k}^t)}{\Strategy{\text{old}}(\Action{k}^t | \State^t, \Gaussian_{k}^t)}\Advantage (\State^t, \Action{k}^t), \text{clip}\big( \frac{\Strategy{\Neural}(\Action{k}^t | \State^t, \Gaussian_{k}^t)}{\Strategy{\text{old}}(\Action{k}^t | \State^t, \Gaussian_{k}^t)}, 1 - \epsilon_{{\text{PPO}}}, 1+\epsilon_{{\text{PPO}}}\big)\Advantage (\State^t, \Action{k}^t) \Big)\\ 
        \Loss_{\text{PPO}} &= -\E_t\PPOSurrogate_\Gaussian - \Entropy(\Gaussian_{k}^t) 
        \end{split}
\end{equation}
To train the critic, we use Generalized Advantage Estimate (GAE) \citep{schulman2015gae}. Due to the similarity of PPO to the mixture of Gaussians, we include the hyperparameters in \Cref{tab:ppo_hyper}
\subsection{Magnetic Mirror Descent}
Magnetic Mirror Descent is a modification to discrete PPO that introduces the KL divergence between the current strategy and the magnet strategy into the loss. In our mixture-of-Gaussians reparametrization, we use the MMD loss to train the categorical distribution $\Loss_\CategoricalStrategy$. As a baseline, we use MMD also as a solver for the discretized game. The discrete bins we use are uniformly spread through the action space. All the MMD hyperparameters are included in \Cref{tab:ppo_hyper}
\begin{table}[]
    \centering
    \begin{tabular}{l|c|c|c|c|c|c|c}
         & \multicolumn{3}{c|}{\textbf{One-shot games}} & \multicolumn{4}{c}{\textbf{Sequential games}} \\
         & \textbf{PPO} &  \textbf{MMD} & \textbf{MMPO} & \textbf{PPO} &  \textbf{MMD} & \textbf{MMPO} & \textbf{BR-PPO} \\ \hline
         PPO epochs & 2 & 2 & 2 & 1 & 1 & 1 & 1 \\
         Entropy weight & 0.05 & 0.05 & 0.05 & 0.05 & 0.02 & 0.02 & 0.01 \\
         Magnet weight $\RegularizationStrength$ & - & 0.2 & 0.2  & - & 0.2 & 0.2 & -  \\
         Magnet update & - & 500 & 500 & - & 1000 & 1000  & - \\
         Minimal std $\Variance_{\text{min}}$ & - & - & 0.001 & - & - & 0.1 & 0.001\\
         Components $K$ & 1 & - & * & 1 & - & 4 & 1\\
         Exploration $\epsilon$ & 0.0 & 0.0  & 0.0& 0.2 & 0.2  & 0.2 & 0.0  \\
         Value loss weight & 0.5 & 0.5 & 0.5 & 0.5 & 0.5 & 0.5 & 0.5 \\
         GAE lambda $\lambda$ & - & - & -  & 0.95 & 0.95 & 0.95 & 0.95 \\
         V-trace clipping $\overline{\rho}$ & - & - & -  & 2 & 2 & 2 & 2 \\
         V-trace clipping $\overline{c}$ & - & - & - & 1 & 1 & 1 & 1   \\
         Hidden sizes & (64,64)& (64,64)& (64,64)& (64,64)& (64,64)& (64,64) & (64,64) \\
         Activation & GELU & GELU & GELU & GELU & GELU & GELU & GELU\\
         Normalization & RMS & RMS & RMS & RMS & RMS & RMS & RMS\\
         Optimizer & Adam & Adam & Adam & Adam & Adam & Adam & Adam  \\
         Learning Rate $\LearningRate$ & $10^{-3}$ & $10^{-3}$ & $10^{-3}$ & $10^{-3}$ & $10^{-3}$ & $10^{-3}$  & $3\cdot 10^{-3}$ \\
         Max grad norm & 0.5 & 0.5 & 0.5 & 0.5 & 0.5 & 0.5 & 0.5 \\
         Target update $\tau$  & $10^{-3}$ & $10^{-3}$ & $10^{-3}$ & $10^{-3}$ & $10^{-3}$ & $10^{-3}$ & $10^{-3}$ \\
         Batch size & 256 & 256 & 256  & 512 & 512 & 512 & 512\\
    \end{tabular}
    \caption{Hyperparameters used in experiments for different policy-gradient algorithms and also for computing the approximate exploitability. The amount of components in one-shot games depend on game, in matching pennies it is 1, and in other games it is 4.}
    \label{tab:ppo_hyper}
\end{table}
\subsection{Neural fictitious self-play}
Neural fictitious self-play stores two strategy networks at each time, a BR strategy and an average strategy. The algorithm works in two alternating phases: in the first, it trains the best response against every strategy; then, in the second, it adds some trajectories sampled by the BR into the replay buffer, which is then used to train the average strategy. Compared to the vanilla NFSP, we use PPO to compute the best response and not DQN. We provide hyperparameters in \Cref{tab:br_hyper}
\subsection{Policy space response oracles}
The training of Policy space response oracles is also split into two phases. In the first phase, both players train best response against the meta strategy. Then this BR is added to each player's portfolio, and the expected utility of these strategies against all other strategies in the portfolio is computed. Then the resulting matrix game is solved, and the solution is used as the new meta-strategy, to which both players compute their best responses. Unlike the NFSP, the number of strategies PSRO stores grows with the training. The hyperparameters of PSRO are in \Cref{tab:br_hyper}

\begin{table}[]
    \centering
    \begin{tabular}{l|c|c|c|c}
         & \textbf{NFSP} &  \textbf{PSRO} &  \textbf{NFSP} &  \textbf{PSRO}  \\ \hline
         BR train steps & $10^3$ & $10^3$ & $10^4$ & $10^4$  \\
         BR batch size & 64 & 64 & 64 & 64 \\
         PSRO payoff samples & - & 256 & - & 20000 \\
         NFSP average steps & 400 & - & 1000 & - \\
         Reservoir capacity & $10^6$ & - &  $10^6$ & -\\
         NFSP reservoir batch & 256 & - & 256 & -\\
         NFSP anticipatory $\eta$ & 0.1 & - & 0.1 & - \\
    \end{tabular}
    \caption{Hyperparameters used in experiments for different best-response based algorithms. Training of the best response uses PPO hyperparameters from \Cref{tab:ppo_hyper}, besides a different batch size.}
    \label{tab:br_hyper}
\end{table}
\subsection{Randomized policy networks}
Unlike other techniques, randomized policy networks do not explicitly produce the strategy. Instead, they store it implicitly in the network's weights. The network takes random noise at the input and outputs the sampled action. This makes RPN similar to diffusion models, which perform a single step. Unlike diffusion models, which use supervised training, the randomized policy networks are used in reinforcement learning. Since the strategy is not parametrized in any way, the network can learn an arbitrary probability distribution given its architecture. We provide all hyperparameters in \Cref{tab:rpn_hyper}

\begin{table}[]
    \centering
    \begin{tabular}{l|c|c|c|c}
        Optimizer & Optimistic GD\\
        Optimism & 0.333 \\
        Learning rate & $10^{-3}$\\
        Max grad norm & 5 \\
        Batch size & 256 \\
        Noise dimension & 16 \\
        Activation & Mish \\
        Normalization & RMS\\
        Extragradient step & 0.01\\
    \end{tabular}
    \caption{Randomized policy networks hyperparameters used in one-shot games }
    \label{tab:rpn_hyper}
\end{table}

\subsection{Local best response}
Local best response (LBR) is a Poker-specific technique that estimates a lower bound on the exploitability of the strategy $\Strategy{}$. Let us describe it from the perspective of evaluating Player 1's strategy, so the local best response is Player 2. When Player 2 is about to make a decision, it computes the exact probabilities for each hand of Player 1. The local best response uses action abstraction; we have used fold, call, pot, all-in abstraction. For each action, local best-response estimates its value by assuming both players would only check from this point to the end, and then pick the best one. LBR is more effective if it lets the opponent make more mistakes and therefore, we use the setting where it always checks and calls in the first two betting rounds of the game and only estimates action values and plays the best of them in turn and river.

\subsection{Computational resources}
All of our experiments used Python and JAX as a machine learning framework \citep{jax2018}. Every experiment except Heads-up no limit Texas hold'em used a single CPU core of AMD EPYC 7543 with 16 GB of RAM and 16Gb of memory. In \Cref{fig:single_ppo,fig:one_shot,fig:seq,fig:single_sampled,fig:one_shot_capacity,fig:seq_capacity} we report means and 95\% confidence intervals computed from 5 different random seeds. Experiment in \Cref{fig:single_ideal} do not use an any randomness.
\paragraph{Convergence of the regularized Gaussian Head} All of the experiments in this section took less than 4 hours using when ran sequentially.
\paragraph{One-shot games} Single training run for each game took less than 10 minutes. We have used 5 baselines, 5 seeds, 4 games. Altogether this makes the experiment run 16 hours.
\paragraph{Sequential games} Single training run for each game took less than 4 hours in \cref{fig:seq}. However, the approximate exploitability computation took up to 12 hours. In total for each of the 4 baselines, 5 seeds and 4 games, the experiment took $4\cdot4\cdot5 + 3\cdot(4+12)\cdot5\cdot4=1040$ hours.
\paragraph{Capacity experiments} For one-shot games we have used the same limit, for sequential games, we have used longer time-frame for training, when number of components and bins was 16 or greater. Similarly for those runs the best response took up to 24 hours. In total this experiment took $64+11\cdot5\cdot4 + 7\cdot5\cdot3\cdot(4+12) + 4\cdot5\cdot3\cdot(12+24) = 4124$ hours.
\paragraph{Heads-up no limit Texas hold'em} This is the only experiment where we have used GPUs, and we have trained the strategy for 4 days on H200. Then for evaluation, each checkpoint comparison against Slumbot took roughly 120 CPU-hours, while the LBR computation took roughly 12 CPU hours per checkpoint. Altogether this gives $8976$ CPU-hours for the final evalutaion.
\section{Additional experiments}
\label{app:add}
We have run two additional experiments to evaluate different aspects of the \AlgorithmShort{}. First, we show that the \AlgorithmShort{} converges in the same games as in \cref{sec:exp_gaussian} when the neural network does not parametrize the strategy but the gradients are estimated from the samples. Second, we show how the strategy's exploitability varies with the number of Gaussians in the model. A different perspective on \AlgorithmShort{} is that it trains the abstraction during training. We also compare \AlgorithmShort{} with uniform discretization of the action space using the same number of components.
\subsection{Regularized gaussian head}
\label{app:add_gradient}
We have run the same experiment as in \Cref{sec:exp_gaussian}, but now we do not store the strategy in the neural network, but we use samples to estimate the gradient. This serves as a middle step between the \Cref{fig:single_ideal} and \Cref{fig:single_ppo}. We have used a learning rate of 0.05 and a batch size of 256. The results of this experiment are in \Cref{fig:single_sampled}. Even these results confirm that without a magnet, the gradient descent diverges to the most exploitable strategies. With the magnet, the strategies do not converge arbitrarily close to the Nash equilibrium because convergence is limited by noise in gradient estimation. Still, in all games, the exploitability is roughly $10^{-2}$, which is several orders of magnitude better than without a magnet.

We also show how the strategies in continuous matching pennies evolve during training across the three training regimes in \Cref{fig:conv_mp}. It also confirms that the strategies gradually diverge without the magnet regularization.

\begin{figure*}[!h]
    \centering    
    \begin{subfigure}[b]{0.32\textwidth}
        \centering
    \includegraphics[width=0.99\linewidth]{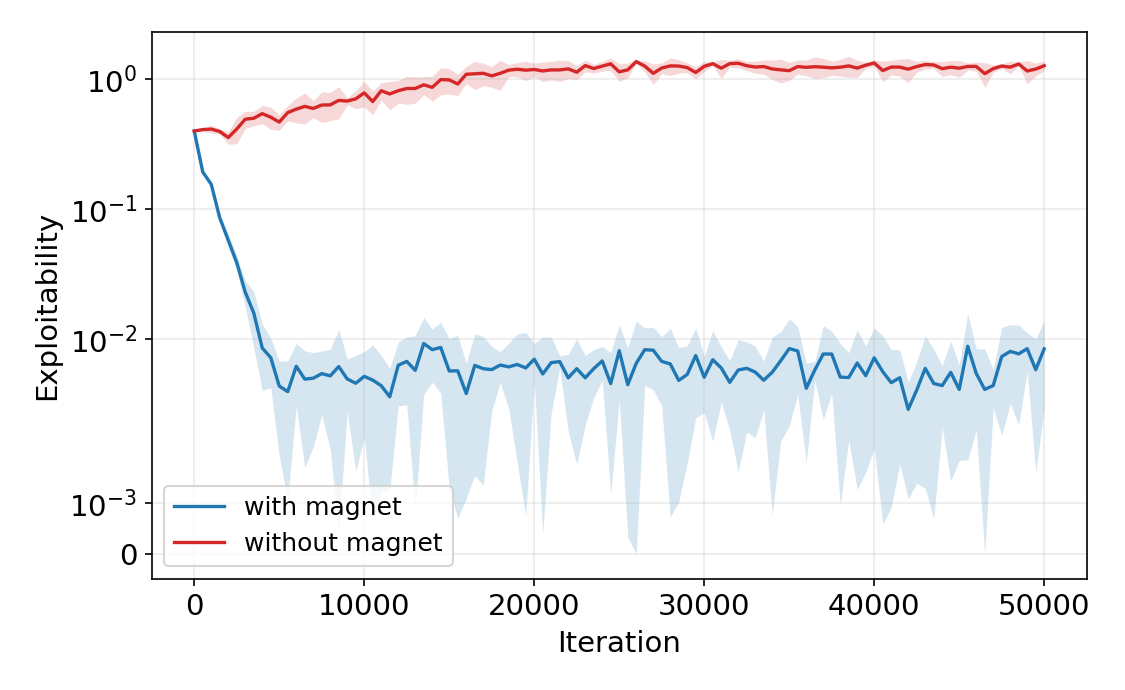}
        \caption{Continuous matching pennies}
        \label{fig:single_sampled_mp}
    \end{subfigure}
    \hfill % Adds horizontal space between subfigures
    \begin{subfigure}[b]{0.32\textwidth}
        \centering
    \includegraphics[width=0.99\linewidth]{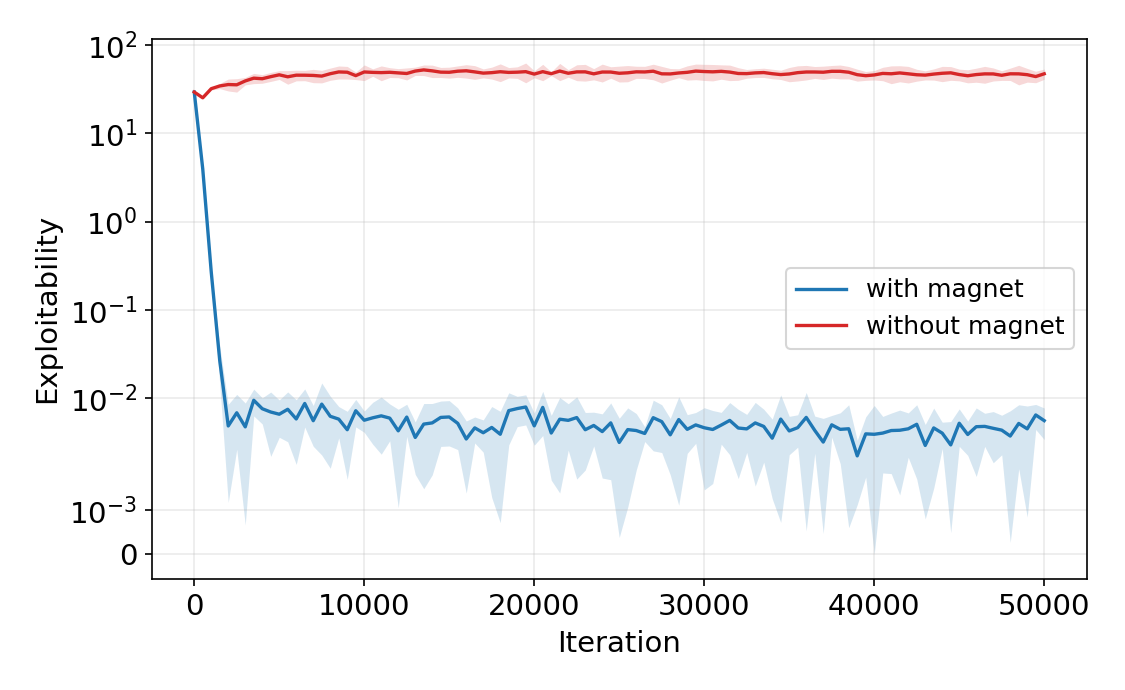}
        \caption{Rotational game, dim=2}
        \label{fig:single_sampled_rot2}
    \end{subfigure} 
    \hfill % Adds horizontal space between subfigures
    \begin{subfigure}[b]{0.32\textwidth}
        \centering
    \includegraphics[width=0.99\linewidth]{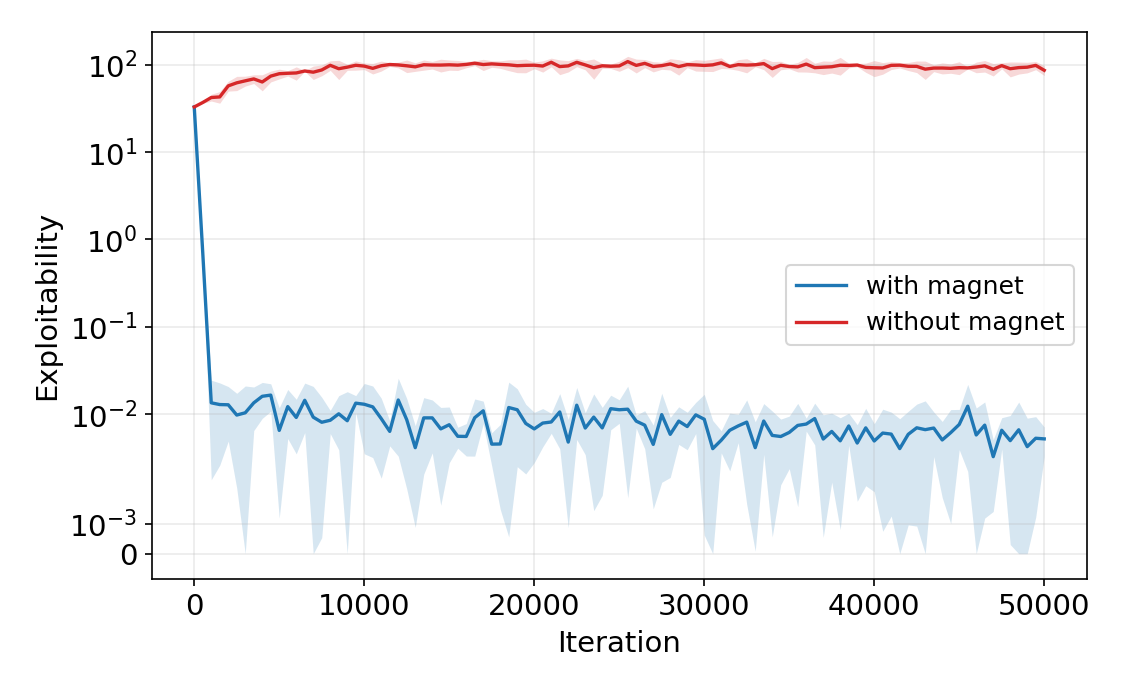}
        \caption{Rotational game, dim=3}
        \label{fig:single_sampled_rot3}
    \end{subfigure} 
    \caption{Exploitability in one-shot games with pure Nash equilibria using gradient descent with or without magnet, using only the gradient estimate.}
    \label{fig:single_sampled} 
\end{figure*}  

\begin{figure*}[!h]
    \centering    
    \begin{subfigure}[b]{0.32\textwidth}
        \centering
    \includegraphics[width=0.99\linewidth]{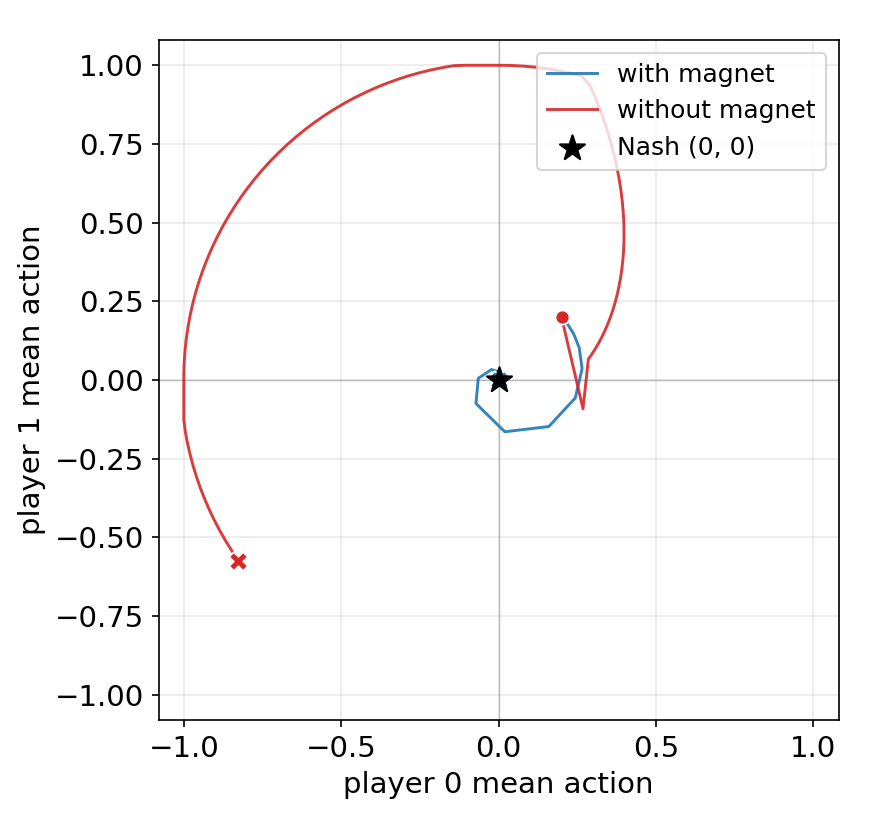}
        \caption{Exact gradient}
        \label{fig:conv_mp_ideal}
    \end{subfigure}
    \hfill % Adds horizontal space between subfigures
    \begin{subfigure}[b]{0.32\textwidth}
        \centering
    \includegraphics[width=0.99\linewidth]{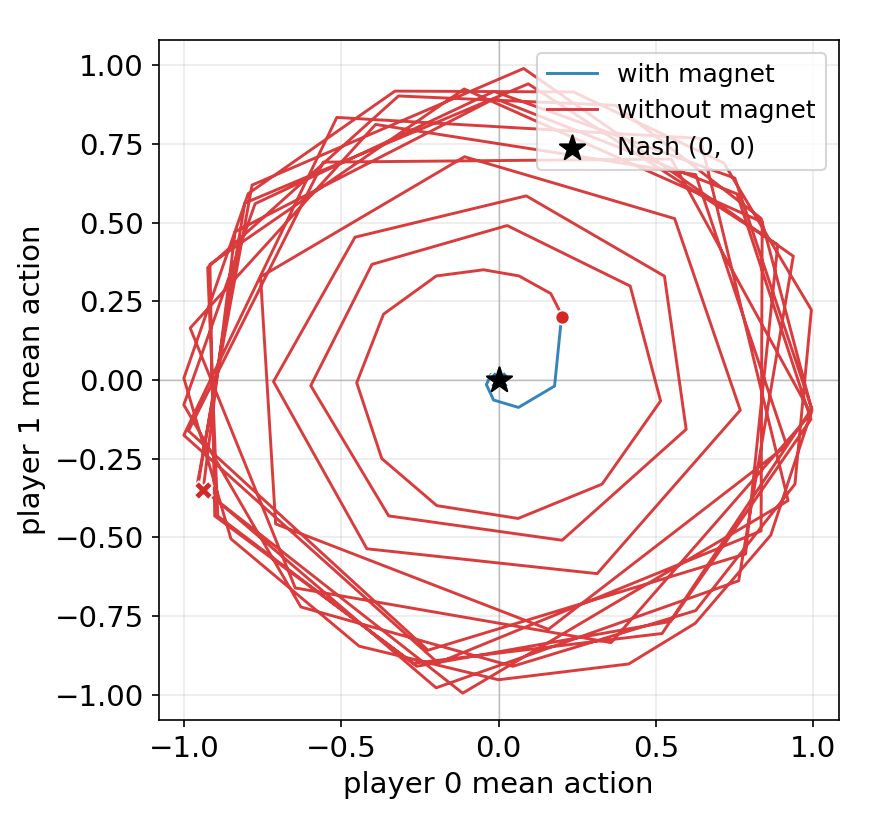}
        \caption{Estimated gradient}
        \label{fig:conv_mp_sampled}
    \end{subfigure} 
    \hfill % Adds horizontal space between subfigures
    \begin{subfigure}[b]{0.32\textwidth}
        \centering
    \includegraphics[width=0.99\linewidth]{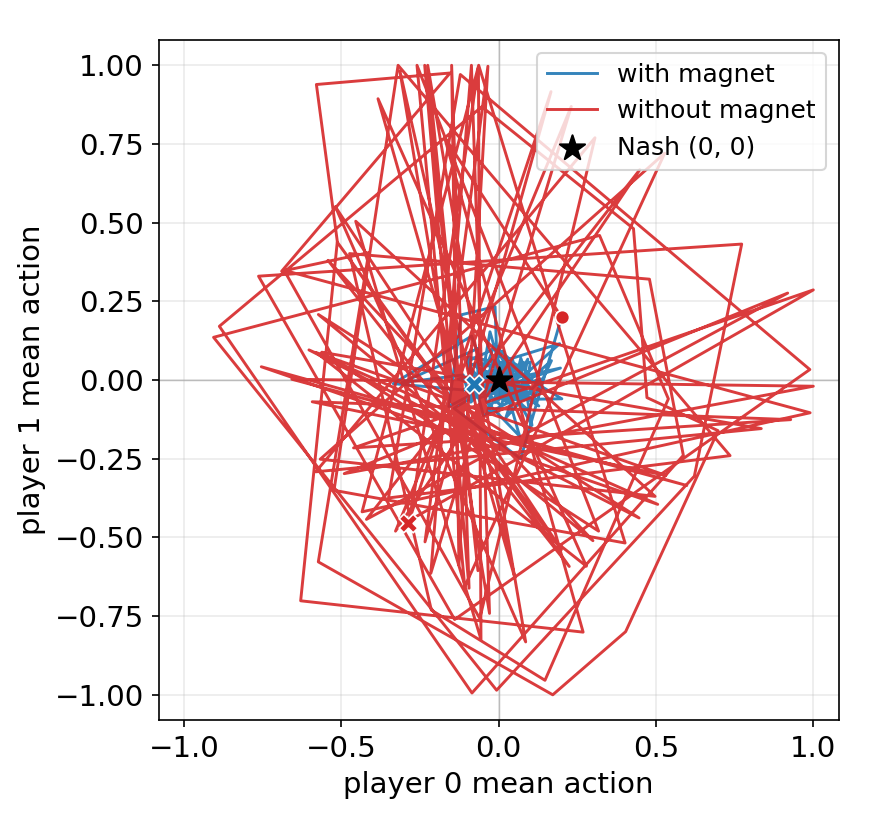}
        \caption{Estimated gradient, neural}
        \label{fig:conv_mp_ppo}
    \end{subfigure} 
    \caption{The evolution of a strategy in continuous matching pennies during training for each of the experimental setting}
    \label{fig:conv_mp} 
\end{figure*}  

\subsection{Wall-time comparison}
\label{app:walltime}
In \Cref{fig:one_shot,fig:seq} we compared the algorithms in terms of environment interactions. Here, we also provide the comparison of the same runs, but when using wall-time for comparison. The results are in \Cref{fig:one_shot_combined,fig:seq_combined}. Both the environment interactions and wall-time agree with the relative performance of the algorithms, which show that \AlgorithmShort{} is as fast as PPO, but converges much quicker than NFSP and PSRO.

\begin{figure*}[!h]
    \centering    
    \begin{subfigure}[b]{0.49\textwidth}
        \centering
    \includegraphics[width=0.99\linewidth]{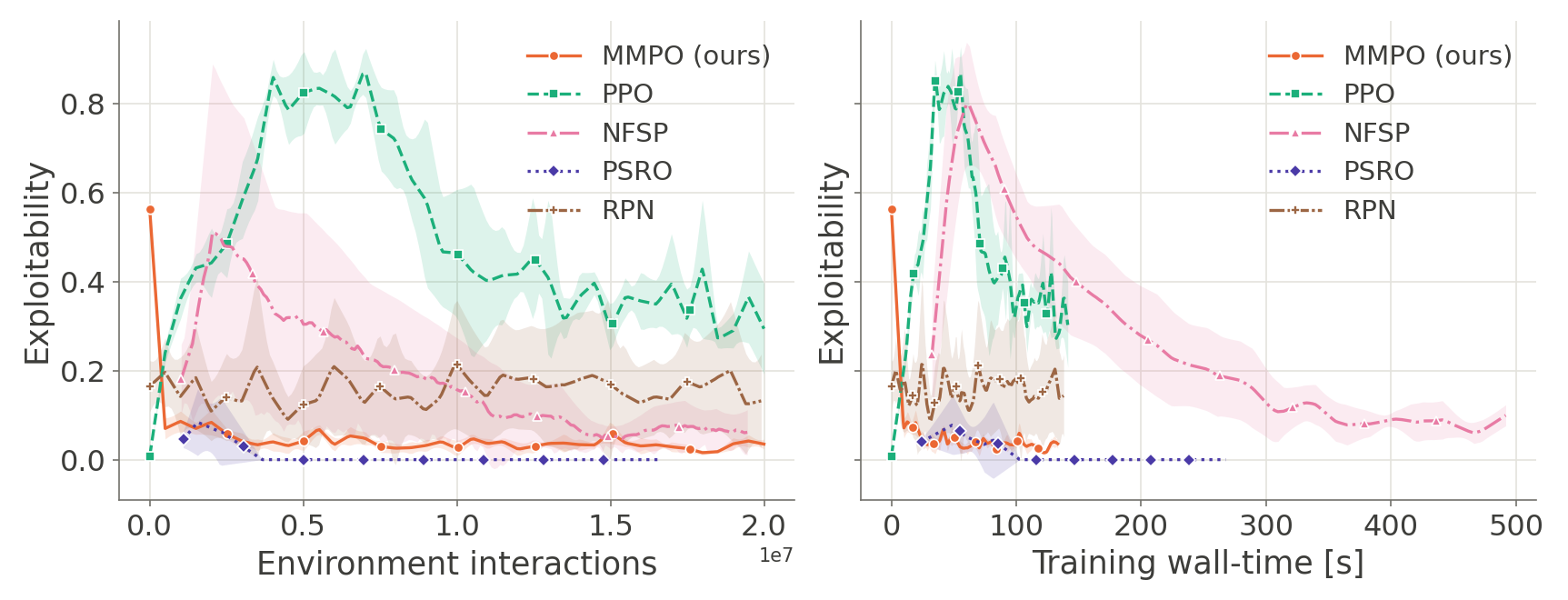}
        \caption{Continuous Matching Pennies}
        \label{fig:one_shot_combined_mmd}
    \end{subfigure}
    \hfill % Adds horizontal space between subfigures
    \begin{subfigure}[b]{0.49\textwidth}
        \centering
    \includegraphics[width=0.99\linewidth]{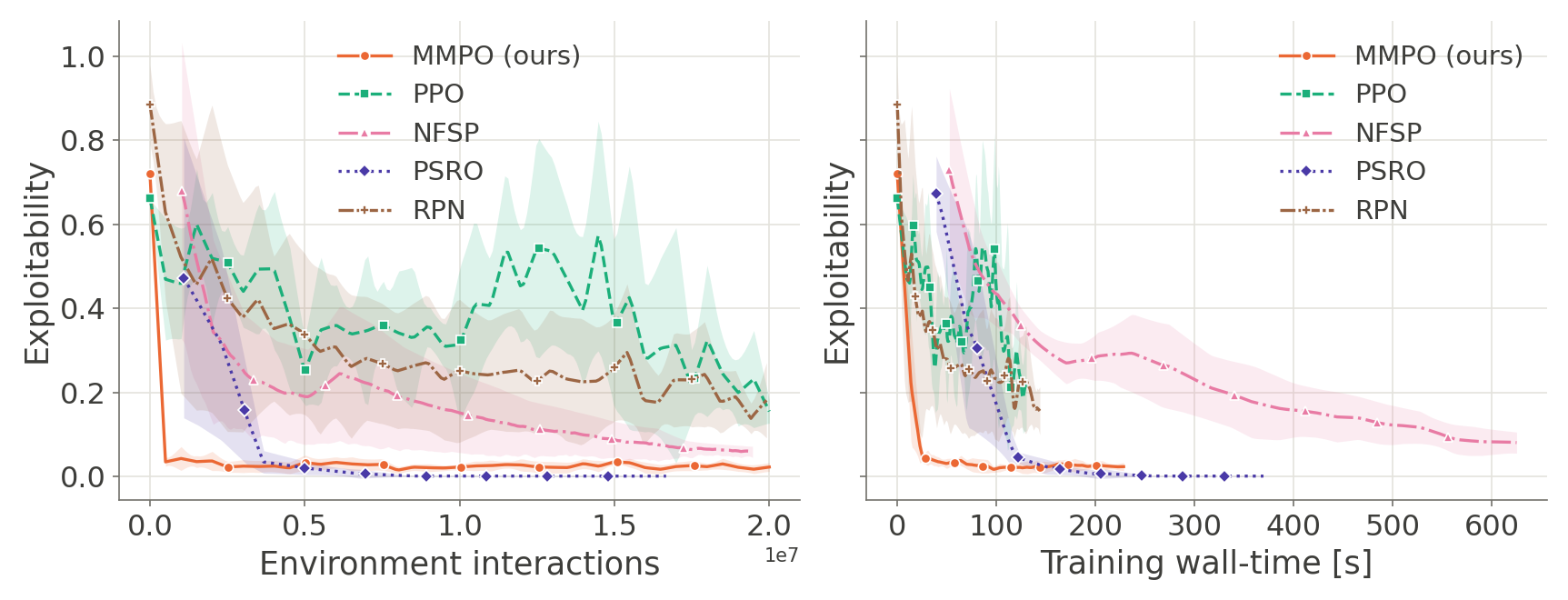}
        \caption{Glicksberg-Gross game}
        \label{fig:one_shot_combined_glicks}
    \end{subfigure}  
    \begin{subfigure}[b]{0.49\textwidth}
        \centering
    \includegraphics[width=0.99\linewidth]{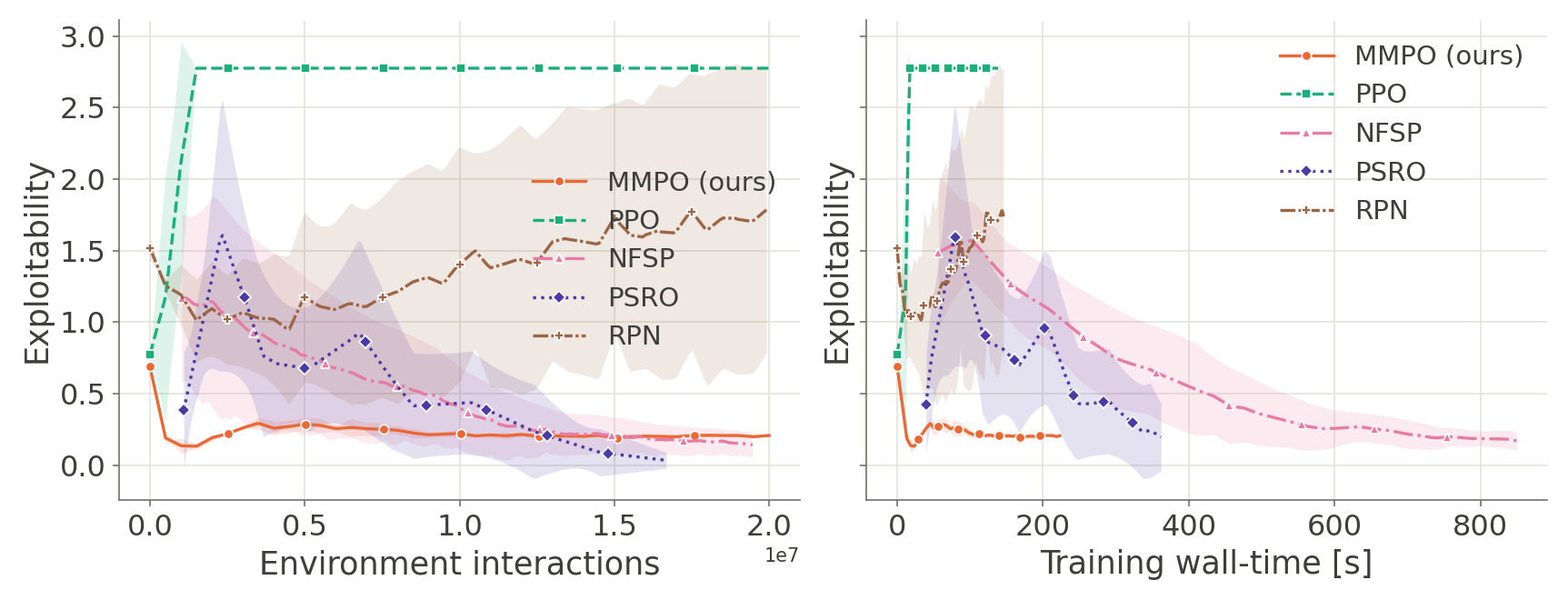}
        \caption{Circle game}
        \label{fig:one_shot_combined_circle}
    \end{subfigure}
    \hfill % Adds horizontal space between subfigures
    \begin{subfigure}[b]{0.49\textwidth}
        \centering
    \includegraphics[width=0.99\linewidth]{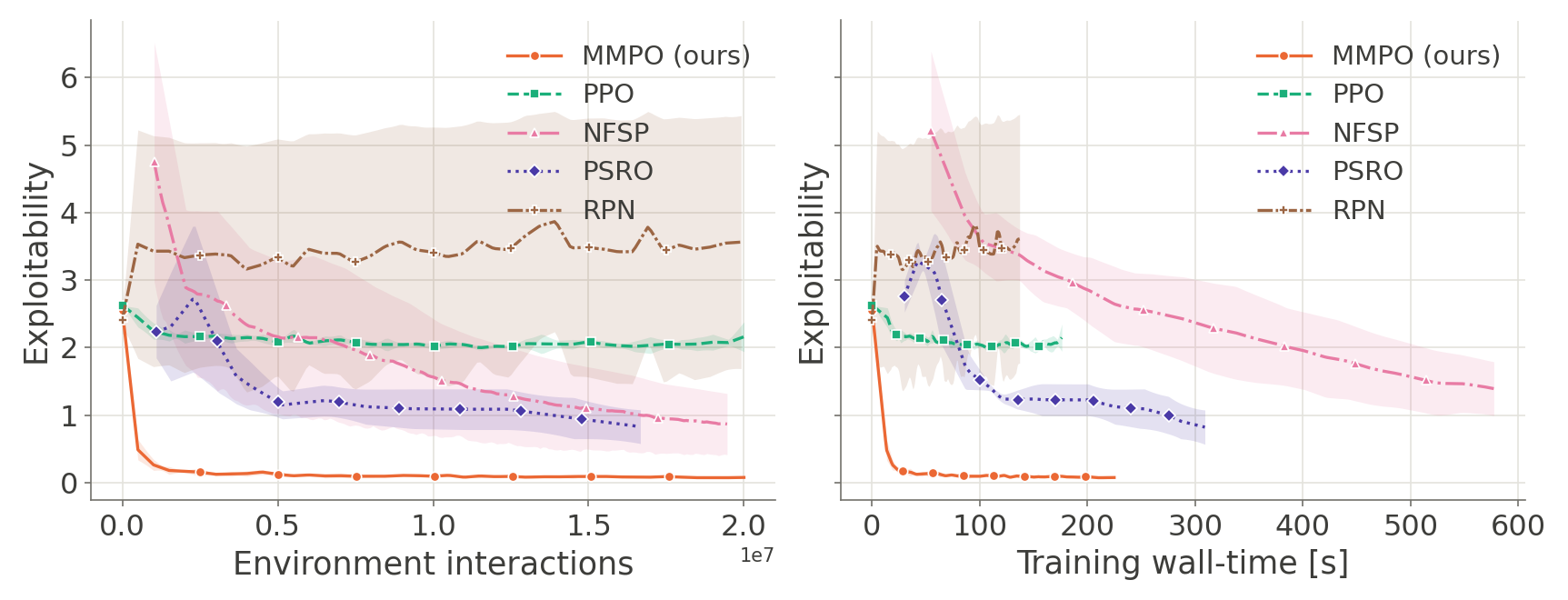}
        \caption{Two-point game}
        \label{fig:one_shot_combined_two}
    \end{subfigure}  
    \caption{Exploitability with 95\% confidence intervals of different baselines in one-shot games}
    \label{fig:one_shot_combined} 
\end{figure*}  

\begin{figure*}[!h]
    \centering    
    \begin{subfigure}[b]{0.49\textwidth}
        \centering
    \includegraphics[width=0.99\linewidth]{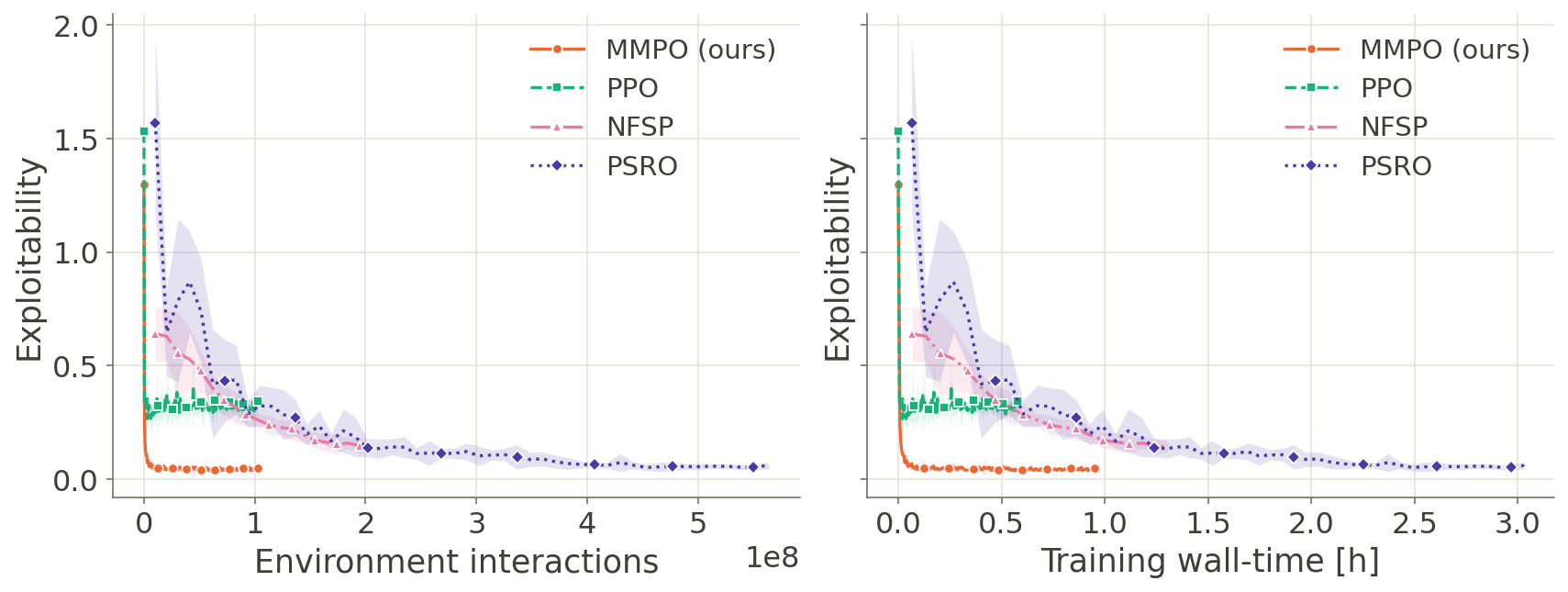}
        \caption{Kuhn Poker}
        \label{fig:seq_combined_kuhn}
    \end{subfigure}
    \hfill % Adds horizontal space between subfigures
    \begin{subfigure}[b]{0.49\textwidth}
        \centering
    \includegraphics[width=0.99\linewidth]{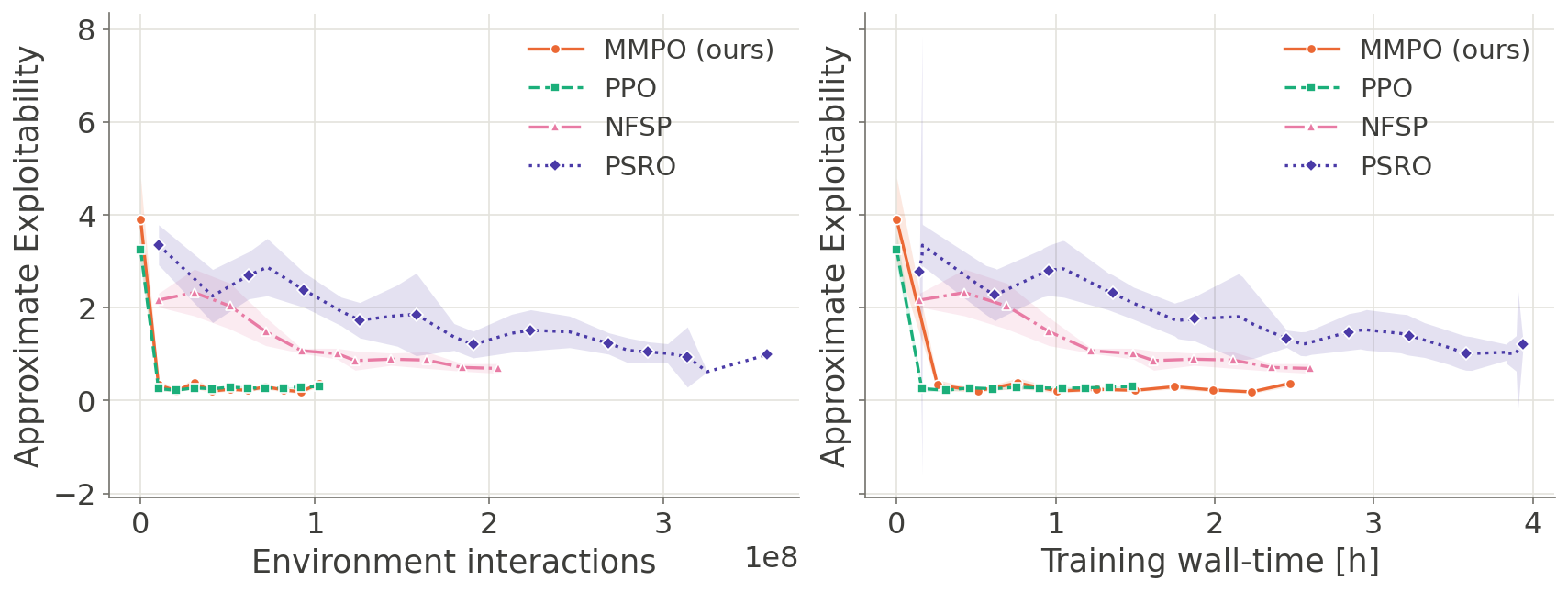}
        \caption{Leduc Poker}
        \label{fig:seq_combined_leduc}
    \end{subfigure}  
    \begin{subfigure}[b]{0.49\textwidth}
        \centering
    \includegraphics[width=0.99\linewidth]{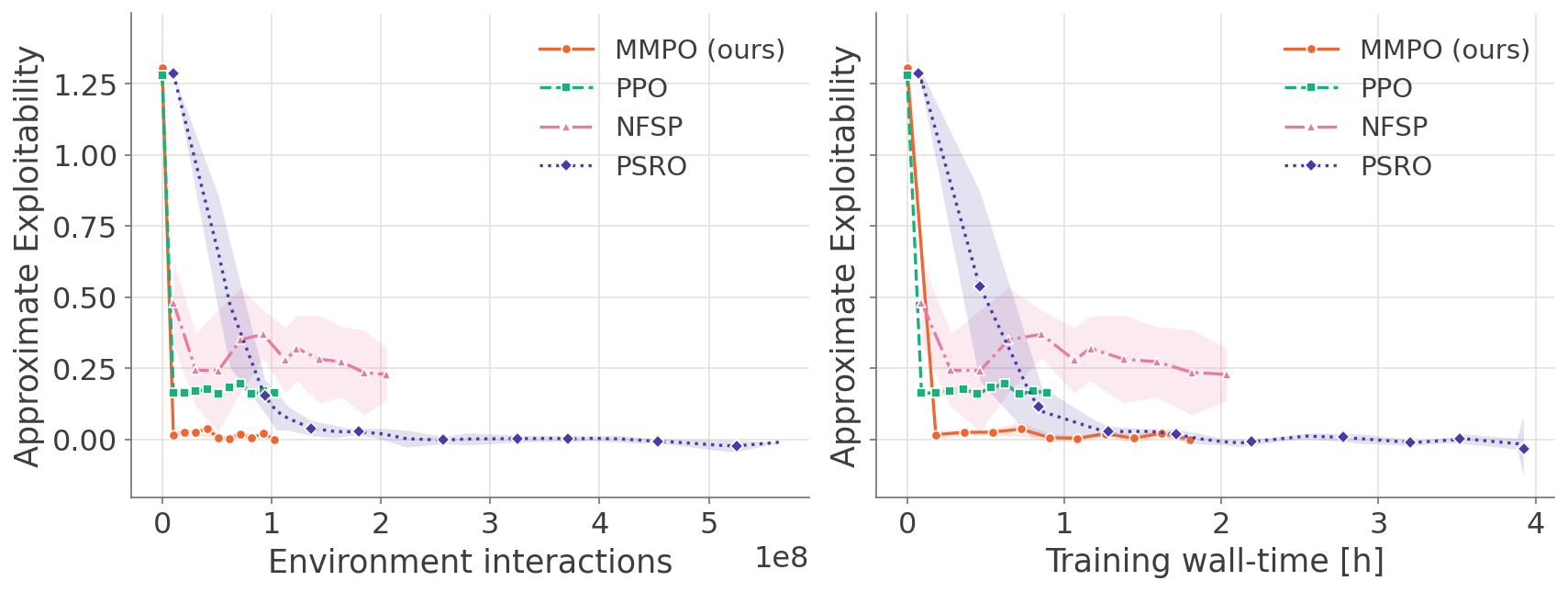}
        \caption{Sequential Blotto 3}
        \label{fig:seq_combined_blotto3}
    \end{subfigure}
    \hfill % Adds horizontal space between subfigures
    \begin{subfigure}[b]{0.49\textwidth}
        \centering
    \includegraphics[width=0.99\linewidth]{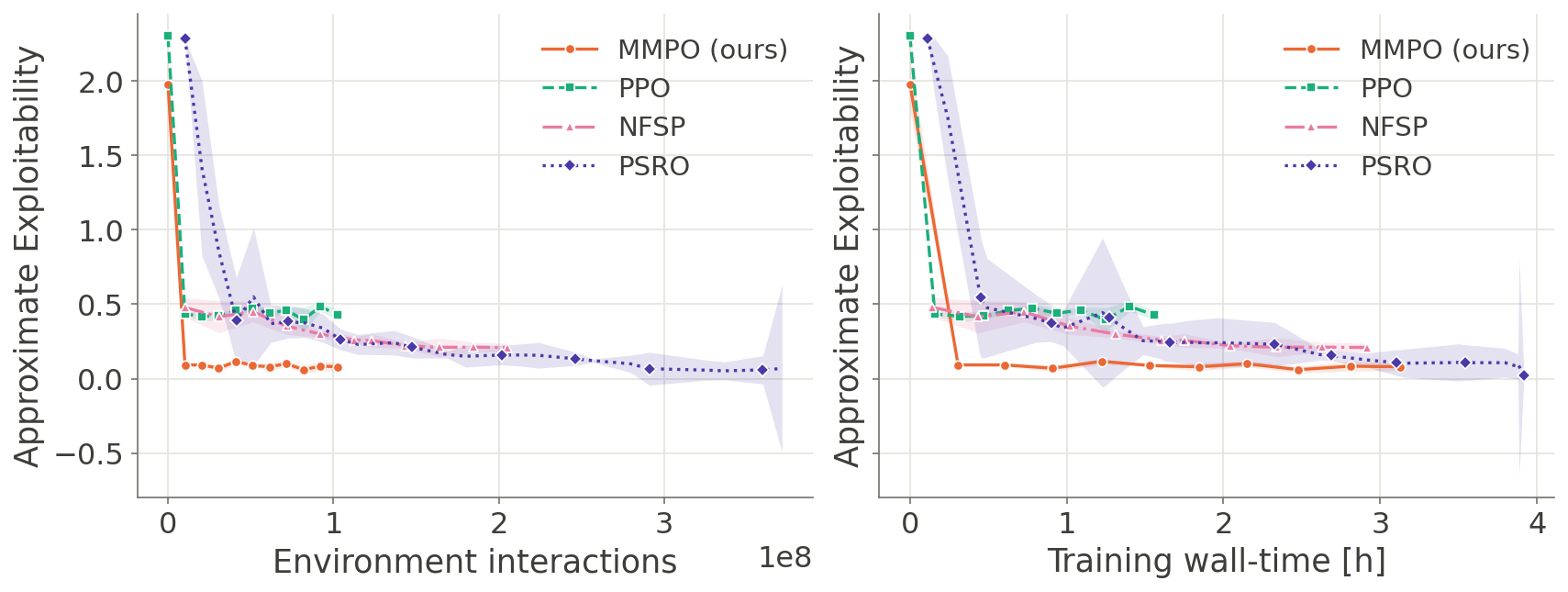}
        \caption{Sequential Blotto 5}
        \label{fig:seq_combined_blotto5}
    \end{subfigure}
    \caption{Approximate exploitability with 95\% confidence intervals of different baselines in sequential games}
    \label{fig:seq_combined} 
\end{figure*}  

\subsection{Effect of the discretization and amount of components in the mixture}
\label{app:add_components}
The strategy space of \AlgorithmShort{} is limited by the number of Gaussians used in the mixture. We have run another experiment that shows how the final exploitability changes when we use different numbers of components. We also compared it with the alternative approach that uniformly discretizes the action space and then uses MMD in the discretized game. We have used the same games as in \cref{sec:exp_one_shot,sec:exp_sequential}.

The results for one-shot games are in \Cref{fig:one_shot_capacity}. The discretization works well in matching pennies, because even if the game contains a pure Nash equilibrium, any strategy that has a mean of 0 is also a Nash equilibrium, so a uniform discretization that is symmetrical around 0 can always find a Nash equilibrium. In contrast, the noise from the Gaussian sampling in \AlgorithmShort{} limits how close to equilibrium the \AlgorithmShort{} can converge, regardless of the number of components.

In the Glicksberg-Gross game, the discretized version performs the best with 6 bins, and adding more bins degrades the performance. Similar behavior can be observed in some sequential games in \Cref{fig:seq_capacity}. This happens because of entropy regularization, which slightly skews the distribution toward a uniform distribution. This highlights another advantage of the \AlgorithmShort{}, because even if the categorical distribution is skewed toward playing all actions with some nonzero probability, the Gaussian can shift to a place where this probability is better utilized. This means that \AlgorithmShort{} is less sensitive to the entropy setting.

A two-point game presents the worst-case scenario for the discretization: when the two pure actions of the equilibrium are not in the discretization, the discretized version will never converge arbitrarily close, as we can see when we increase the number of discrete bins. On the other hand, the \AlgorithmShort{} can reliably converge with three Gaussian components. Two components should be sufficient for convergence, but \AlgorithmShort{} is sensitive to their initialization. Specifically, in this game where the two pure actions are -1 and 1, if all components are initialized greater than 1, then only very strong exploration would be able to move one of the components to -1. This is not specific to our method, but to any method that uses gradient descent to move pure actions. In practice, it may be possible to detect these collapsed actions, reinitialize the $\Mean$ and $\Variance$ for one of them, and resume training, but we decided to avoid this to keep the setting as simple as possible.

Even if adding more components broadens the set of strategies the algorithm can converge to, it also slows training, as more samples are required to learn more actions. As such, the optimal amount of components can differ based on the available budget.

\begin{figure*}[!h]
    \centering    
    \begin{subfigure}[b]{0.24\textwidth}
        \centering
    \includegraphics[width=0.99\linewidth]{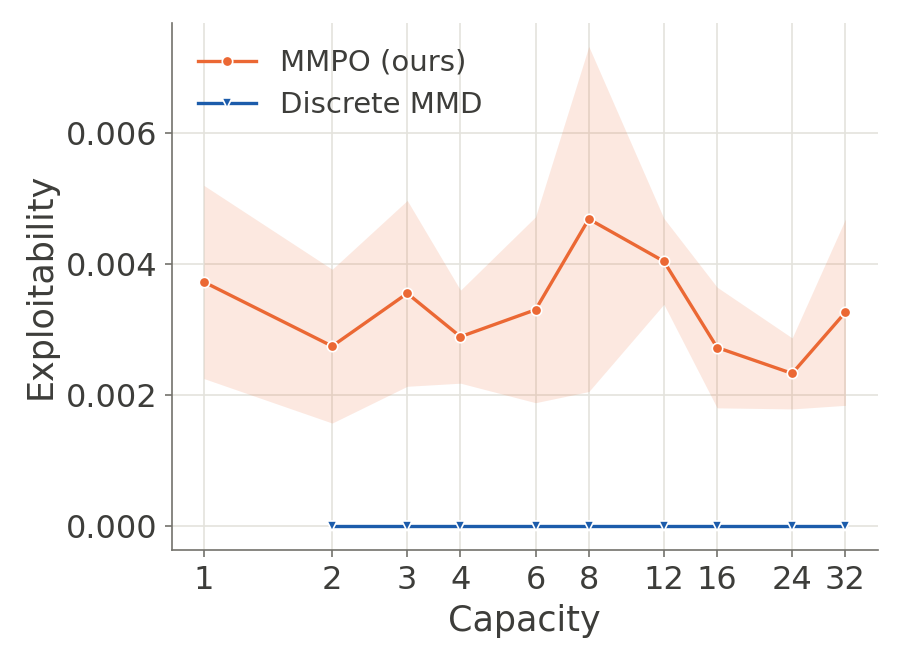}
        \caption{Matching Pennies}
        \label{fig:one_shot_mp_capacity}
    \end{subfigure}
    \hfill % Adds horizontal space between subfigures
    \begin{subfigure}[b]{0.24\textwidth}
        \centering
    \includegraphics[width=0.99\linewidth]{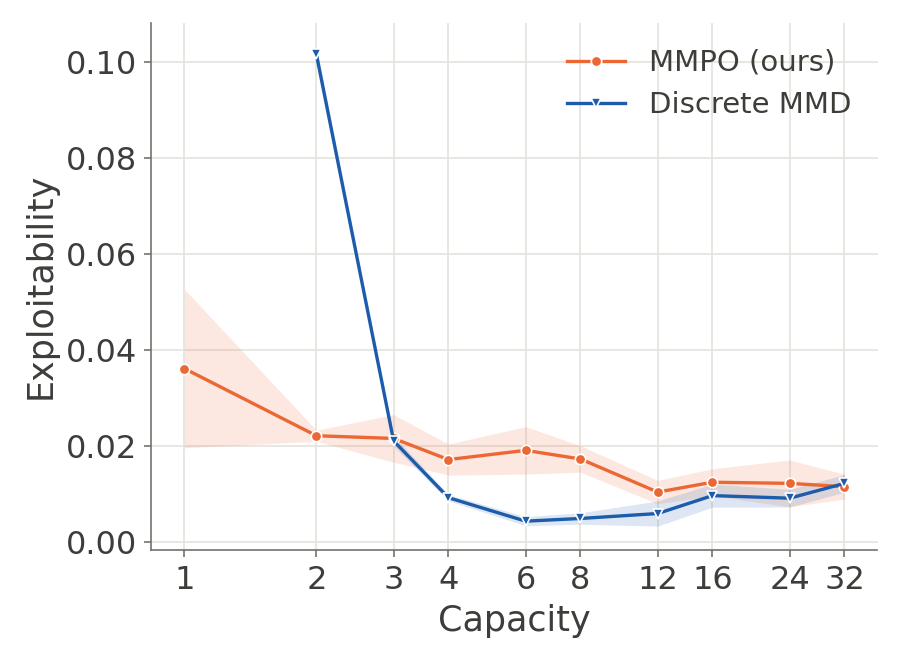}
        \caption{Glicksberg-Gross}
        \label{fig:one_shot_glicks_capacity}
    \end{subfigure}  
    \begin{subfigure}[b]{0.24\textwidth}
        \centering
    \includegraphics[width=0.99\linewidth]{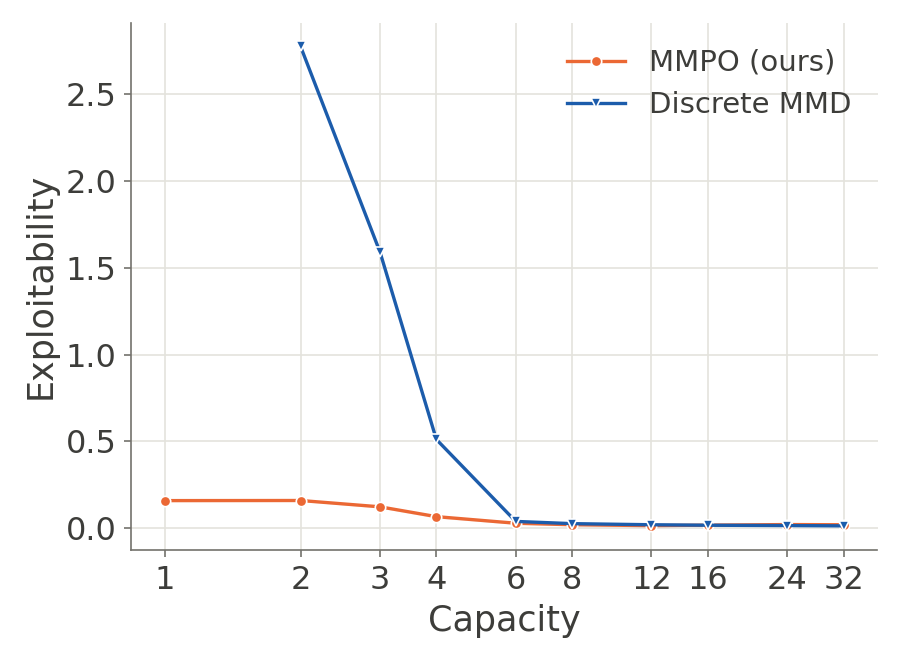}
        \caption{Circle game}
        \label{fig:one_shot_circle_capacity}
    \end{subfigure}
    \hfill % Adds horizontal space between subfigures
    \begin{subfigure}[b]{0.24\textwidth}
        \centering
    \includegraphics[width=0.99\linewidth]{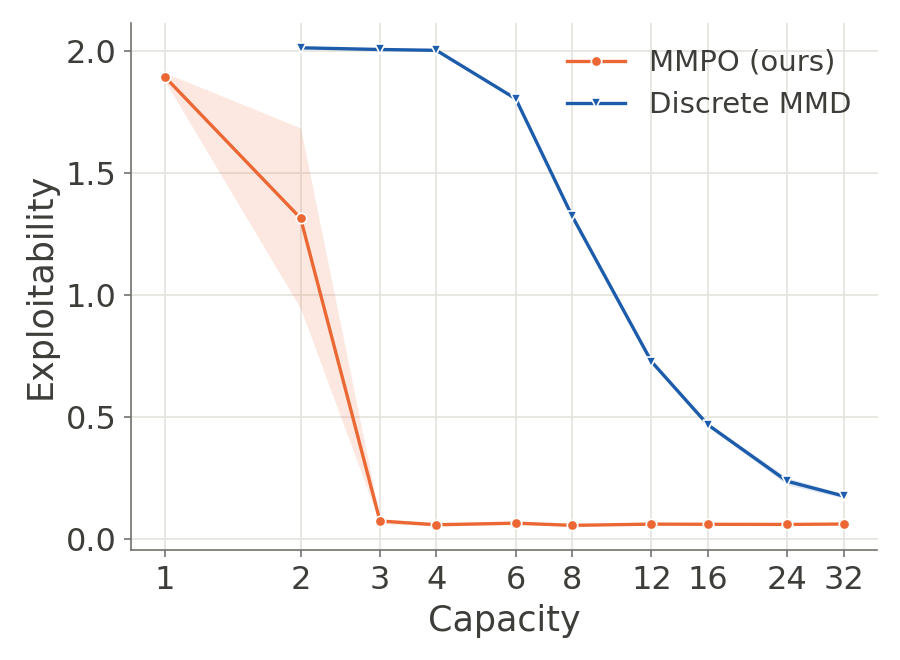}
        \caption{Two-point game}
        \label{fig:one_shot_two_capacity}
    \end{subfigure}  
    \caption{Final exploitability of \AlgorithmShort{} and MMD on discretized action space based on the amount of Gaussians or discrete components used.}
    \label{fig:one_shot_capacity} 
\end{figure*}

\begin{figure*}[!h]
    \centering    
    \begin{subfigure}[b]{0.24\textwidth}
        \centering
    \includegraphics[width=0.99\linewidth]{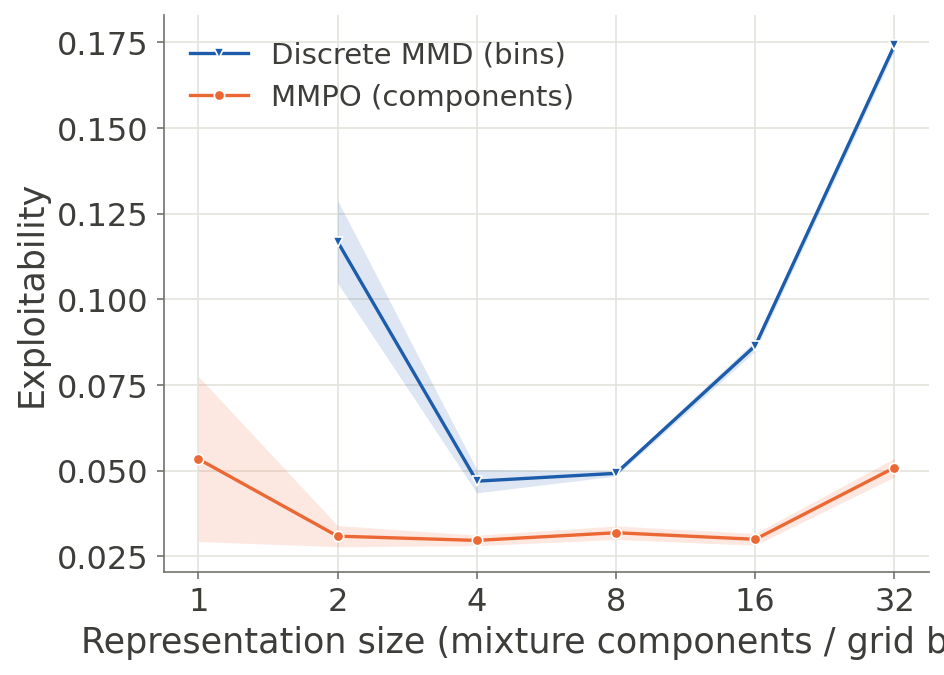}
        \caption{Kuhn Poker}
        \label{fig:seq_kuhn_capacity}
    \end{subfigure}
    \hfill % Adds horizontal space between subfigures
    \begin{subfigure}[b]{0.24\textwidth}
        \centering
    \includegraphics[width=0.99\linewidth]{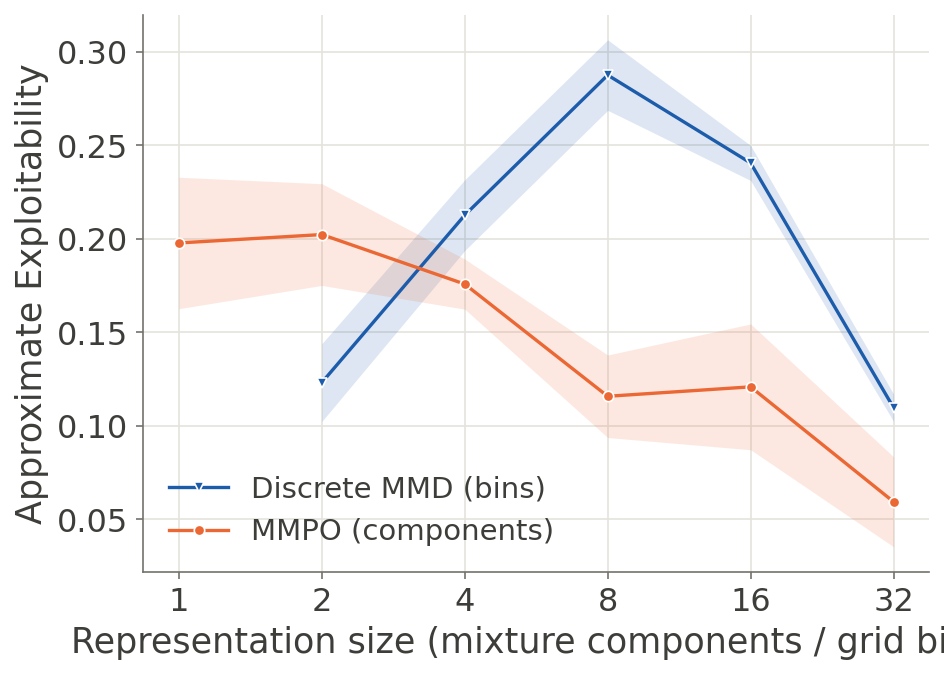}
        \caption{Leduc Poker}
        \label{fig:seq_leduc_capacity}
    \end{subfigure}  
    \begin{subfigure}[b]{0.24\textwidth}
        \centering
    \includegraphics[width=0.99\linewidth]{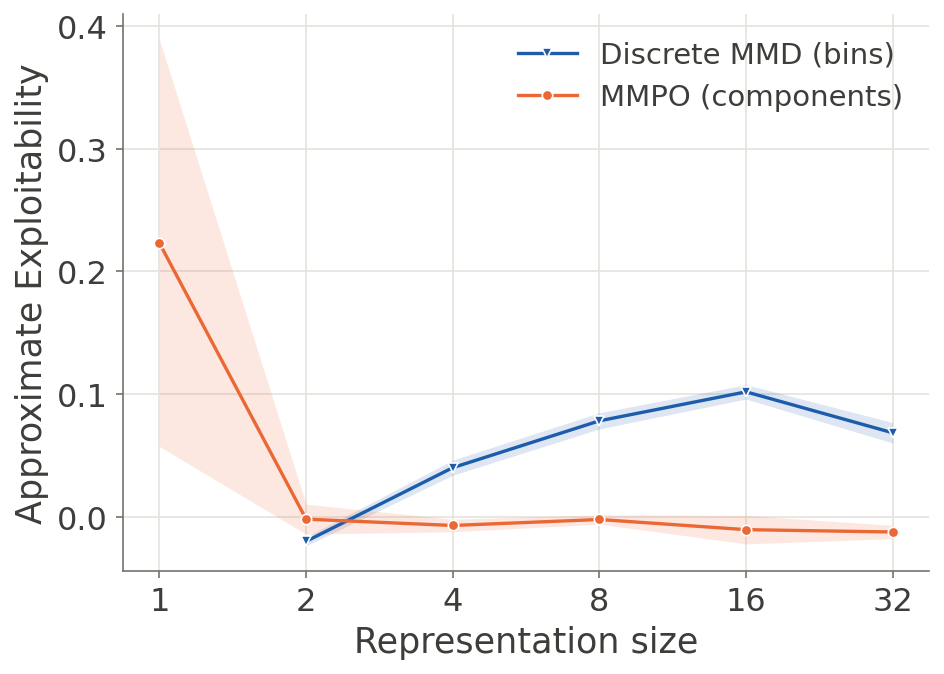}
        \caption{Blotto 3}
        \label{fig:seq_blotto3_capacity}
    \end{subfigure}
    \hfill % Adds horizontal space between subfigures
    \begin{subfigure}[b]{0.24\textwidth}
        \centering
    \includegraphics[width=0.99\linewidth]{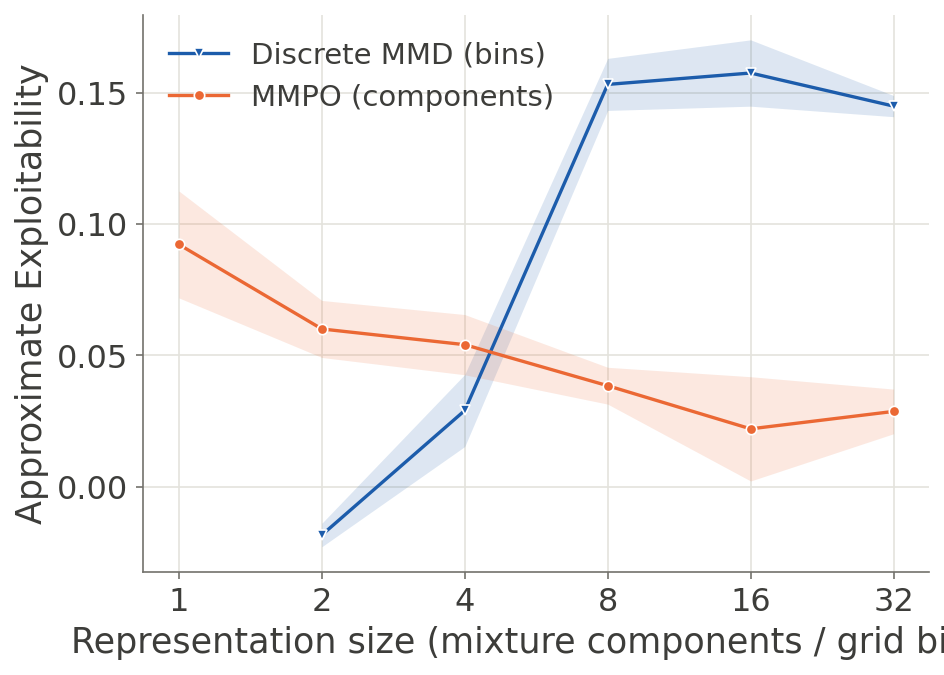}
        \caption{Blotto 5}
        \label{fig:seq_blotto5_capacity}
    \end{subfigure}
    \caption{Final exploitability of \AlgorithmShort{} and MMD on discretized action space based on the amount of Gaussians or discrete components used.}
    \label{fig:seq_capacity} 
\end{figure*}  
\section{Heads-up no limit Texas hold'em details}
\label{app:hunl}
In this section, we detail several changes to the MMPO we made specifically for the Texas hold 'em experiment and include additional results from that experiment.

\subsection{Experimental detail}
\paragraph{Action abstraction} We used these discrete actions: Fold, Call/Check, and All-in. Then we treat the betting as a single continuous action. However, we do not use the raw bet-size; the bet action is a fraction of the pot, which was used in prior Poker AIs \cite{brown2018libratus,moravcik2017deepstack}. We use 5 Gaussian components in this continuous action, and we do not initialize them randomly; instead, we initialize them so that their means cover the interval $[\frac{1}{3}, 3]$ geometrically. If the mean of the action is beyond the legal betting range, we clip it. Besides that, we are not predicting the mean $\Mean$ directly. We predict $\log(\Mean)$, which we then exponentiate. In the game, bets are discrete, so after we sample the real value, we round the bet to the nearest legal amount. During training, we use this rounded bet rather than the sampled one.
\paragraph{Action Clipping} We apply a trick similar to Neural replicator dynamics that stops learning certain discrete actions when their logits are too far apart \citep{hennes2020neurd}. Specifically, when an action logit is $\beta$ smaller than the highest logit, only the gradients that will increase the probability of such actions are used, and vice versa when the logit is higher than the smallest logit. We applied this trick because in one of our runs, the logits were getting too small, which made them unrecoverable when the action turned out to be valuable. We also use sigmoid for the $\log(\Variance)$ rescaled to values between [-3, -0.7]
\paragraph{Exploration} Early training runs showed that later checkpoints improved their performance against older checkpoints, but their performance against Slumbot was degrading substantially. We found out that it was because some bets were not being explored, and the strategy for those bets was degrading. We have made several changes to avoid this problem. The neural replicator dynamics' maximal-difference logits are one such change, and the second change is that, with 1\%, we sample the random categorical action. We use an importance-sampling correction to ensure that all estimates are unbiased. In imperfect-information games, it is not enough to apply an importance-sampling correction to the future values in V-trace; it is also necessary to correct past decisions. To avoid the correction becoming too high or too low, we clip it to the range [0.1, 10]. We do not use any exploration for the Gaussian head, besides the samples from the Gaussian itself.
\paragraph{Hyperparameter specifics} We use two training schedules. First is cosine decay of the learning rate from $10^{-3}$ with $T=1.5\cdot10^6$ and $\alpha = 0.06$. We also used an exponential decay of entropy from $0.03$ to $0.003$ over $2\cdot10^{6}$ iterations. We also did not use the same magnet and entropy regularization coefficients for the Gaussian and categorical heads. In the categorical head, we use the KL divergence coefficient $\RegularizationStrength_\CategoricalStrategy=0.2$ and the entropy coefficient $\EntropyCoefficient$ with exponential decay. For Gaussian, we used the KL divergence coefficient $\RegularizationStrength_\Gaussian = 0.02$, and scaled the entropy by 0.5. We decreased these coefficients for the Gaussian head because, in our early experiments, the Gaussian means were not learning due to the magnet loss outweighing the advantage gradient. All other hyperparameters are in \Cref{tab:hunl_hyper}
\begin{table}[]
    \centering
    \begin{tabular}{l|c} 
         PPO epochs & 1 \\
         Categorical magnet weight $\RegularizationStrength_\CategoricalStrategy$ & 0.2 \\
         Gaussian magnet weight $\RegularizationStrength_\Gaussian$ & 0.02 \\
         Magnet update & 20000\\
         Minimal std $\Variance_{\text{min}}$ & 0.05\\
         Components $K$ & 5\\
         Exploration $\epsilon$ & 0.01 \\
         Value loss weight & 1.0\\
         GAE lambda $\lambda$ & 0.95 \\
         V-trace clipping $\overline{\rho}$ & 1 \\
         V-trace clipping $\overline{c}$ & 1   \\
         Hidden sizes & (1024,1024,1024,1024) \\
         Activation & GELU \\
         Normalization & Layer\\
         Optimizer & Adamw  \\
         Weight decay & 0.002 \\ 
         Max grad norm & 1.0 \\
         Target update $\tau$  & $4 \cdot 10^{-4}$ \\
         Batch size & 4096\\
         Max game length & 40 
    \end{tabular}
    \caption{Hyperparameters used in training of heads-up no limit Texas hold'em}
    \label{tab:hunl_hyper}
\end{table}

\subsection{Additional results}

In \Cref{fig:hunl_small,fig:hunl_big} we show the initial bets of small blind and big blind after small blind checked, respectively. The probabilities are averages across all possible suits for that combination. The values in the main diagonal correspond to the player holding a pair; the lower-left triangle corresponds to the player holding 2 offsuit cards; and the upper-right triangle corresponds to the player holding cards in the same suit. Interestingly, these two tables differ from the tables provided by the authors of the GTO wizard AI \citep{provost2026gto}. For the Small blind bet, the GTO wizard's bot never folds with suited cards, and with off-suited cards, it folds much more aggressively. And for the Big blind bet, MMPO never folds, whereas GTO wizard folds quite often, especially when holding 2 and 3 offsuit.

In \Cref{fig:hunl_additional}, we show the average win rate against older checkpoints and weighted outcomes against Slumbot for checkpoints produced during training. Even if the Slumbot performance does not seem to improve substantially in later iterations, the network itself is still improving against older checkpoints. This highlights the strategy non-transitivity in imperfect information games. The improvement in gains for later checkpoints seems to be mainly from improving how much the player loses at showdown. On the other hand, rewards from folding remain roughly the same in later stages.

In \Cref{fig:hunl_bets,fig:hunl_pots}, we show all the bets as fractions of the pot and the pot sizes from self-play of the final checkpoint, which shows that, unlike abstraction bots, the \AlgorithmShort{} does not have blind spots on some decisions, which would be caused by static discretization.

\begin{figure}
    \centering
    \includegraphics[width=0.95\linewidth]{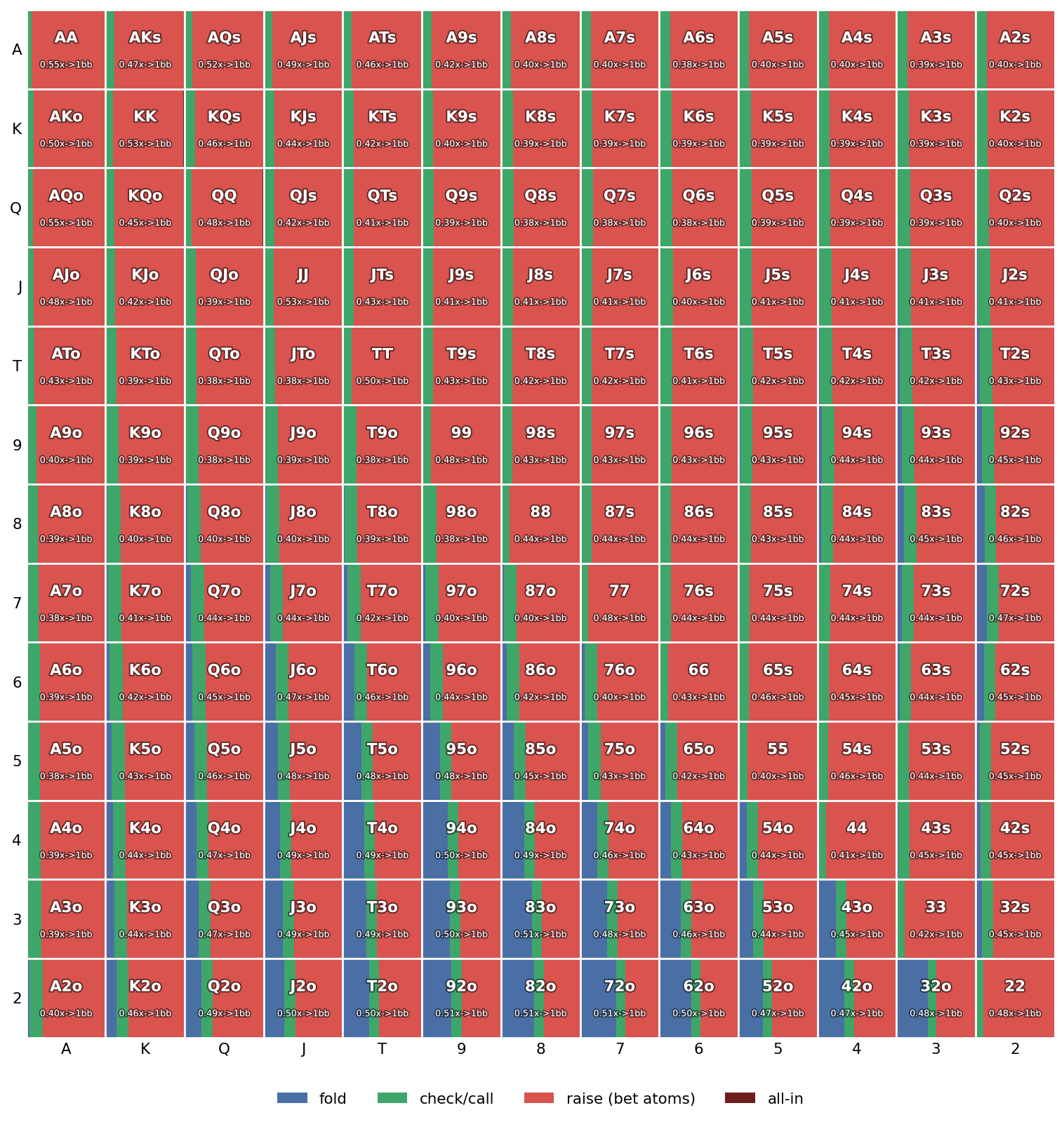}
    \caption{Small blind initial bet with different cards averaged over all suits.}
    \label{fig:hunl_small}
\end{figure}

\begin{figure}
    \centering
    \includegraphics[width=0.95\linewidth]{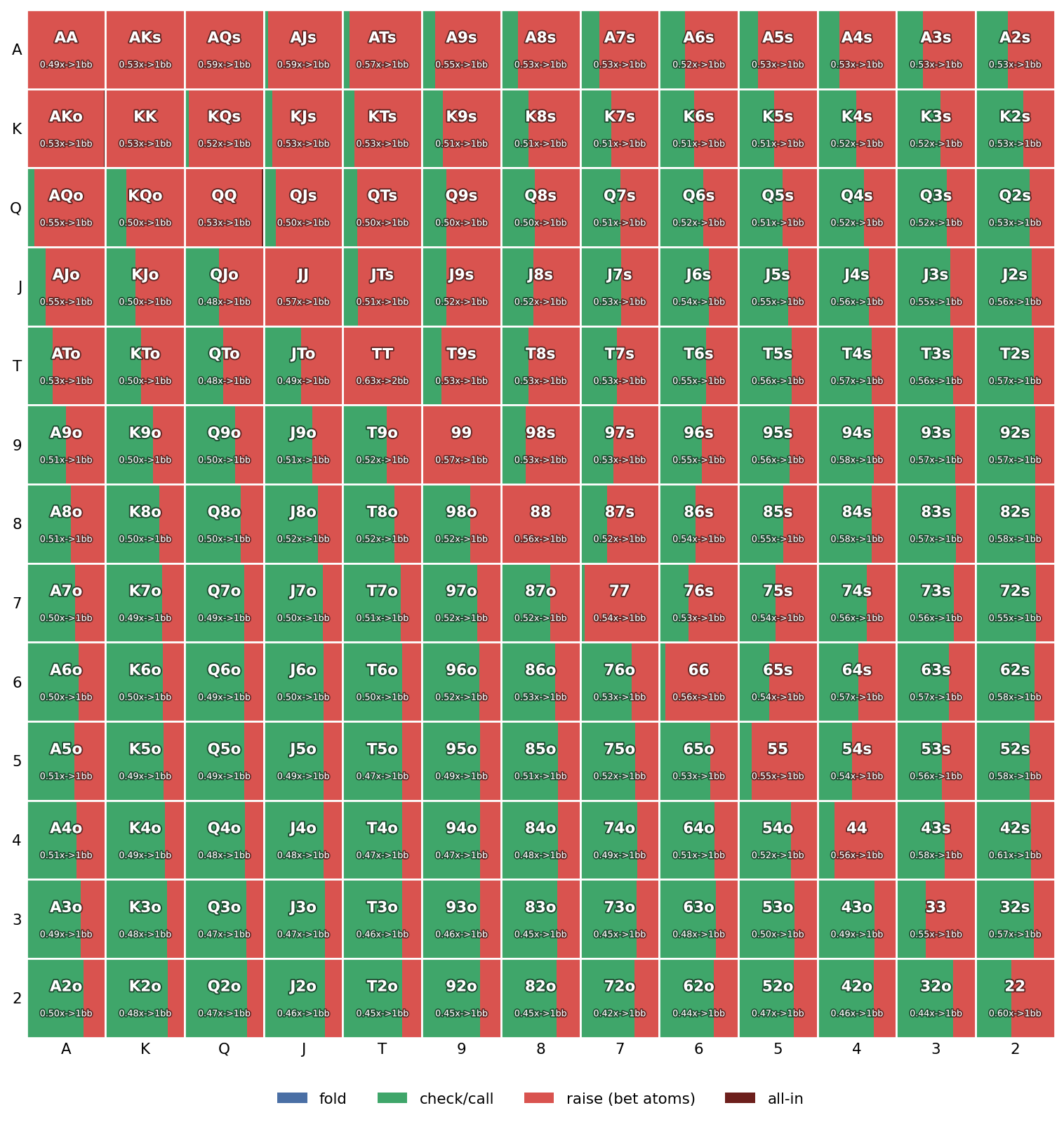}
    \caption{Big blind initial bet after check from the small blind with different cards averaged over all suits.}
    \label{fig:hunl_big}
\end{figure}

\begin{figure*}[!h]
    \centering
    \begin{subfigure}[b]{0.49\textwidth}
        \centering
    \includegraphics[width=0.99\linewidth]{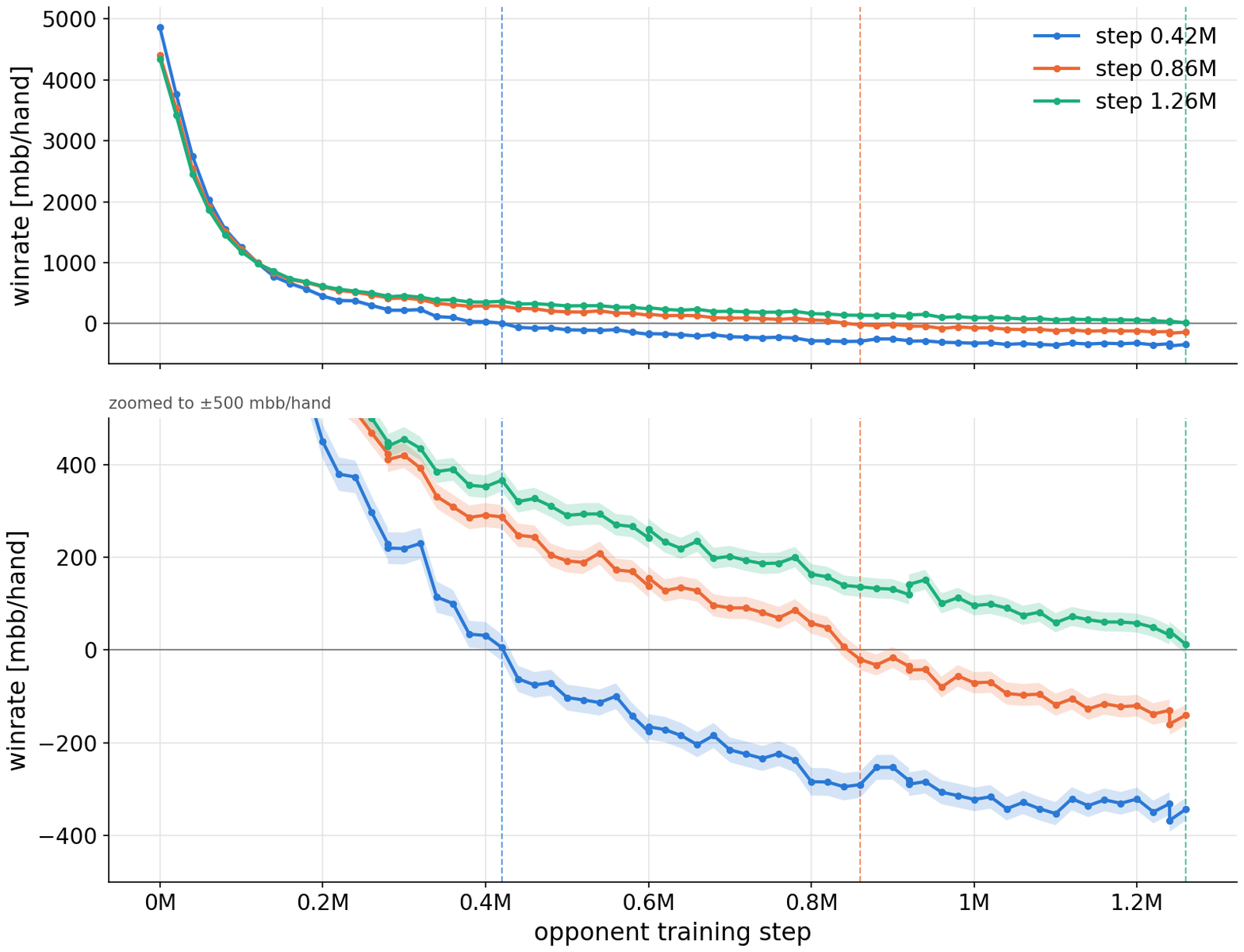}
        % \caption{Average hand length during the training}
        \label{fig:hunl_net_vs_net}
    \end{subfigure}
    \hfill % Adds horizontal space between subfigures
    \begin{subfigure}[b]{0.49\textwidth}
        \centering
    \includegraphics[width=0.99\linewidth]{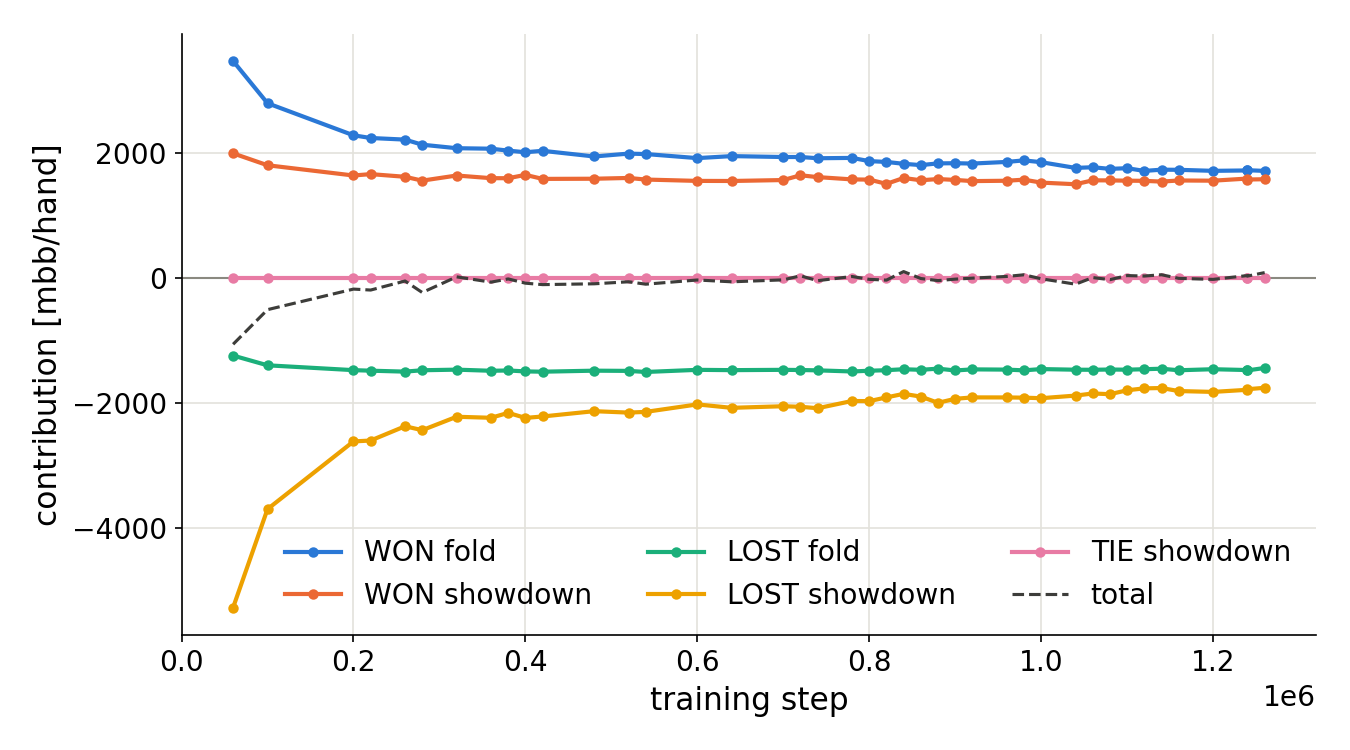}
        % \caption{Estimated gradient}
        \label{fig:hunl_weighted_outcome}
    \end{subfigure} 
    \caption{(left) Average win-rate of different checkpoints against all other checkpoints. (right) Weighted outcomes based on different results of the hand against Slumbot.}
    \label{fig:hunl_additional} 
\end{figure*}  
 
\begin{figure}
    \centering
    \includegraphics[width=0.95\linewidth]{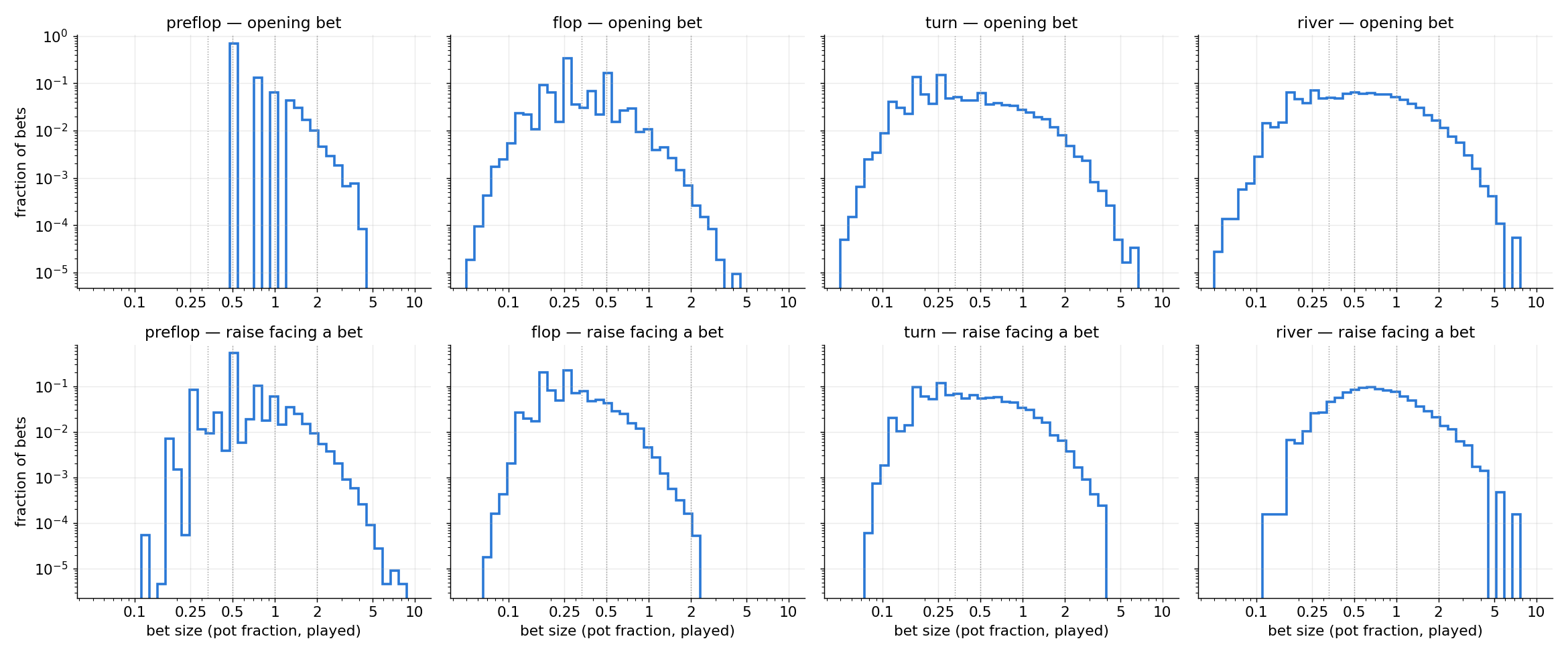}
    \caption{All bets made by the final trained checkpoint in self-play over million hands}
    \label{fig:hunl_bets}
\end{figure}

\begin{figure}
    \centering
    \includegraphics[width=0.95\linewidth]{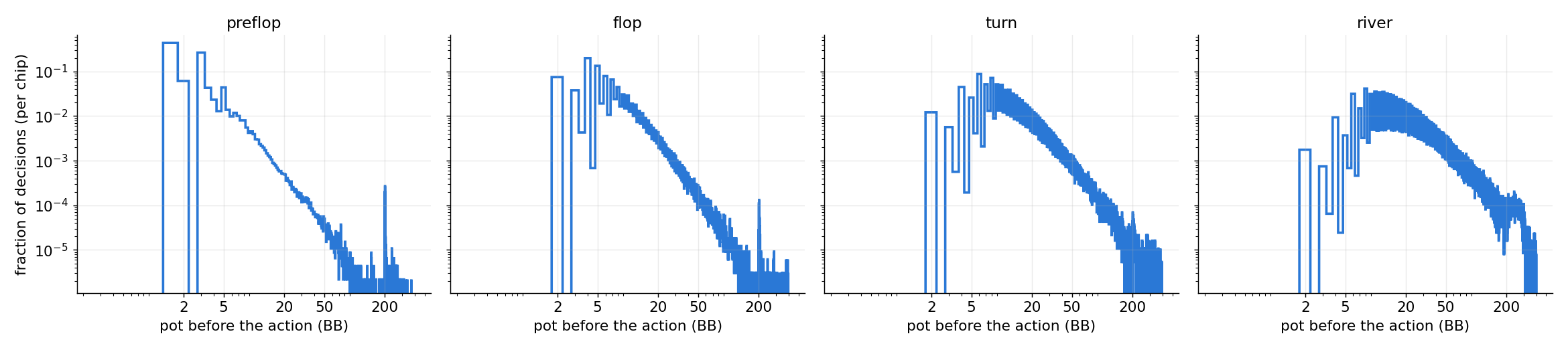}
    \caption{All pots, which were encountered by the final trained checkpoint in self-play over million hands}
    \label{fig:hunl_pots}
\end{figure}

\end{document}